\documentclass[journal]{IEEEtran}

\usepackage{amsmath, amssymb, amsfonts, amsbsy, amsthm}
\usepackage{mathtools, mathrsfs}
\usepackage{bbm}
\usepackage{xfrac}

\usepackage[utf8]{inputenc}
\usepackage[english]{babel}
\usepackage[T1]{fontenc}
\usepackage{stmaryrd}
\usepackage{csquotes}

\usepackage{subfigure}
\usepackage{float}
\usepackage{graphicx}
\usepackage{booktabs}
\usepackage{multirow, multicol}
\usepackage{color}
\usepackage{url}
\usepackage{hyperref}
\usepackage{cite}

\usepackage[framed,numbered]{matlab-prettifier}
\usepackage[linesnumbered,ruled]{algorithm2e}

\newtheorem{theorem}{Theorem}[section]
\newtheorem{corollary}{Corollary}[theorem]
\newtheorem{lemma}[theorem]{Lemma}

\newtheorem{definition}[theorem]{Definition}

\newcommand{\st}{\mid} 

\newcommand{\R}{\mathbb{R}} 
\newcommand{\N}{\mathbb{N}} 

\definecolor{warningcol}{rgb}{.99,.1,.1}
\definecolor{todocol}{rgb}{.1,.1,.99}
\definecolor{sketchcol}{rgb}{.2,.7,.35}
\definecolor{outlinecol}{rgb}{.8,.4,.3}
\definecolor{blackcol}{rgb}{.0,.0,.0}

\newcommand{\X}{\mathcal{X}} 
\newcommand{\Noise}{\Omega} 
\newcommand{\noise}{\omega} 
\newcommand{\BdryC}{\textup{\textbf{Boundary Condition}}} 
\newcommand{\InvC}{\textup{\textbf{Invariance Condition}}} 
\newcommand{\RFIS}{\mathcal{R}} 
\newcommand{\mRPI}{\RFIS_s} 
\newcommand{\MRPI}{\RFIS_m} 

\newcommand{\Conv}{\mathrm{Conv}} 
\renewcommand{\P}{\mathcal{P}} 
\newcommand{\eps}{\varepsilon} 
\newcommand{\bdry}{\partial}   
\newcommand{\interior}{\mathrm{int}} 

\newcommand{\Si}{\Delta} 
\newcommand{\Normal}[1]{N_{#1}} 

\newcommand{\V}[1]{\mathrm{Vert}(#1)} 
\newcommand{\SiCo}{\mathcal{K}} 
\newcommand{\SiCoII}{\mathcal{L}} 
\newcommand{\tri}{\mathbb{T}} 
\renewcommand{\star}{\overline{\mathrm{St}}} 
\renewcommand{\S}{\mathbb{S}} 
\newcommand{\Cent}{\S_C} 
\newcommand{\Bary}{\S_B} 
\newcommand{\idx}{\alpha} 
\newcommand{\idxcap}{\mathscr{A}} 
\newcommand{\Vol}{\mathrm{Vol}} 

\renewcommand{\ss}{\delta} 
\newcommand{\invss}{\delta^{-1}} 
\newcommand{\M}{M_{m}} 
\newcommand{\Lm}{l_{m}} 
 
\newcommand{\B}{\mathcal{B}} 
\newcommand{\TP}{\mathcal{T}} 
\newcommand{\A}{\mathcal{A}} 

\renewcommand{\H}{\phi} 
\newcommand{\CH}{g} 
\renewcommand{\o}{c} 
\newcommand{\decay}{\alpha} 
\newcommand{\NTP}{\mathcal{N}} 
\newcommand{\ray}{\mathcal{B}} 
\newcommand{\SiCoSet}{\mathcal{I}} 
\newcommand{\status}{\textsc{ind}} 

\newtheorem{innercustomgeneric}{\customgenericname}
\providecommand{\customgenericname}{}
\newcommand{\newcustomtheorem}[2]{%
  \newenvironment{#1}[1]
  {%
   \renewcommand\customgenericname{#2}%
   \renewcommand\theinnercustomgeneric{##1}%
   \innercustomgeneric
  }
  {\endinnercustomgeneric}
}

\newcustomtheorem{customthm}{Theorem}
\newcustomtheorem{customlem}{Lemma}
\newcustomtheorem{customdef}{Definition}

\SetCommentSty{mycommfont}

\begin{document}

\title{Computing Robust Forward Invariant Sets of Multidimensional Non-linear Systems via Geometric Deformation of Polytopes}

\author{Taha Ameen, Shayok Mukhopadhyay, and Nasser Qaddoumi
    \thanks{The authors are with the Department of Electrical Engineering, American University of Sharjah, Sharjah, United Arab Emirates. Emails: \{b00066555, smukhopadhyay, nqaddoumi\}@aus.edu}
}


\maketitle


\begin{abstract}
This paper develops and implements an algorithm to compute sequences of polytopic Robust Forward Invariant Sets (RFIS) that can parametrically vary in size between the maximal and minimal RFIS of a nonlinear dynamical system. This is done through a novel computational approach that geometrically deforms a polytope into an invariant set using a sequence of homeomorphishms, based on an invariance condition that only needs to be satisfied at a finite set of test points. For achieving this, a fast computational test is developed to establish if a given polytopic set is an RFIS. The geometric nature of the proposed approach makes it applicable for arbitrary Lipschitz continuous nonlinear systems in the presence of bounded additive disturbances. The versatility of the proposed approach is presented through simulation results on a variety of nonlinear dynamical systems in two and three dimensions, for which, sequences of invariant sets are computed.
\end{abstract}


\IEEEpeerreviewmaketitle

\section{Introduction} \label{S: Introduction}


\IEEEPARstart{A} Forward Invariant Set (FIS) for a dynamical system is a subset of the state space, from which state vector trajectories never escape as time runs forward. Invariant sets are fundamental objects in dynamical systems, and find applications in reachability analysis, as well as characterization of system stability and robustness. The latter is particularly useful in applications with strict performance and safety requirements, as the FIS provides bounds on system states. This is more valuable when such bounds can be guaranteed in the presence of disturbances through finding Robust Forward Invariant Sets (RFIS) for the system. Because any system trajectory that begins in an FIS stays in it, finding an invariant set that is disjoint from an unsafe set of states can ensure that system trajectories never reach an undesirable subset of the state space, if the initial conditions are picked to start within such a computed FIS. Further, an RFIS of the appropriate size can guarantee safety and robustness without imposing over-conservative or over-aggressive bounds on system trajectories. Thus motivated, this work focuses on computing RFISs in an $n$-dimensional setting. The two extremes, which are useful for performance analysis and controller synthesis~\cite{Rakovic_TAC_2005}, are the minimal RFIS, $\mRPI$ and the maximal RFIS $\MRPI$. 

Invariant sets also find application in Model Predictive Control (MPC), where they are used as target sets, and can be used to provide domains of attraction for MPC-based controllers~\cite{Bravo_Automatica_2005}. Of special interest are polytopic FIS, which guarantee stability without compromising any loss of performance~\cite{Rubin_ECC_2018, Cannon_Automatica_2003}. It is therefore no surprise that the characterization and computation of robust forward invariant sets for various systems is an active area of research, as seen in the following review of literature.


\par
Traditional approaches for computing estimates of the invariant set include Carleman Linearizations~\cite{Loparo_TAC_1978} and PDE-based approaches such as Zubov's method~\cite{Genesio_TAC_1985}. However, advances in computational processing have shifted the paradigm towards \emph{polytopic} invariant sets, which are non-conservative but lead to computationally intensive algorithms~\cite{Blanchini_Automatica_1999}. Subsequently, the vast majority of recent approaches involve optimization of objective functions that are derived from invariance conditions, over polytopic constraint sets. These invariance conditions are usually algebraic constraints that candidate sets must satisfy for being invariant, and have various characterizations depending on the linearity of the system dynamics, as well as whether they are continuous or discrete in time~\cite{Blanchini_Automatica_1999}.  

Invariant sets for linear systems have been extensively studied. For instance, theoretical considerations for their existence are outlined in~\cite{Castelan_TAC_1993}. Existing approaches for calculation of invariant polytopes for linear systems include constraint tightening~\cite{Ghaemi_CDC_2008}, finite-time Aumann integrals~\cite{Rakovic_CDC_2007}, MPC schemes~\cite{Rohal-Ilkiv_ARC_2004}, state prediction trees~\cite{Pluymers_ACC_2005} and semi-definite programs~\cite{Blanco_IJC_2010} to name a few. Despite this progress, invariant set computation for linear systems continues to be an active field of research, with recent works focusing on robust invariance~\cite{Seron_ANZCC_2019} and set invariance in the presence of a control action (control invariance)~\cite{Gupta_Automatica_2019, Liu_IJRNC_2019, Anevlavis_CDC_2019, Anevlavis_ICHSCC_2020, Yu_IJCAS_2018}. Among linear systems, the class of discrete time systems has received considerable attention, following the work in~\cite{Kolmanovsky_MPE_1998}. Approaches include geometric construction using zonotopic bounds~\cite{Stoican_IFAC_2011}, Linear Matrix Inequality (LMI)-based algorithms~\cite{Rubin_ECC_2018, Tahir_TAC_2015}, Minkowski partial sums-based approximations~\cite{ Ong_CDC_2005, Ong_Automatica_2006}, linear programming~\cite{Trodden_TAC_2016}, and control invariance using convex optimization~\cite{Rakovic_TAC_2010}. Some of these approaches have also been extended to piecewise affine systems, such as the works in~\cite{Vasak_ECC_2007, Rakovic_CDC_2004, Alamo_AFD_2007, Dan_ICARCV_2006}. Finally, computation of invariant sets for Linear Parameter-Varying (LPV) systems has also been studied, for instance in~\cite{Miani_CDC_2005, Pluymers_ACC_2005b, Klintberg_CSL_2020}. Recently, viability theoretic approaches~\cite{Aubin_Viability_2009} to modeling disturbances in systems have gained attraction. For example, the authors in~\cite{Fiacchini_SCL_2012} model discrete time systems with disturbances as convex difference inclusions and derive invariance conditions. Another example is~\cite{Fiacchini_ACC_2011}, where the authors use differential inclusions to derive control invariant polytopic sets. Similar approaches are used in~\cite{Kouramas_CDC_2005, Rakovic_Automatica_2005} to model additive state disturbances for computing the minimal RFIS of linear systems.

In contrast, the body of literature on nonlinear systems is relatively limited, where most of the work focuses on discrete-time systems or is restricted to certain classes of nonlinear systems. Examples include~\cite{Artstein_Automatica_2007}, where the authors study robust invariance in discrete-time nonlinear systems by lifting the feedback operation to the space of sets, and~\cite{Fiacchini_Automatica_2010, Fiacchini_CDC_2007}, where the authors compute control invariant sets for such systems using differences of convex functions. For continuous-time nonlinear systems, the focus has primarily been on polynomial dynamics. For example, the authors in~\cite{Iannelli_IFAC_2018} use sum of squares relaxations to solve an optimization problem that estimates the region of attraction (ROA) of polynomial systems. Similarly, the authors in~\cite{Henrion_TAC_2014, Korda_NOLCOS_2013} formulate this ROA estimation problem as an infinite-dimensional linear program. Yet another approach involves Lyapunov type methods such as the work in~\cite{Valmorbida_ACC_2014}, and the work in~\cite{Topcu_TAC_2010}, where parameter independent Lyapunov functions are used to characterize invariant sets of polynomial systems with bounded disturbances. Of particular relevance to our work is the work in~\cite{Sassi_Automatica_2012, Sassi_CDC_2014}, where a test for invariance is developed based on the sub-tangentiality condition, which compares the angle between the system dynamics and the tangent cone at the boundary of a candidate invariant set~\cite{Blanchini_Automatica_1999}. However, due to the computational infeasibility of testing for the condition at infinitely many points, the authors limit the scope of their work to dynamical systems with polynomial vector fields, for which the problem of testing the sub-tangentiality condition is overcome using optimization-based approaches. 

All these works use algebraic rather than geometric formulations. This is motivated by the algebraic nature of set invariance conditions, and the concise representation of polytopes as LMIs, which provide convenient formulations for optimization-based approaches. However, no general algebraic framework exists for computing invariant sets for arbitrary nonlinear systems. Subsequently, geometric and computational approaches have attracted some attention. Even so, the majority of computational approaches deal with linear and polynomial systems. Examples can be found in~\cite{Lazar_ACC_2006, Rakovic_IFAC_2008, Alessio_Automatica_2007}, where polytopic FIS are constructed between ellipsoidal sets, and~\cite{Schreier_TCS_1997}, where geometric methods are used to search for a convex FIS. Other approaches include~\cite{Blanco_ECC_2007}, where the authors determine symmetrical polytopic invariant sets with state, input and rate constraints for systems without disturbances, and~\cite{Schaich_CDC_2015}, where authors derive the maximal robust positively invariant set for linear systems with additive disturbances. Finally, a recent line of computational approaches use data-driven techniques, such as~\cite{Chakrabarty_CDC_2018}, where the authors sample state trajectories starting from a variety of initial conditions to infer invariant sets for discrete time linear systems, and~\cite{Chen_CDC_2018}, where the authors use machine learning methods for systems whose dynamics are not explicitly modeled with equations.

The most significant advantage of geometric approaches is their ability to generalize to wide classes of systems due to their non-reliance on the explicit algebraic representation of system dynamics. This is evident from some recent works such as~\cite{Mukhopadhyay_Thesis_2014, Mukhopadhyay_ACC_2014, Mukhopadhyay_AJC_2020}, where closed polytopes are constructed using path planning algorithms on radial graphs for arbitrary nonlinear systems in continuous time, but limited to systems in $\R^2$. Another illustration of this can be found in~\cite{Varnell_CDC_2016, Varnell_ACC_2018}, where simplex-based approaches are used for arbitrary nonlinear systems. Besides these limited works, few computation-oriented results are available for generic nonlinear systems~\cite{Fiacchini_ACC_2011}.


This work focuses on Lipschitz continuous nonlinear systems in continuous time, and computes sequences of polytopic robust forward invariant sets in the presence of bounded additive disturbances. We adopt a geometric approach, using concepts from viability theory~\cite{Aubin_Viability_2009} and computational topology~\cite{Edelsbrunner_Book_2010}. Doing so allows us to develop a framework that does not impose any restrictions on the system dynamics, besides the usual conditions for existence and uniqueness of solutions~\cite{Khalil_Nonlinear_2002}. Further, no additional assumptions on the nature of disturbances are made besides their being bounded and additive. Thus, our approach is valid for a large class of nonlinear systems. To the best of our knowledge, this is the first work that develops an implementable computer algorithm for computing an RFIS of continuous-time nonlinear systems in the presence of disturbances without any further restrictions on state vector dimensions and trajectories. Our strategy is to geometrically deform a polytope into an RFIS using a sequence of homeomorphisms that are guided by an invariance condition. In contrast to previous work, we propose an explicit set of test points and prove that the satisfaction of an invariance condition at these test points is sufficient to conclude that a polytopic set is an RFIS for the system. Thus, the contribution of our work is two-fold: 
\begin{itemize}
    \item We propose a test for invariance of a polytopic RFIS by explicitly identifying test points on the boundary of the polytope.
    \item We invoke this invariance test to develop an algorithm that is easily implementable, and can generate sequences of invariant sets that vary in size between the maximal and minimal invariant sets. for $n$-dimensional Lipschitz continuous nonlinear systems with bounded additive disturbances. 
\end{itemize} 

The rest of this paper is organized as follows. Section~\ref{S: Mathematical Preliminaries} reviews the necessary mathematical background, and Section~\ref{S: Problem Formulation} formulates the problem and outlines the approach. In Section~\ref{S: Invariance Test}, an algorithmic test for invariance is developed. It is used as a tool in Section~\ref{S: Algorithm} to develop a computational algorithm for generating sequences of RFISs. Section~\ref{S: Simulations} presents the results of implementing this algorithm on a variety of nonlinear systems, and Section~\ref{S: Conclusion} draws conclusions.

\section{Mathematical Background} \label{S: Mathematical Preliminaries}


This section reviews mathematical background used in the rest of the work. Section~\ref{SS: Computational_Topology} introduces set-valued maps and Hausdorff distance, which are useful in modeling nonlinear systems with disturbances. Next, it reviews polyhedral sets, simplices and simplicial complexes - which allow modeling of the boundary of a candidate polytopic RFIS. Finally, it presents vertex maps and their induced homeomorphisms on simplicial complexes, which allow structure-preserving deformations of polytopes. Section~\ref{SS: Set Invariance} formalizes the notion of set invariance.

\subsection{Computational Topology} \label{SS: Computational_Topology}


\begin{definition}[Set-Valued Map]
Let $X$ and $Y$ be topological spaces. $F(\cdot)$ is said to be a set-valued map from $X$ to $Y$ if $F(x) \subset Y \ \forall x \in X$.
\end{definition}

Set-valued maps are used in defining differential inclusions~\cite{Aubin_Inclusions_2012}, which are of the form
\begin{equation}
    \dot{x}(t) \in F(x(t)),
\end{equation}
where $F(\cdot)$ is a set-valued map. This work uses differential inclusions to model nonlinear system dynamics in the presence of disturbances, in Section~\ref{S: Problem Formulation}. Since such sets are fundamental objects for our purposes, it is convenient to review some relevant definitions. We begin with the Hausdorff metric, which is a measure of set-to-set distance.

\begin{definition}[Hausdorff Distance]
The Hausdorff distance between two sets $X$ and $Y$ is denoted $d_H(X,Y)$ and calculated as
\begin{equation} \label{Eq: Hausdorff_Distance}
        \! d_H(X,Y) = \max \! \left( \sup_{x \in X } \inf_{y \in Y} \| x\!-\!y \|, \sup_{y \in Y} \inf_{x \in X} \| x\!-\!y \| \!\right).
\end{equation}
\end{definition}

Next, we review basics about polyhedra, simplices, simplicial complexes and triangulation using simplices.

\begin{definition}[Polyhedron]
    A polyhedron $\P \subset \R^n$ is defined as the intersection of finitely many half-spaces. Thus,
    \begin{align}
        \P = \{x \in \R^n \st Ax \leq b\},
    \end{align}
    where $A \in \R^{m \times n}$, $b \in \R^m$ and $\leq$ denotes componentwise inequality.
\end{definition}

Bounded polyhedra are called polytopes. A polytope which is also a convex set is said to be a convex polytope. This definition is standard in literature, and convenient for convex optimization based approaches. A convex polytope can be equivalently defined as the convex hull of its vertices~\cite{Boyd_Book_2004}: 
\begin{align}
    \P = \Conv(\{v_1, \cdots, v_m\}),
\end{align}
where $\Conv(\cdot)$ represents the convex hull of the set of points $v_1,\cdots,v_m \in \R^n$. 

\begin{figure}[t]
    \centering
    \includegraphics[width = 0.49 \textwidth]{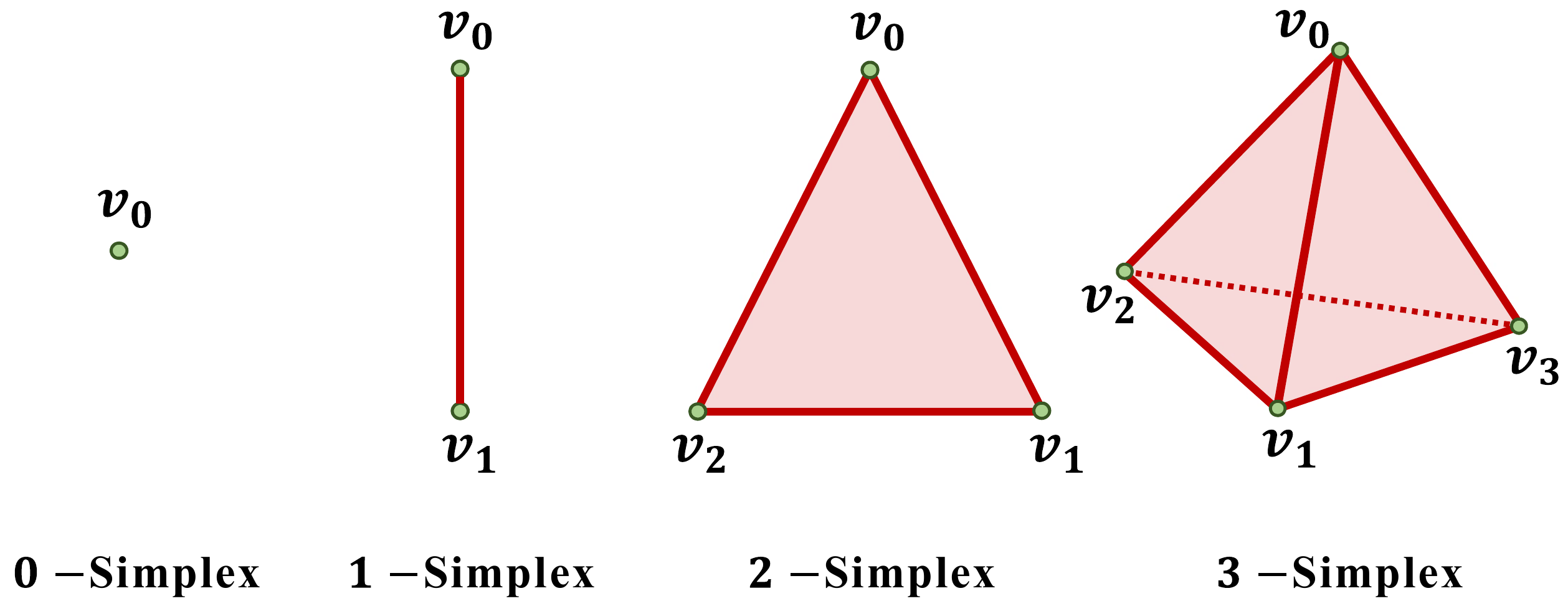}
    \caption{Examples of Simplices}
    \label{fig: Simplices_Examples}
\end{figure}

Simplices are generalization of triangles to higher dimensions, and are fundamental objects in algebraic topology, where they are used as building blocks for creating topological manifolds~\cite{Edelsbrunner_Book_2010}. 


\begin{definition}[Simplex]
    A simplex, $\Si \subset \R^n$ of dimension $k$ (also called a $k$-simplex), with $0 \leq k \leq n$ is the convex hull of a set of $(k+1)$ vertex vectors, $\V{\Si} = \{v_0, \cdots, v_k\} \subset \R^n$, such that the matrix
    \begin{align} \label{Eq: B_Delta}
    B_{\Si} = \big[ (v_1 - v_0), \cdots, (v_k-v_0) \big] \in \R^{n\times k}
    \end{align}
    has linearly independent columns.
\end{definition}

Since the vertices uniquely define a simplex, we identify $\Si$ with a matrix, $L_\Si \in \R^{(k+1) \times n}$, where row $i$ of $L_{\Si}$ is $v_{i+1}^T$:
\begin{align} \label{Eq: L_Delta}
    L_{\Si} = \begin{bmatrix}  v_0 , v_1, \cdots, v_{k} \end{bmatrix}^T
\end{align}
Each face of $\Si$ is defined as the convex hull of a proper subset of $\V{\Si}$. Note that an $n$-simplex, $\Si$, itself contains $k$-simplices where $k < n$. We refer to each of these \enquote{sub-simplices} as $k$-\emph{faces} of $\Si$. Figure~\ref{fig: Simplices_Examples} shows some examples.    

We will be interested in $(n-1)$-simplices embedded in $\R^n$, and will construct invariant sets with boundaries made up of these simplices. A simplex can have one of two possible orientations, as determined by the ordering of its vertices. For an $(n-1)$-simplex, $\Si \subset \R^n$, this orientation is equivalently characterized by the direction of the normal vector to $\Si$,
\begin{align} \label{Eq: Normal_Of_Simplex}
    (\Normal{\Si})_i = (-1)^{n+i} \det (B_{\setminus i}),
\end{align}
where $(\Normal{\Si})_i$ is the $i$-th component of the normal vector, and $B_{\setminus i} \in \R^{(n-1)\times (n-1)}$ is the matrix obtained by deleting the $i$-th row of $B_\Si$. Further, an $n$-simplex, $\Si \subset \R^n$ has a volume,
\begin{align}\label{Eq: Volume_Of_Simplex}
    \Vol(\Si) = \left| \frac{1}{n!} \det(B_{\Si}) \right|.
\end{align}


\begin{definition}[Simplicial Complex]
A geometric simplicial complex $\SiCo$ is a set of simplices such that
\begin{enumerate}
    \item If $\Si \in \SiCo$, then for any face $\Si'$ of $\Si$, we have $\Si ' \in \SiCo$.
    \item If $\Si_1$, $\Si_2 \in \SiCo$, then $\Si_1 \cap \Si_2$ is either the empty set or a face of both $\Si_1$ and $\Si_2$.
\end{enumerate}
\end{definition}
 \begin{figure}[t]
	\centering
	\subfigure[$\P$.]
	{
		\includegraphics[width=0.145 \textwidth]{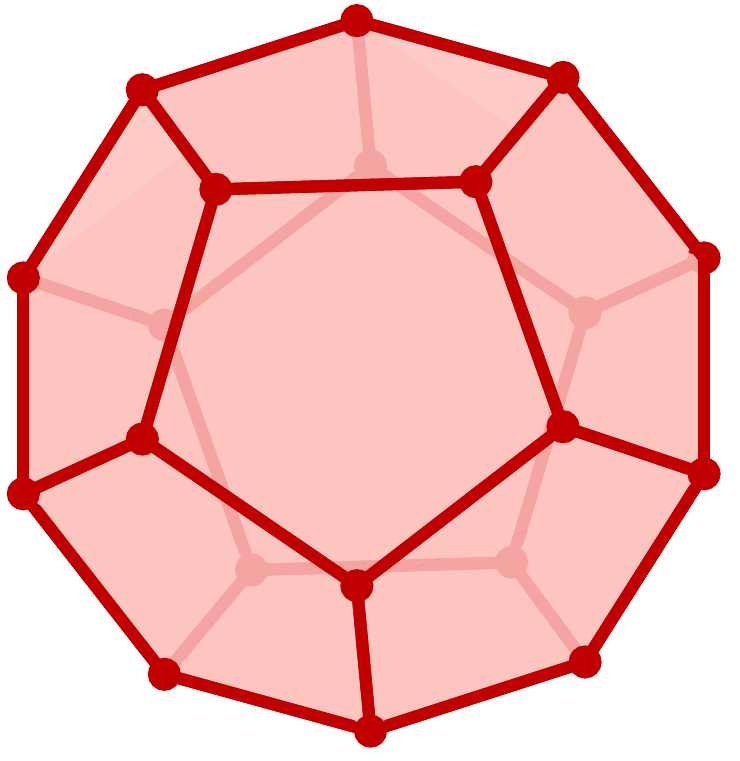}
		\label{fig: Dodecahedron}
	}
	\hfil
	\subfigure[$\tri(\bdry \P)$.]
	{
        \includegraphics[width=0.145 \textwidth]{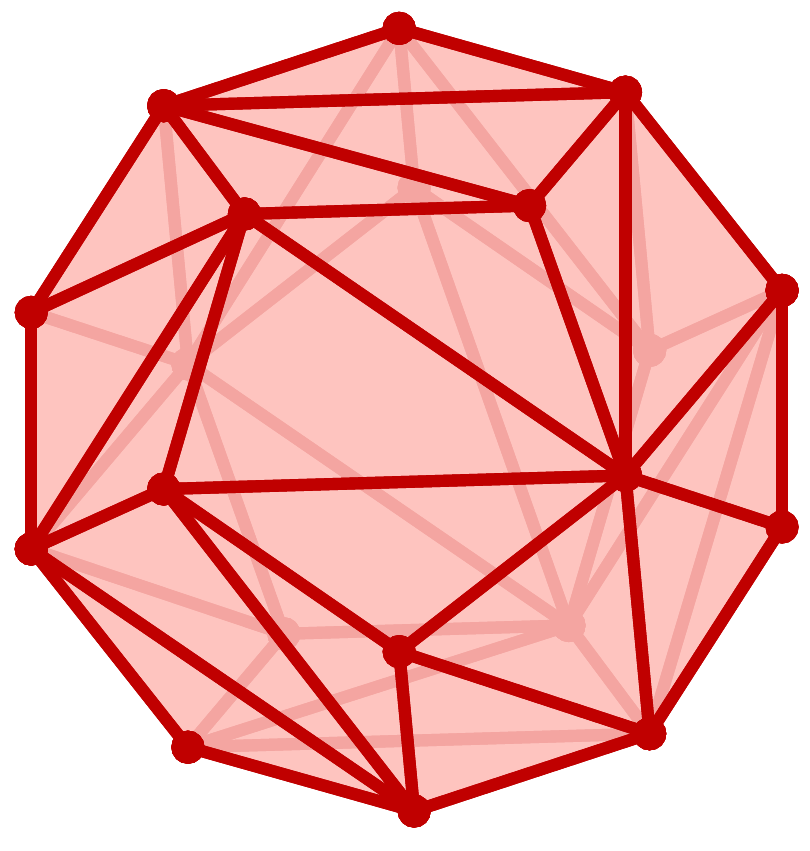}
		\label{fig: Dodecahedron_Triangulated}
	}
	\hfil
	\subfigure[$\star_{\SiCo}(v) \subset \tri(\bdry \P)$.]
	{
    \includegraphics[width=0.145 \textwidth]{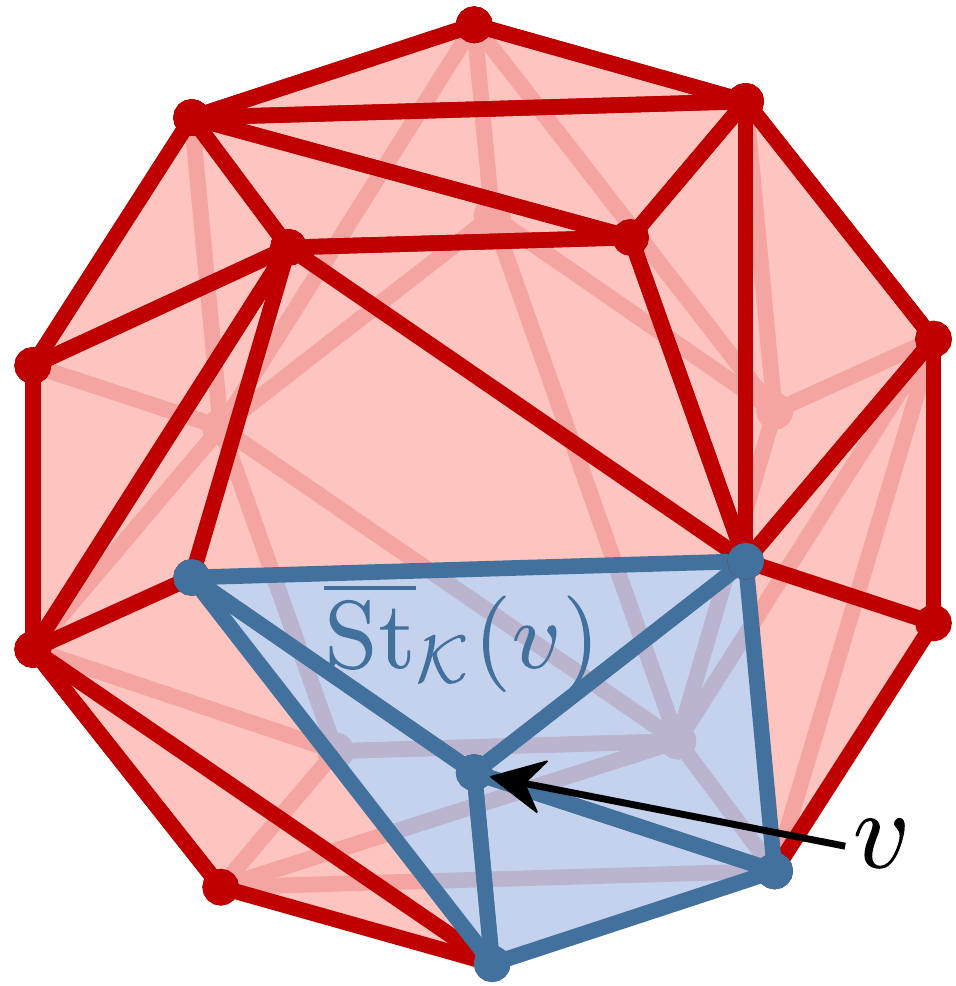}
	\label{fig: Star_of_Vertex}
	}
	\caption{Example of Triangulating the Boundary of a Convex Polytope.}
	\label{fig: Triangulation}
\end{figure}

A simplicial $k$-complex is a simplicial complex with the additional property that the largest dimension of any simplex in $\SiCo$ equals $k$. Finally, a \emph{homogeneous simplicial $k$-complex} is a simplicial complex where every simplex of dimension less than $k$ is a face of a $k$-simplex~\cite{Edelsbrunner_Book_2010}. In the remainder of this work, we will take a simplicial complex to mean a homogeneous simplicial $(n-1)$-complex in $\R^n$. These simplicial complexes are particularly relevant to our work because they triangulate convex polytopes, which will serve as a starting point for our algorithm. The triangulation of the boundary of a polytope $\P$, denoted $\tri(\bdry \P)$, is a set of simplices that partition $\bdry \P$, such that their union is the boundary of the polytope and the intersection of the \emph{interiors} of any two simplices in $\tri(\bdry \P)$ is empty. This is illustrated for a dodecahedron in Fig.~\ref{fig: Triangulation}. Figure~\ref{fig: Dodecahedron} shows the polytope $\P$, whereas Fig.~\ref{fig: Dodecahedron_Triangulated} shows the polytope with triangulated boundary. Each triangle on the boundary is a simplex, and the set of all simplices that form the boundary is a simplicial complex, $\SiCo = \tri(\bdry \P)$. The set of all $0$-faces of $\SiCo$ is then given by $\V{\SiCo} = \bigcup_{\Si \in \SiCo} \V{\Si}$.
Given a vertex $v \in \V{\SiCo}$, we denote the closed star of $v$ as $\star_{\SiCo}(v)$, which is the set of all simplices in $\SiCo$ that contain $v$ as a $0$-face. For a homogeneous simplicial $(n-1)$-complex embedded in $\R^n$, this is easily visualized, as shown in Fig.~\ref{fig: Star_of_Vertex}. Observe that $\star_{\SiCo}(v)$ is itself a simplicial complex.

We also remark that a convex polytope $\P$, being a convex compact subset of $\R^n$ forms a topological $n$-manifold with boundary~\cite{Bredon_Topology_2013}. Therefore, $\bdry \P$ forms an $(n-1)$-manifold~\cite{Lee_Manifolds_2010}. Following standard notation, we denote the geometric realization of $\SiCo$ by $|\SiCo|$, which is the topological space obtained by gluing all the simplices together, as determined by their faces. Since $\tri(\bdry\P)$ and $\SiCo$ are the same set, it follows that $\bdry\P$ and $|\SiCo|$ are homeomorphic topological manifolds~\cite{Edelsbrunner_Book_2010}. Recall that two topological manifolds $X$ and $Y$ are homeomorphic if $\exists \ \CH: X \rightarrow Y$ such that $\CH$ is invertible, and both $\CH$ and $\CH^{-1}$ are continuous.

\begin{figure}[t]
	\centering
	\subfigure[$\SiCo$.]
	{
		\includegraphics[width=0.227 \textwidth]{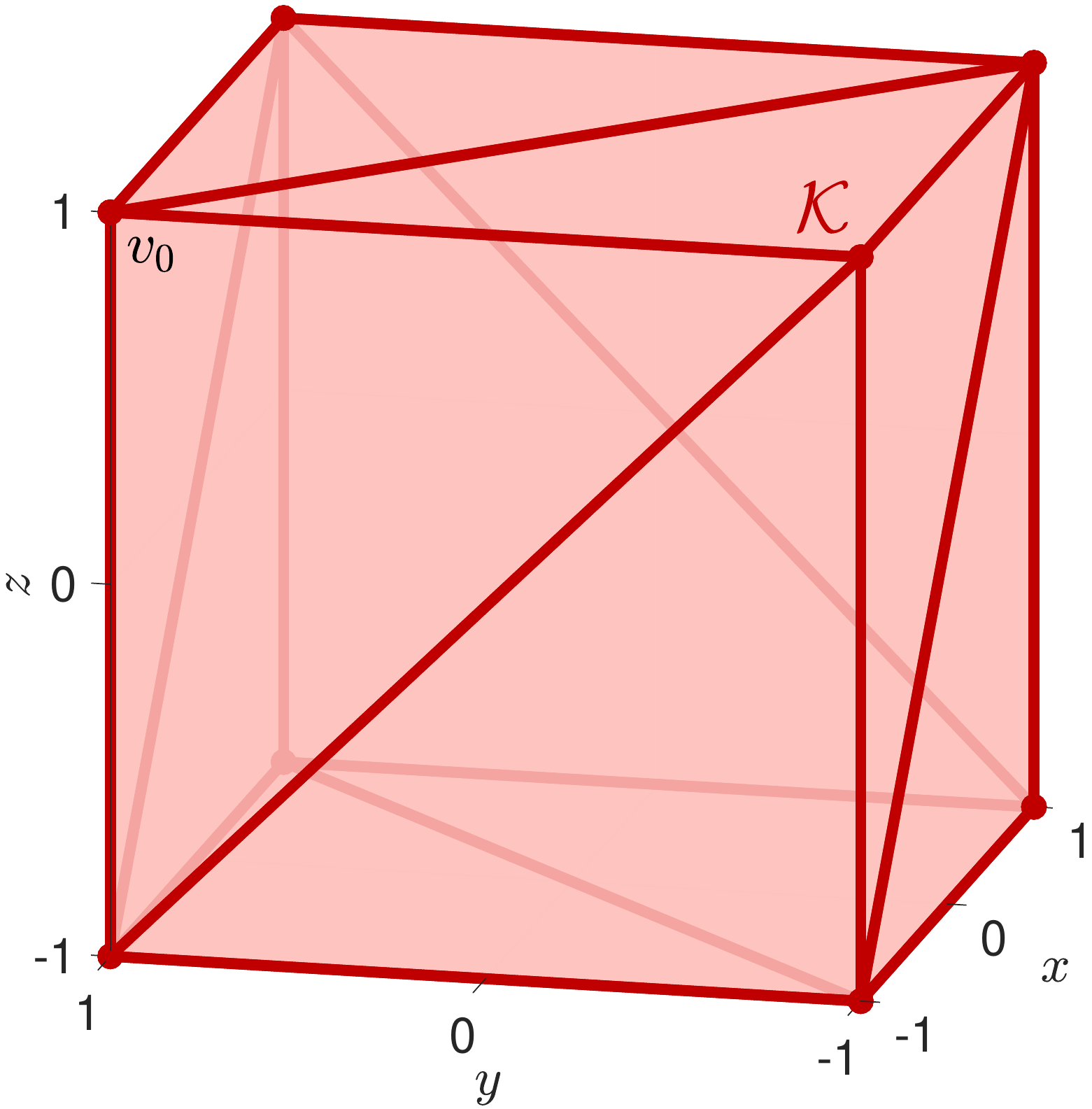}
		\label{fig: Cube_K}
	}
	\hfil
	\subfigure[$\SiCoII$.]
	{
        \includegraphics[width=0.227 \textwidth]{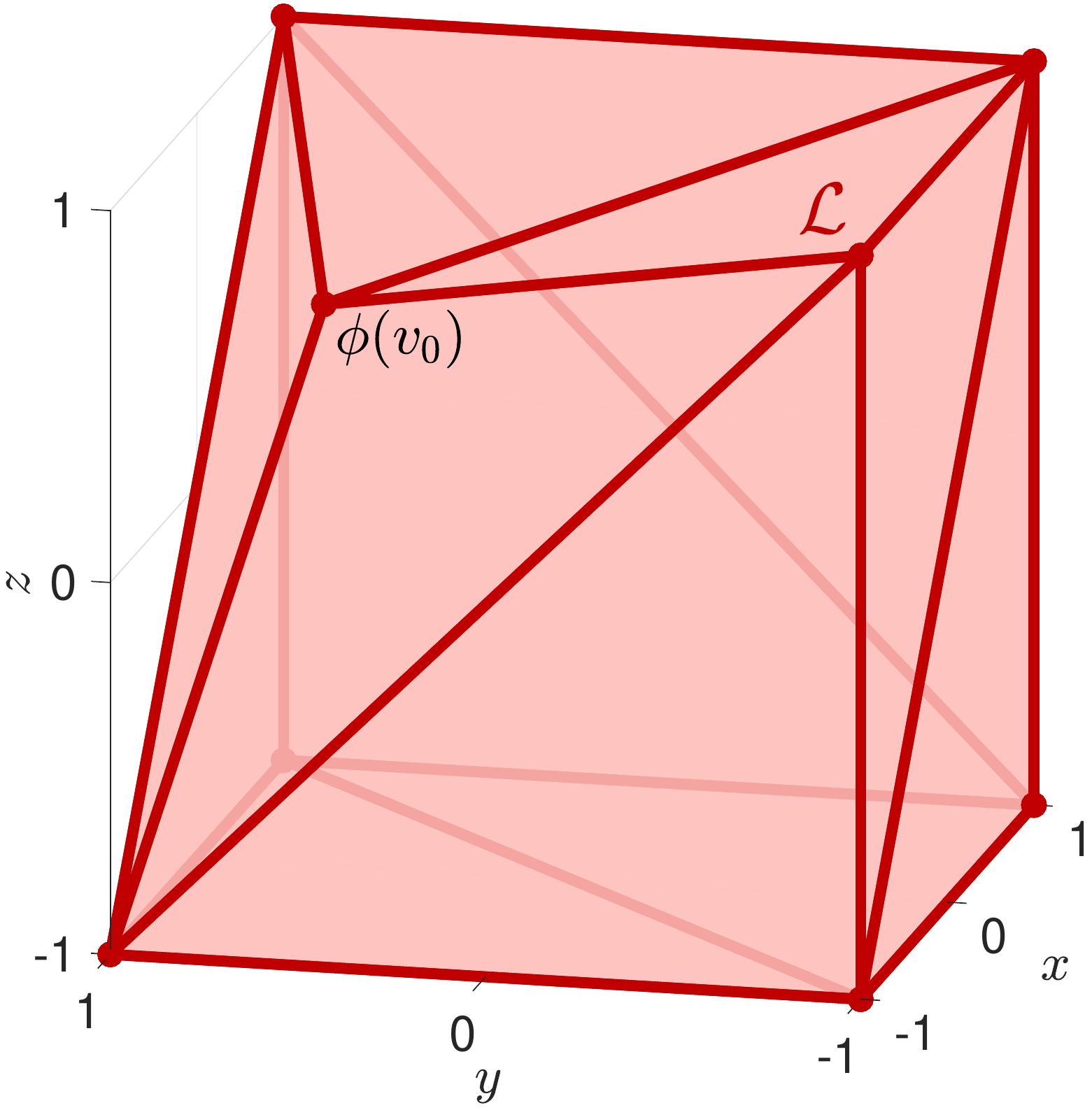}
		\label{fig: Cube_L}
	}
	\caption{Example of an induced homeomorphism from $\SiCo$ to $\SiCoII$.}
	\label{fig: Vertex_Map}
\end{figure}

In order to deform polytopes, it is convenient to study homeomorphisms between them. Homeomorphisms will be used in our work to deform polytopes into robust forward invariant sets. These homeomorphisms are induced from simplicial maps with the help of barycentric coordinates as explained in what follows. Simplicial maps are mappings between simplicial complexes that are the natural equivalent of continuous maps between topological spaces. A \emph{vertex map} between two simplicial complexes $\SiCo$ and $\SiCoII$ is a function $\H : \V{\SiCo} \rightarrow \V{\SiCoII}$ such that the vertices of every simplex in $\SiCo$ map to the vertices of a simplex in $\SiCoII$. $\H$ can be extended to a continuous map $\CH: |\SiCo| \rightarrow |\SiCoII|$, as
\begin{align} \label{Eq: Cont_Homeo}
    \CH (x) = \sum_{i = 0}^{N}b_i(x) \H (v_i),
\end{align}
 where $b_i(x)$ are the barycentric coordinates of $x$, defined in the next paragraph. The continuous map $\CH$ is called the \enquote{induced map} due to $\H$, and is unique for a given vertex map. Further, if $\H$ is bijective, and $\H^{-1}: \V{\SiCoII} \rightarrow \V{\SiCo}$ is also a vertex map, then $\CH$ is a homeomorphism between $|\SiCo|$ and $|\SiCoII|$~\cite{Edelsbrunner_Book_2010}. An example is illustrated in Fig.~\ref{fig: Vertex_Map}, where $\SiCo$ triangulates a cube with vertices $(\pm 1, \pm 1, \pm 1)$ as shown in Fig.~\ref{fig: Cube_K}. The vertex map $\phi: \V{\SiCo} \rightarrow \V{\SiCoII}$ maps $v_0 = (-1,1,1)$ to $\H(v_0) = (-0.7,0.5,0.7)$ and all other vertices to themselves. The resulting simplicial complex, $\SiCoII$ is shown in Fig.~\ref{fig: Cube_L}.

The barycentric coordinates of a point $x$ may be found as follows. Let $\SiCo$ be a simplicial complex with $\V{\SiCo} = \{v_0, v_1, \cdots, v_N\}$. Let the set of indices $\idxcap = \{\idx_0, \idx_1, \cdots, \idx_{n-1} \} \subset \{0,1,\cdots,N\}$, so that a given simplex, $\Si \in \SiCo$ has $\V{\Si} = \{v_{\idx_0}, \cdots, v_{\idx_{n-1}}\}$. Thus, any point $x \in \Si$ can be written as $x = \sum_{i \in \idxcap} \lambda_i v_{i}$, with $\sum_{i \in \idxcap}\lambda_i = 1$ and $\lambda_i \geq 0 \ \forall i$. The \emph{barycentric coordinates} of a point $x \in \Si$ are obtained by setting $b_i(x) = \lambda_i \ \forall i \in \idxcap$ and $b_i(x) = 0 \ \forall i \notin \idxcap$. Note that a point in a homogeneous simplicial $(n-1)$-complex can be identified by its barycentric coordinates, as $b_i(x)$ is unique even if the point lies in the intersection of multiple simplices of the complex.

Another concept from computational topology that we will use is simplex subdivisioning. A simplex subdivision of $\Si$, denoted $\S(\Si)$ is a simplicial complex obtained by introducing new vertices in the interior of the simplex, $\interior(\Si)$, and constructing a triangulation of $\Si$. We overload notation and use $\S(\SiCo)$ to represent the simplicial complex obtained as
\begin{align}
\S(\SiCo) = \bigcup_{\Si \in \SiCo} \S(\Si).
\end{align}
There are two key properties: First, the vertices of the simplices in $\S(\Si)$ can be ordered so that the normal direction for each $s \in \S(\Si)$ is the same as that of $\Si$. Second, the geometric realization of the subdivision is the same as that of the original complex, i.e. $|\S(\SiCo)| = |\SiCo|$~\cite{Edelsbrunner_Book_2010}. Thus, subdivisioning does not change the normal directions or the geometric realization.

There are several ways to introduce new vertices in the interior of the simplex for subdivisioning. Examples include centrodial and barycentric subdivisions, which are illustrated in Fig.~\ref{fig: Simplex_Subdivisions}. We will denote the resulting simplicial complexes with $\Cent(\cdot)$ and $\Bary(\cdot)$ respectively. Given a $k$-simplex, $\Si$, we can construct $\Cent(\Si)$ and $\Bary(\Si)$ knowing only $\V{\Si}$. In particular, $\Cent(\Si)$ subdivides a $k$-simplex into three $k$-simplices by introducing a new vertex at the centroid of $\Si$, (the coordinate of which can easily be calculated as $v_b = \frac{1}{k+1}\sum_{i=0}^{k}v_i$), and connecting each $v_b$ to $v_i \ \forall i$. Similarly, $\Bary(\Si)$ introduces new coordinates at the barycentre of all faces of $\Si$. The vertices of $\Si$ are connected to the new vertices iteratively (see~\cite{Edelsbrunner_Book_2010} for procedure). 

\begin{figure}[t]
	\centering
	\subfigure[$\Si$.]
	{
		\includegraphics[width=0.145 \textwidth]{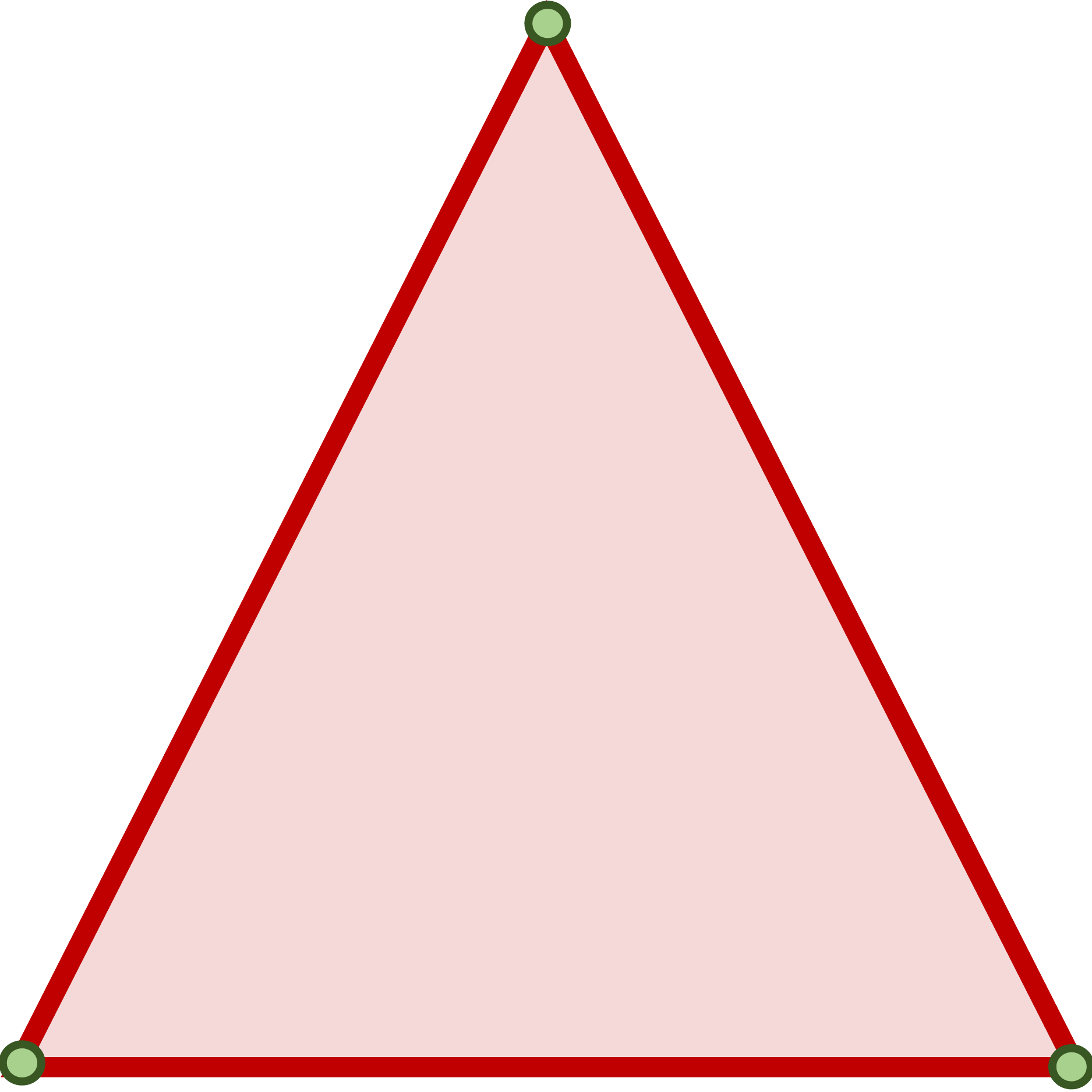}
		\label{fig: Simplex_Without_Subdivision}
	}
	\hfil
		\subfigure[$\Cent(\Si)$.]
	{
        \includegraphics[width=0.145 \textwidth]{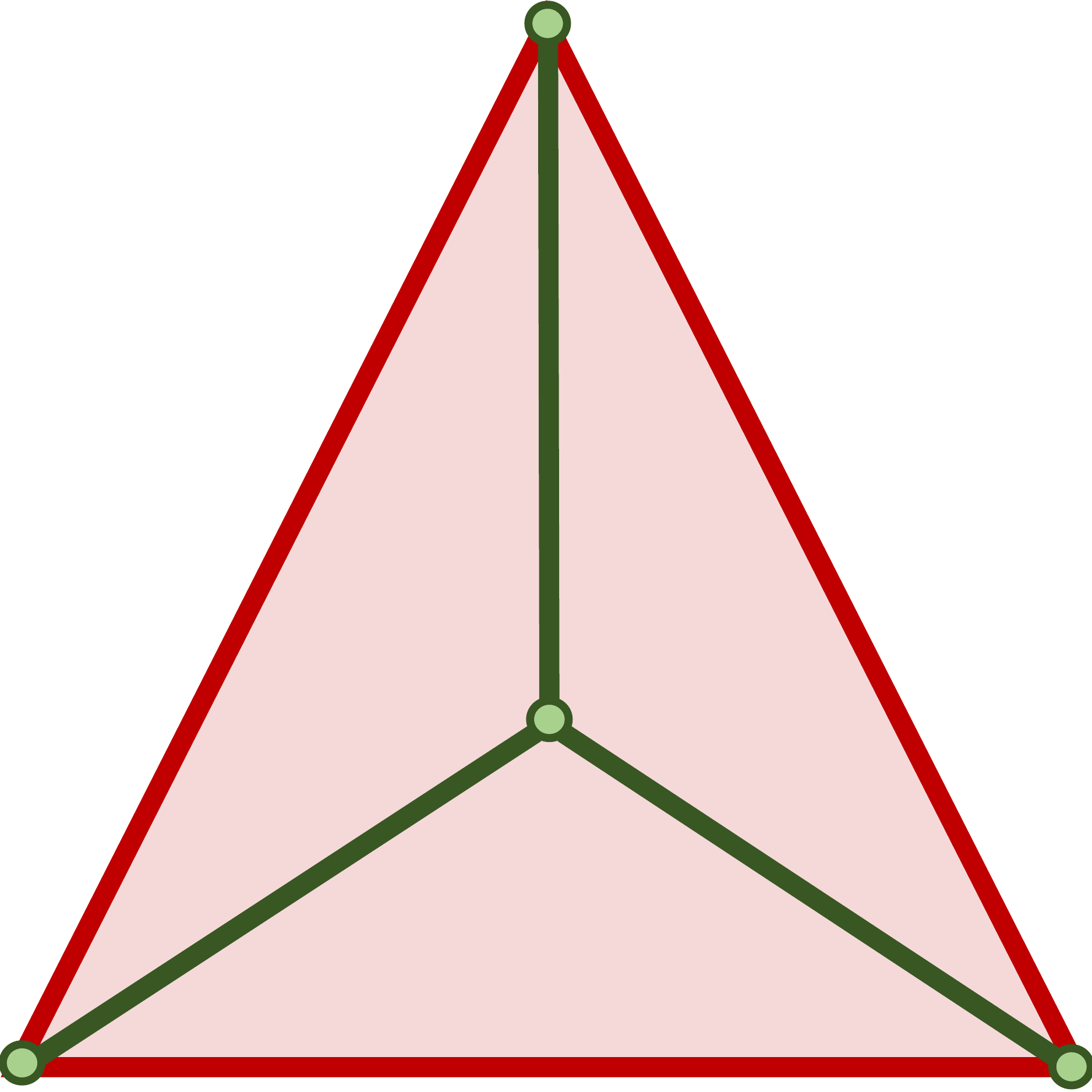}
		\label{fig: Centroidal_Subdivision}
	}
	\hfil
	\subfigure[$\Bary(\Si)$.]
	{
        \includegraphics[width=0.145 \textwidth]{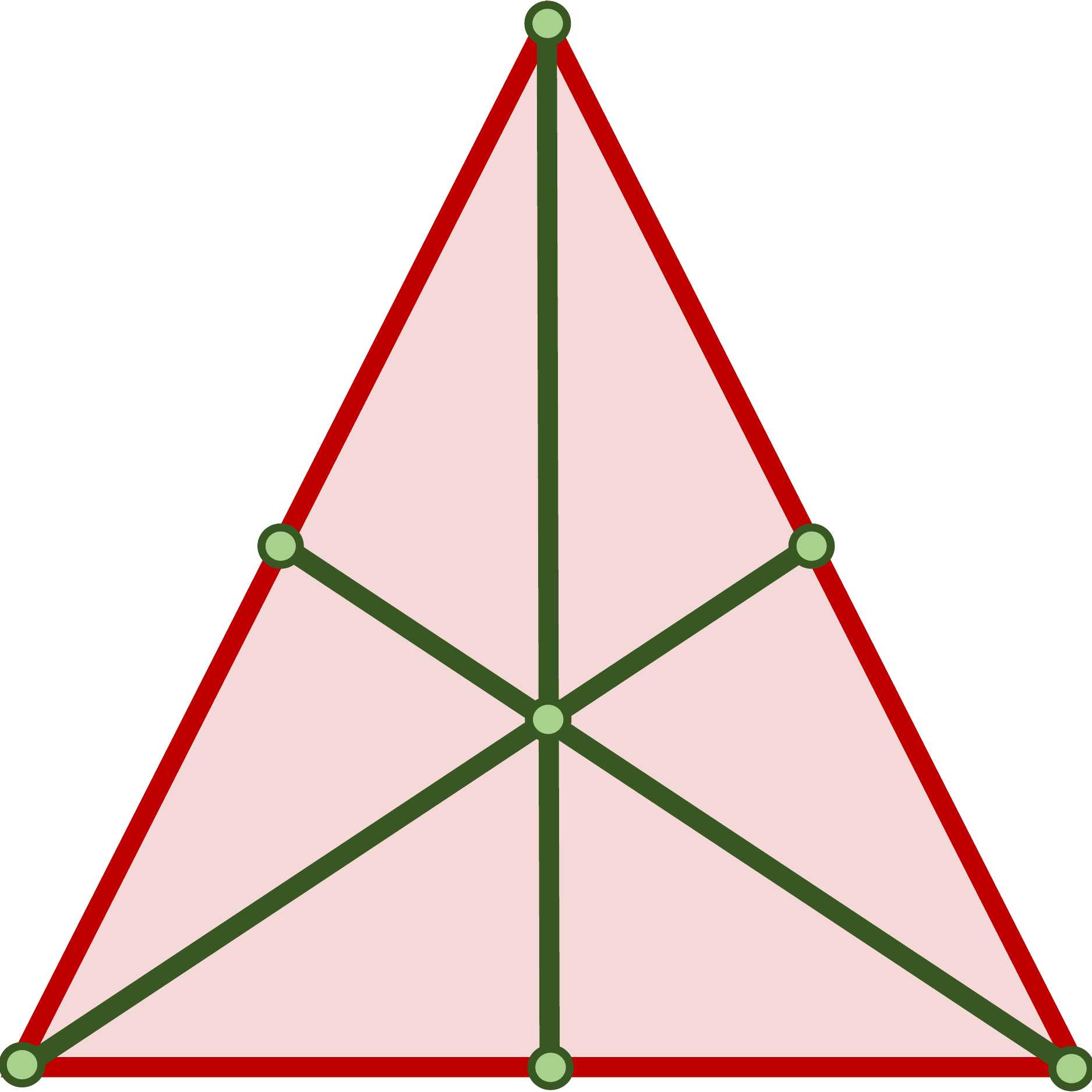}
		\label{fig: Barycentric_Subdivision}
	}\\
	\caption{Subdivisioning a Simplex.}
	\label{fig: Simplex_Subdivisions}
\end{figure}

\subsection{Robust Forward Set Invariance} \label{SS: Set Invariance}

The performance and robustness of a dynamical system can be characterized by forward invariance. A subset of the system's state space is said to be a forward invariant set (FIS) if no system trajectory that begins in the set leaves it at any future time. Further, the set is said to be a Robust FIS (RFIS) if it is an FIS in the presence of disturbances.

\begin{definition}[Robust Forward Invariant Set]
Let $S \subset \R^n$ represent the state space of a given dynamical system, $\dot{x}(t) = f(x(t),\noise(t))$, where $\noise(t): [0,\infty) \rightarrow \Noise$ is a bounded disturbance function that takes values in $\Noise$. Further, let $\X$ represent the set of solutions for this system. A set $\mathcal{R} \subset S$ is said to be an RFIS if $\forall x \in \X$, $x(t_0) \in \mathcal{R} \implies x(t) \in \mathcal{R} \ \forall t \geq t_0$.
\end{definition}

The \emph{minimal} RFIS is a set $\mRPI$ such that no proper subset of $\mRPI$ is an RFIS for the system. Conversely, the maximal RFIS, $\MRPI$ is not a proper subset of any RFIS. 


\section{Problem Formulation} \label{S: Problem Formulation}


Consider a non-linear system with bounded additive disturbances modeled by a differential inclusion $\dot{x}(t) \in F(x(t))$, where $F(\cdot)$ is a set-valued map. Specifically,
\begin{subequations} \label{Eq: F_As_A_Set} 
\begin{align} 
    \dot{x}(t) &\in F(x(t)) \\
    F(x(t)) &\!=\! \big\{f(x(t)) \!+\! \noise(t) \st \noise(t) \in \Noise \ \forall t \in [0,\infty) \big\}.
\end{align}
\end{subequations}
Here, $f: S \rightarrow \R^n$ models the dynamical system with state space $S \subset \R^n$, and $\noise: [0,\infty) \rightarrow \Noise$ is the disturbance function, where $\Noise \subset \R^n$ is the bounded noise space. If $\Noise = \{0\}$, the dynamics reduce to the usual $\dot{x}(t) = f(x(t))$. The standard conditions for existence and uniqueness of solutions~\cite{Khalil_Nonlinear_2002} are assumed to be met. Thus, $F(\cdot)$ is assumed to be Lipschitz continuous in the Hausdorff metric, so that
\begin{align} \label{Eq: Lipschitz_Condition}
    d_H(F(x), F(y)) \leq \ell \| x-y \| \ \forall x,y \in \X,
\end{align}
where $\ell$ is the Lipschitz constant for the system. Starting with a polytope $\P \subset \X$ and a simplicial complex $\SiCo$ that triangulates $\bdry \P$, the objective is to find families of polytopic RFISs for the system in~\eqref{Eq: F_As_A_Set}. 

The proposed strategy begins with a simplicial complex $\SiCo$ such that $|\SiCo| = \bdry  \P$, where $\P$ is a convex polytope. Since $\P$ is a topological manifold that can be identified uniquely by its boundary, we apply a sequence of homeomorphisms on $|\SiCo|$ to geometrically deform $\P$ so that it sequentially approximates either $\mRPI$ or $\MRPI$. These deformations are vertex maps on $0$-faces in $\SiCo$ and its subdivisions, such that all simplices satisfy a \BdryC. We first develop this condition and show that if all simplices in $\SiCo$ satisfy the \BdryC, then $|\SiCo|$ is the boundary of an RFIS. This serves as a computational invariance test that decides whether a deformation is to be kept or discarded. The proposed algorithm can be easily implemented knowing only the system dynamics and the coordinates of the $0$-faces in $\SiCo$.

\section{Invariance Test} \label{S: Invariance Test}


Let $\SiCo$ be a simplicial complex that triangulates the boundary of a polytope $\P \subset \R^n$, and $\Si$ be a simplex in $\SiCo$. An example is the dodecahedron in Fig.~\ref{fig: Triangulation}. Let $x \in \Si$ be any point on the simplex. Note that $|\SiCo| =  \bdry \P$, and $x \in \Si \implies x \in \bdry \P$. In this section, we develop a computational tool to decide whether $\SiCo$ triangulates the boundary of an RFIS, for a system with dynamics given in~\eqref{Eq: F_As_A_Set}. 


Let set $A \subset \R^n$ and vector $b \in \R^n$. Let $\langle A,b \rangle$ denote $\sup_{\substack{a \in A}} \langle a,b \rangle$, where $\langle \cdot, \cdot \rangle$ is the standard inner product.

\begin{definition}[\InvC]
A point $x \in \Si \in \SiCo$ is said to satisfy the \InvC~if 
\begin{align}
    \langle F(x), \Normal{\Si} \rangle \leq 0,
\end{align}
where $\Normal{\Si}$ is the outward unit normal vector to $\Si$.
\end{definition}

Note that the \InvC~is related to the sub-tangentiality condition~\cite{Blanchini_Automatica_1999}. It is a condition on the angle between the extreme ray of the tangent cone to the dynamics $F$ given in~\eqref{Eq: F_As_A_Set}, and the normal vector to the simplex $\Si$, at $x$. This is illustrated in Fig.~\ref{fig: Normal_Dynamics} where two boundary points are shown. Here, $ x(t_0)$ does not satisfy the \InvC, but $x(t_1)$ does. If a system trajectory finds itself at $x(t_0)$, then it may escape the set through the boundary point $x(t_0)$, i.e. for some disturbance $\noise$, $\exists \ \eps$ such that $x(t_0 - \eps) \in \P$ but $x(t_0+\eps) \notin \P$. However, no disturbance $\noise \in \Noise$ can result in a system trajectory escaping the set through $x(t_1)$. Similar boundary conditions have been used in works such as~\cite{Varnell_CDC_2016, Varnell_ACC_2018, Mukhopadhyay_ACC_2014, Mukhopadhyay_AJC_2020, Sassi_CDC_2014} for arbitrary sets, but the novelty of this work is in the development of an implementable invariance test and its invocation to perform geometric deformations. The above intuition that the \InvC~being satisfied at each point on the set implies that system trajectories cannot escape the set is proved in the result below.

\begin{theorem} \label{Thm: Bdry_Condition_General}
    If $x_0 \in \Si \in \SiCo$ satisfies the \InvC, then no system trajectory can escape the set $\P$ through $x_0$.   
\end{theorem}
\begin{proof}
    \begin{figure}[t]
    \centering
    \includegraphics[width = 0.495 \textwidth]{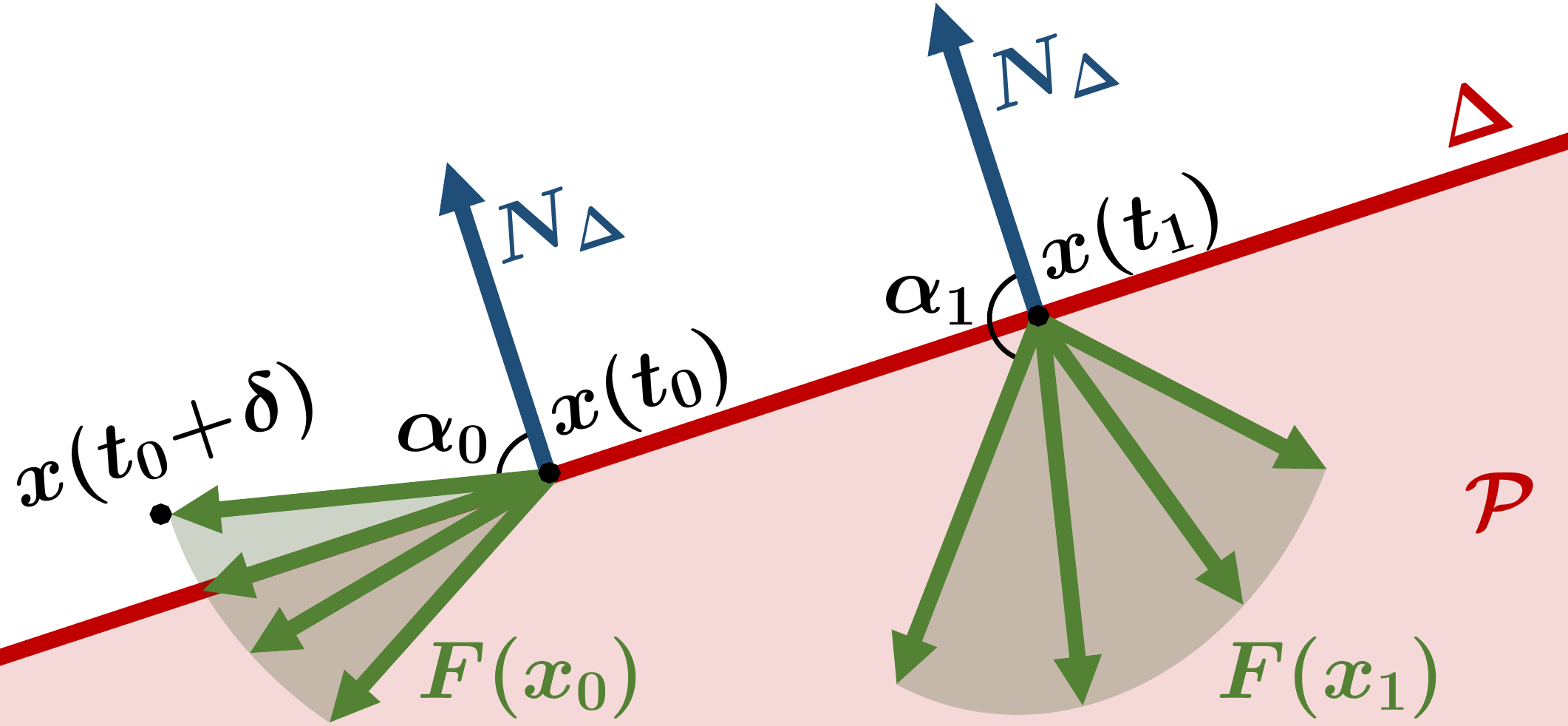}
    \caption{Angle between normal vector and extreme ray of system dynamics}
    \label{fig: Normal_Dynamics}
    \end{figure}
    We prove the contrapositive of this statement, i.e. If a system trajectory escapes the set $\P$ through $x(t_0) \in \Si$, then $\langle F(x(t_0)), \Normal{\Si} \rangle > 0$. Let $\noise(t): [0,\infty) \rightarrow \Noise$ be the disturbance function that resulted in this trajectory, so that $\dot{x}(t) = f(x(t)) + \noise(t)$, and let $x(t_0) \in \bdry \P$. Since the system trajectory escapes $\P$ through $x(t_0)$, $\exists \ \delta > 0$ such that $x(t_0 + \delta) \notin \P$. However,
    \begin{align}
        x(t_0+\delta) &= x(t_0) + \delta \times \dot{x}(t_0) + O(\delta^2) \\
                      &= x(t_0) + \delta \times \left( f(x(t_0)) + \noise(t_0)  \right) + O(\delta^2).
    \end{align}
    As $\delta \rightarrow 0$, 
    \begin{align} \label{Eq: y_delta}
        x(t_0 + \delta) -  x(t_0) \rightarrow \delta \times \left( f(x(t_0)) + \noise(t_0)  \right). 
    \end{align}
    Let $y(\delta) = x(t_0 + \delta) -  x(t_0)$ for convenience, Since $\Normal{\Si}$ is normal to the simplex by definition, it follows that the angle between $\Normal{\Si}$ and $y(\delta) < \frac{\pi}{2}$, whenever $\delta > 0$. This is because $y(\delta)$ is simply the vector from the point on $\bdry \P$ to the point on the trajectory at time $t_0+\delta$, which lies \emph{outside} $\P$. 
    It then follows from~\eqref{Eq: y_delta} that
    \begin{align}
        \alpha_0 &= \cos^{-1} \left( \frac{\langle \delta \times \left( f(x(t_0)) + \noise(t_0) \right), \Normal{\Si} \rangle}{\| \delta \times \left( f(x(t_0)) + \noise(t_0) \right) \| \times \|\Normal{\Si} \|} \right) < \frac{\pi}{2}.
    \end{align}

    Cancelling $\delta$, and taking the cosine of both sides due to it being monotone decreasing over $[0, \frac{\pi}{2})$, it follows that
    \begin{align}
        &\langle f(x(t_0)) + \noise(t_0), \Normal{\Si} \rangle > 0 \\
        \implies &\sup_{\substack{\nu(t_0) \in \Noise}} \langle f(x(t_0)) + \nu(t_0), \Normal{\Si} \rangle > 0 \\
        \implies &\langle F(x(t_0)), \Normal{\Si} \rangle > 0,
    \end{align}
    which proves the desired statement.
\end{proof}

\begin{corollary} \label{Cor: No_Escape_from_Simplex}
    If every point $x \in \Si$ satisfies the \InvC, then no system trajectory can escape $\P$ through $\Si \in \tri(\bdry \P)$. We call such a $\Si$ as an \enquote{invariant simplex}.
\end{corollary}

\begin{corollary}
    If every $\Si \in \SiCo$ is an invariant simplex, then $\P$ is an RFIS.
\end{corollary}

\begin{proof}
    $x \in \bdry \P \implies x \in \bigcup_{\Si \in \SiCo} \Si$. But, no system trajectory can escape $\P$ through any such $\Si$ (by Corollary~\ref{Cor: No_Escape_from_Simplex}). Thus, all system trajectories stay in $\P$, making it an RFIS. 
\end{proof}

Note that if $x \in \Si_1 \cap \Si_2$, it must satisfy the \InvC~with respect to both simplices. Hence, Theorem~\ref{Thm: Bdry_Condition_General} is a test for invariance, as the inner product between the normal vector and the system dynamics can be easily computed. The only caveat is that for a given polytope $\P$ to be invariant, every point on $\bdry \P$ must satisfy the \InvC, and it is computationally impossible to test infinitely many points. However, when $F(\cdot)$ is Lipschitz continuous, this constraint can be overcome as follows. 

\begin{theorem} \label{Thm: Lipschitz_Means_Finite}
    Let $F$ be $\ell$-Lipschitz with respect to the Hausdorff metric. Let $x_0 \in \Si$, where $\Si$ is a simplex with normal vector, $\Normal{\Si}$. If $\exists \ \eps > 0$ such that $\langle F(x_0), \Normal{\Si} \rangle \leq - \eps$, then $\langle F(x), \Normal{\Si} \rangle \leq 0 \ \forall x $ such that $\| x-x_0 \| \leq \frac{\eps}{\ell}$.
\end{theorem}
\begin{proof}
    This follows from Proposition $3$ in~\cite{Varnell_CDC_2016}. 
\end{proof}

Based on Theorems~\ref{Thm: Bdry_Condition_General}~and~\ref{Thm: Lipschitz_Means_Finite}, we develop an algorithm to check if $\Si$ is invariant by testing for a variant of the \InvC~at finitely many points. We call this variant the \BdryC, which requires the inner product to not just be negative but sufficiently negative to ensure that points in an appropriate neighborhood of these test points satisfy the \InvC.
\subsection{Test Point Generation} \label{SS: Test Points}

We present an implementable method to generate a finite set of test points, $\TP_m \subset \Si$, and show how a simple calculation at these points can decide whether $\Si$ is an invariant simplex. To the best of our knowledge, no explicit set of test points has been proposed in previous work - we prove that testing for a variant of the \InvC~at these test points can circumvent the problem of testing at infinitely many points. The set, $\TP_m$ is a lattice, generated a priori on a reference simplex and then mapped to $\Si$ using a linear transformation. Although lattices on simplices are used in other applications like the study of mixtures~\cite{Gorman_ASQ_1962}, the method proposed here generates the lattice using only $\V{\Si}$. 

\subsubsection{Generating Lattice on Reference Simplex} \label{SS: Lattice On Reference Simplex}
\begin{figure}[t]
    \centering
    \includegraphics[width = 0.495 \textwidth]{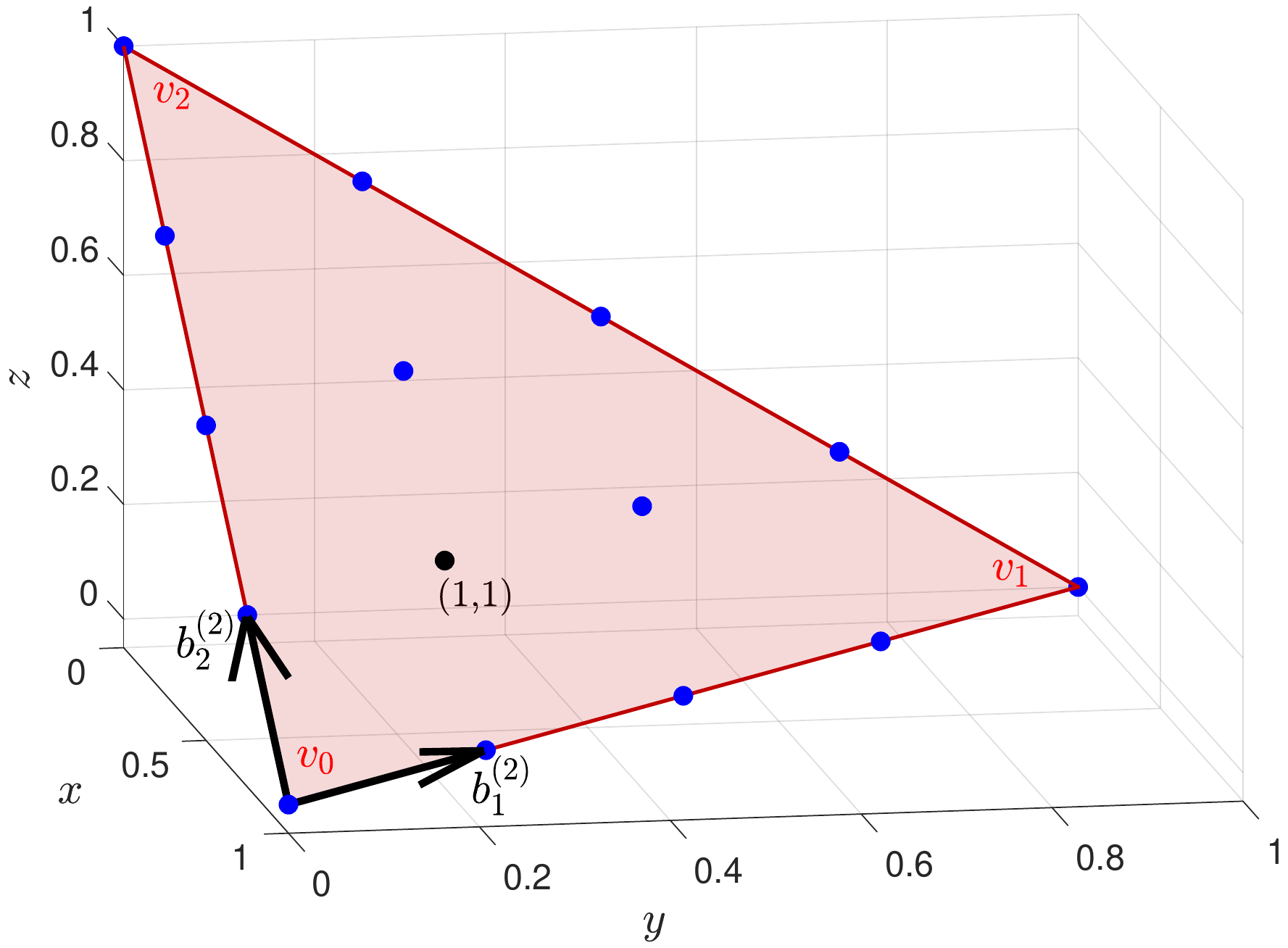}
    \caption{Standard Simplex with Coordinate System}
    \label{fig: Ref_Simplex_1}
\end{figure}

Let $\Lambda$ be the standard simplex in $\R^n$, defined as
\begin{align}
    \Lambda = \{\lambda \in \R^n \st \sum_{i = 1}^{n} \lambda_i = 1 , \lambda_i \geq 0\}.
\end{align}
To control the coarseness of the lattice, a parameter $\ss_m \in (0,1]$ is introduced. Here, $m$ is the iteration number in the algorithm to test for invariance. Since we would like the lattice to get exponentially finer, we set
\begin{align}
    \ss_m = 2^{-m}, \quad m = 0,1,\cdots,m_{\max}, \label{Eq: ss_m}
\end{align}
where $m_{\max}$ is the maximum allowed iteration number for the algorithm. In order to identify test points on the simplex, we introduce a new coordinate system for each value of $m$ and introduce a condition that guarantees that a given point in this coordinate system lies on the simplex. The change of basis allows us to easily identify the test points.

\begin{lemma}
    For a fixed $m$ such that $0 \leq m \leq m_{\max}$, consider a coordinate system with $v_0$ as the origin, and $(n-1)$ basis vectors given by 
    \begin{align}
        b_j^{(m)} = v_0 + \ss_m \left( v_j - v_0 \right),\quad j = 1,\cdots,n-1.
    \end{align}
    In this coordinate system, a point $a = (a_1,\cdots,a_{n-1})$ is in $\Lambda$ if and only if $\sum_{i = 1}^{n-1}a_i \in [0,\invss_m]$.
\end{lemma}
\begin{proof}
Writing out $(a_1,
\cdots a_{n-1})$ using the basis vectors gives
\begin{align}
    \sum_{i = 1}^{n-1} a_i b_i^{(m)} &= \sum_{i = 1}^{n-1}  a_i \left(v_0 + \ss_m (v_i - v_0) \right) \\
    & = \left(\sum_{i = 1}^{n-1} \ss_m a_i v_i \right)+ (1 - \ss_m)v_o \sum_{i=1}^{n-1}a_i
\end{align}
Let $a_0 = \frac{1-\ss_m \sum_{i=1}^{n-1}a_i}{\ss_m}$, so that the point $a$ can be written as $\sum_{i= 0}^{n-1} \ss_m a_i v_i$ with respect to the standard basis. Further, $\sum_{i=0}^{n-1} \ss_m a_i = 1$ and $0 \leq \ss_m a_i \leq 1 \ \forall i$, whenever $a_i \in [0,\invss_m]$. Hence, the point is written as a convex combination of the vertices of the simplex, and therefore lies in it.
\end{proof}

Next, we define a lattice of test points on the standard simplex as follows. This lattice will later be mapped to the test simplex through a linear transformation.
\begin{definition}[Lattice Test Point]
$a = (a_1,\cdots,a_{n-1})$ is a lattice test point if $a \in \Si$ and $a_i \in \N_{0} = \{ 0,1,2,\cdots\} \ \forall 1 \leq i \leq n-1$.
\end{definition}

For instance, the $2$-simplex $\Lambda$ in $\R^3$ with vertices $v_0$, $v_1$ and $v_2$ is shown in Fig.~\ref{fig: Ref_Simplex_1} for $m = 2$, along with the basis vectors $b_1^{(2)}$ and $b_2^{(2)}$, and the test point $(a_1,a_2) = (1,1)$. Note that the vertices of $\Lambda$ are always included in the lattice for all $m$.

Thus, all lattice points can be computed a priori in this coordinate system and then represented in the standard coordinate system for all values of $0 \leq m \leq m_{\max}$. It follows from the construction that the number of test points generated in iteration $m$ for an $(n-1)$-simplex in $\R^{n}$ with parameter $\ss_m$ is
\begin{align}
    \Lm = \binom{\invss_m + n - 1}{\invss_m} = \frac{ (\invss_m + n - 1)!}{\invss_m! \times (n-1)!}.
\end{align}

\begin{figure}[t]
    \subfigure[$\TP_1$]{
    \centering
    \includegraphics[width = 0.225 \textwidth]{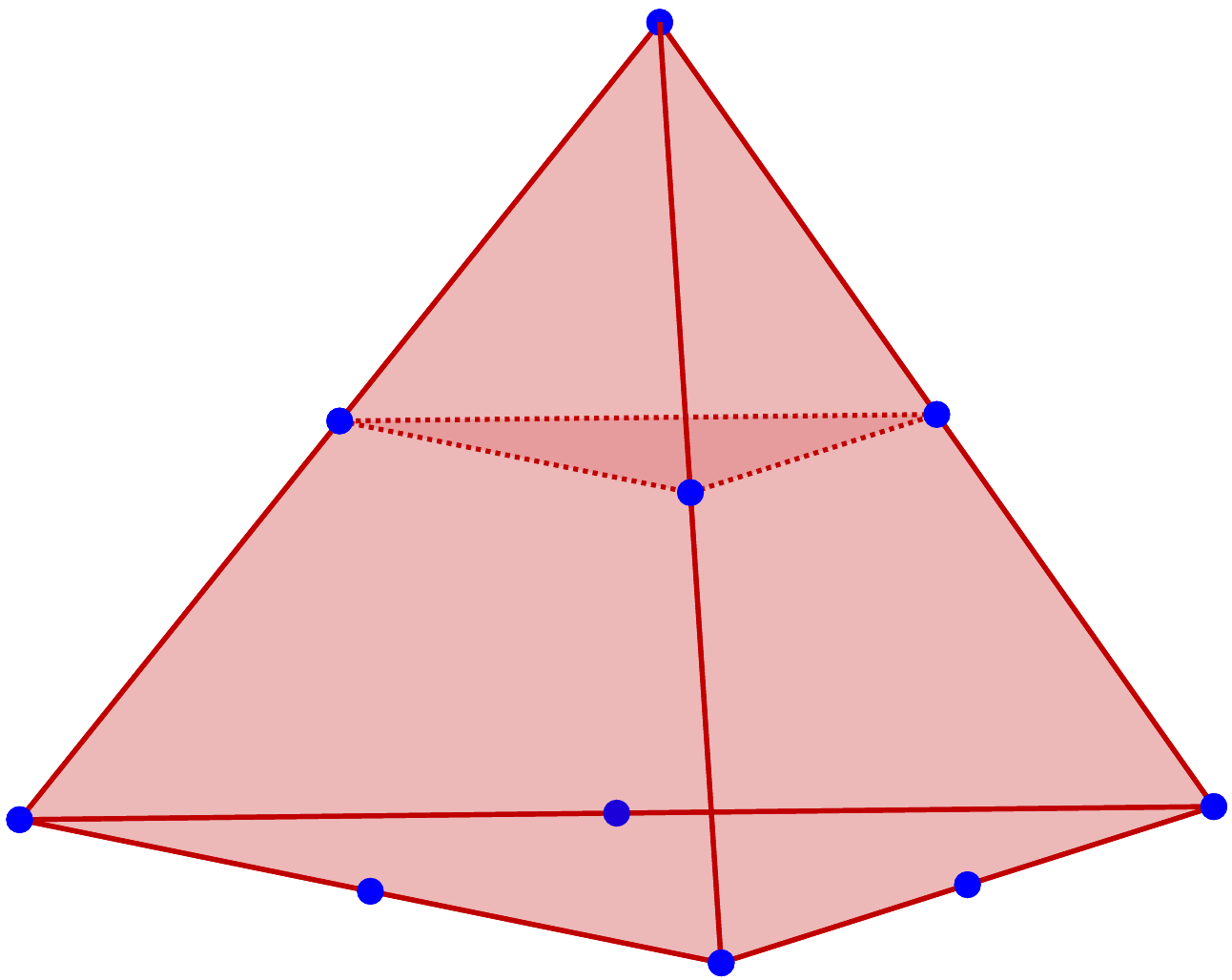}
    \label{fig: Tetrahedron_A}
    } 
    \hfil
    \subfigure[$\TP_2$]{
    \centering
    \includegraphics[width = 0.225 \textwidth]{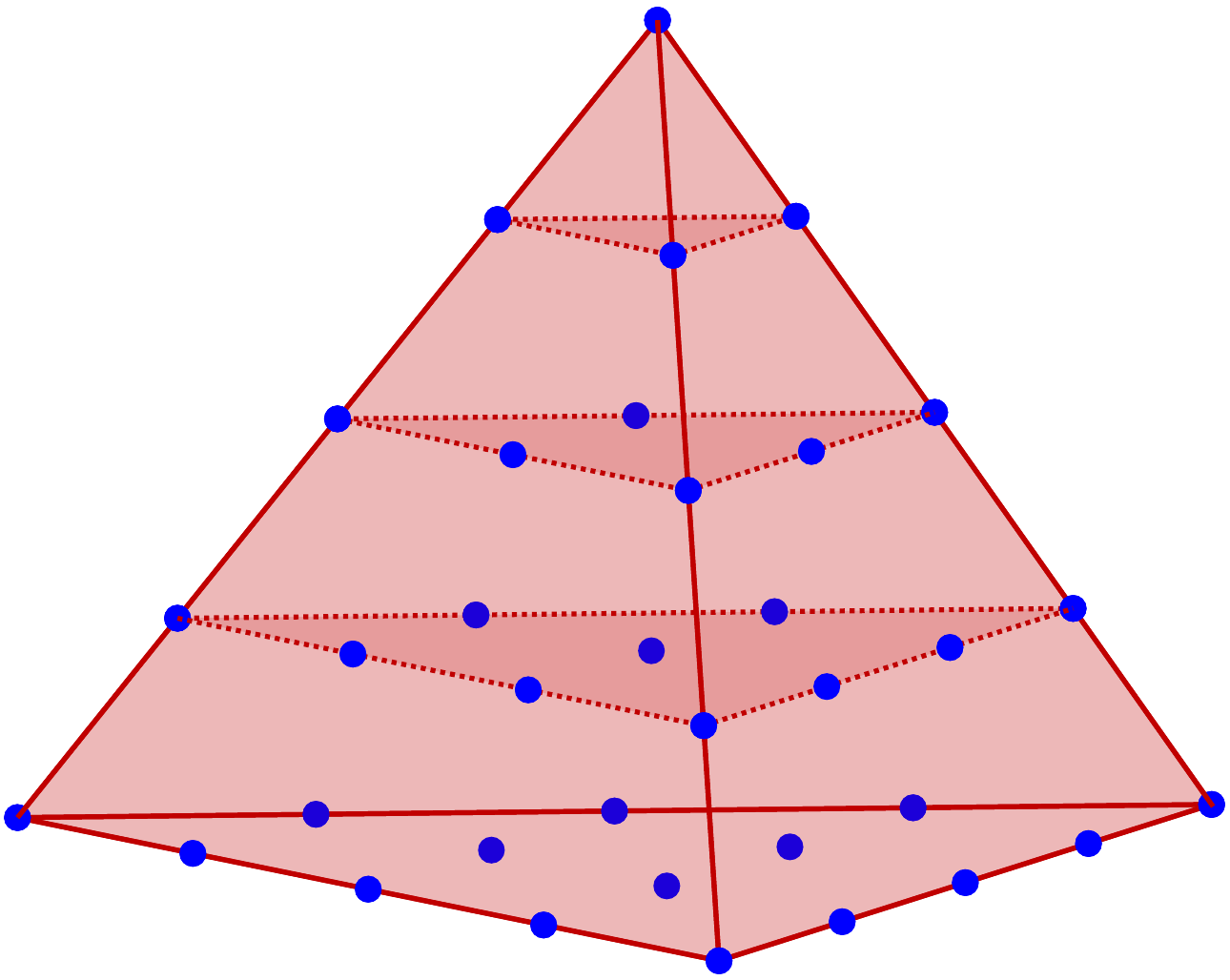}
    \label{fig: Tetrahedron_B}
    } 
   \caption{Lattice of Test Points for a $3$-Simplex}
    \label{fig: Tetrahedron_Test_Points}
\end{figure}

\subsubsection{Mapping test points to a simplex in the simplicial complex that triangulates the boundary of a polytope}
The $\Lm$ coordinates generated on the standard simplex as described above are stored as rows of a matrix, $\M \in \R^{\Lm \times n}$. Note that this representation is in the standard basis. We use the $m$ subscript to emphasize the dependence on $\ss_m$, and note that $\M$ can be computed a priori for different $m$. To obtain the lattice points on the simplex of interest, $\Si$, we perform the following linear transformation, $T : \M  \rightarrow \R^{\Lm \times n}$ given as
\begin{align} \label{Eq: Transformation}
    T(\M) = \M L_{\Si}.
\end{align}
Let
\begin{align}
    \TP_m = \{ x \in \R^n \st x \text{ is a row of } T(\M) \},
\end{align}
be the set of all proposed test points on $\Si$. For instance, Fig.~\ref{fig: Tetrahedron_Test_Points} shows $\TP_1$ and $\TP_2$ for a $3$-simplex. Note that a set of $3$-simplices would triangulate the boundary of a 4-dimensional RFIS.

We now prove that the generated test lattice is indeed on the simplex, and that it is sufficient to check for a variant of the \BdryC~at only these lattice points to guarantee that the simplex is invariant.

\begin{lemma} \label{Lem: Lattice_Is_A_Subset_Of_Simplex}
    $\TP_m \subset \Si$. 
\end{lemma}
\begin{proof}
    See Appendix A.
\end{proof}

Thus, all test points calculated on $\Lambda$ map to the desired simplex, $\Si$ under the transformation $T(\cdot)$. These test points can be used to determine whether or not a given simplex is invariant, using the following theorem.

\begin{theorem} \label{Thm: Bdry_Condition_Discretized}
    Let $\Si$ be a given simplex with $\V{\Si} = \{v_0, \cdots v_{n-1}\}$, $r = \max_{i,j}{\| v_i - v_j \|}$, and $\TP_m$ be the set of test points on $\Si$ for a given $\ss_m$. Suppose that
    \begin{align} \label{Eq: Bdry_Condition_Discretized}
        \langle F(x), \Normal{\Si} \rangle \leq - r \ss_m \ell, \quad \forall x \in \TP_m
    \end{align}
    where $\ell$ is the Lipschitz constant for the system dynamics given by $F$. Then, $\Si$ is an invariant simplex.
\end{theorem}
\begin{proof}
See Appendix B.
\end{proof}
A test point that satisfies~\eqref{Eq: Bdry_Condition_Discretized} is said to satisfy the \BdryC. Theorem~\ref{Thm: Bdry_Condition_Discretized} is a tool to test whether $\Si$ is an invariant simplex. This process of testing whether a given simplex is invariant is summarized in Algorithm~\ref{Alg: Simplex_Invariance_Test}. The algorithm is named \texttt{BCD\_Test} since it generates the set of test points on the simplex and checks whether they satisfy the \BdryC. It takes as input the simplex, which it identifies by its vertices. Additionally, the system dynamics $F$, Lipschitz constant $\ell$, and the maximum number of iterations to run the algorithm $m_{\max}$ are also input to the algorithm. Further, the algorithm also has access to the constants $\ss_m$, $M_m$, $\Lm$ which are computed a priori and stored in memory.  The algorithm computes the ratio of test points for which the \BdryC~is not satisfied to the total number of test points when $\ss_m = 2^{-m_{\max}}$, denoted $\NTP$. Note that $\Si$ is invariant if $\NTP = 0$. 

Line~\ref{L: L_Si} of Algorithm~\ref{Alg: Simplex_Invariance_Test} computes the matrix $L_{\Si}$, and line~\ref{L: r} computes $r = \max_{i,j} \|v_i - v_j\|$, which is the length of the longest 1-face of the simplex. Line~\ref{L: N} computes the normal vector to the simplex using~\eqref{Eq: Normal_Of_Simplex}. The invariance condition is then checked on the test points that successively get finer in each iteration $m$. Observe from~\eqref{Eq: Bdry_Condition_Discretized} that as $m$ increases - so do the number of points that must satisfy the~\BdryC, but the condition on how negative the inner product $\langle F(x), \Normal{\Si} \rangle$ must be is progressively relaxed. Thus, iteration $m$ of the \textbf{for} loop in line~\ref{L: for_A1} begins by initializing $\NTP_m = 0$, where $\NTP_m$ counts the number of test points that violate the~\BdryC. Note that checking the \BdryC~at finitely many points is equivalent to matrix-vector multiplication. Since each row of the matrix $T(\M)$ computed in line~\ref{L: T(M)} is a test point, we denote by $F(T(\M))$ the matrix whose corresponding row is the result of evaluating $F(\cdot)$ at these test points. Since the normal vector is the same for all these test points, the inner product can be computed as in line~\ref{Line: Inner_Prod_as_Matrix_Vector}. The result is a vector $\mathcal{IP}$ containing the values of the inner product at each test point. If all entries are sufficiently negative for any $m$ as determined by Theorem~\ref{Thm: Bdry_Condition_Discretized}, then we are guaranteed that the simplex is invariant. Thus, the \textbf{foreach} loop in line~\ref{L: foreach_A1} counts the number of test points that violate the~\BdryC. If this number is $0$ for any $m$, then the algorithm returns $\NTP = 0$ and declares the simplex to be invariant. If not, then it increments $m$ until the maximum allowed iteration $m_{\max}$ is reached, and outputs $\NTP(\Si)$ as the ratio of test points that violate the~\BdryC~to the total number of test points $\Lm$ when $m = m_{\max}$, as computed in line~\ref{L: Ratio}. Here, $0 \leq \NTP(\Si) \leq 1$ is a measure of how close $\Si$ is to being invariant.

\begin{algorithm}[t]
\DontPrintSemicolon
\caption{\texttt{BCD\_Test} (\BdryC~Test)} \label{Alg: Simplex_Invariance_Test}
	\SetKwInOut{Inputs}{Inputs}
	\SetKwInOut{Constants}{Constants}
    \SetKwInOut{Outputs}{Outputs}
	\Inputs{$\V{\Si}$, $F$, $\ell$, $m_{\max}$}
	\Constants{$\ss_m$, $M_m$, $\Lm \ \forall \ 1 \leq m \leq m_{\max}$}
	\Outputs{$\NTP(\Si)$}
	Compute $L_{\Si}$ from $\V{\Si}$ using~\eqref{Eq: L_Delta} \label{L: L_Si}\;
    Compute $r = \max \| v_i - v_j \|$, where $v_i,v_j \in \V{\Si}$ \label{L: r}\;
    Compute $\Normal{\Si}$ using~\eqref{Eq: Normal_Of_Simplex} \label{L: N} \tcp*{Normal vector}
    \For{$m \in \{0,1,\cdots, m_{\max} \}$ \label{L: for_A1} }{
        Set $\NTP_m = 0$ \label{L: NTP}\; 
        Compute $T(\M)$ using~\eqref{Eq: Transformation} \label{L: T(M)}\;
        Compute $\mathcal{IP} = F(T(\M)) \Normal{\Si}$ \label{Line: Inner_Prod_as_Matrix_Vector} \tcp*{Vector of //inner products at all test points}
        \ForEach{component $x$ of $\mathcal{IP}$ \label{L: foreach_A1}}{
            \If{$x > -r \ss_m \ell$}{
                $\NTP_m = \NTP_m + 1$ \label{L: NTP+1}\;
            }
        }
        \If{$\NTP_m == 0$}{
            \textbf{return} $\NTP(\Si) = 0$ \; 
            \textbf{end program} \;
        }
        }
    \textbf{return} $\NTP(\Si) = \frac{\NTP_{m}}{\Lm} $ \label{L: Ratio} \tcp*{$0 \leq \NTP(\Si) \leq 1$}
\end{algorithm}

The next section invokes this simplex invariance test to perform geometric deformations that attempt to transform convex polytopes to robust forward invariant sets.


\section{RFIS Computation Algorithm} \label{S: Algorithm}


In this section, we use Algorithm~\ref{Alg: Simplex_Invariance_Test} as a tool in developing a polytopic RFIS computation algorithm. 
\subsection{Proposed Vertex Map} \label{SS: Vertex Map}

We are interested in systems with a non-trivial minimal RFIS, $\mRPI$ and maximal RFIS, $\MRPI$. Thus, $\mRPI$ is not just a point, and $\MRPI$ does not enclose infinite volume. We begin with a convex polytope $\P$, such that $\mRPI \subset \P$, and triangulate $\P$ by a simplicial complex $\SiCo$, such that $\V{\SiCo} = \{v_0,\cdots,v_N\}$. 

Consider a point $\o \in \interior(\mRPI)$. Note that although $\mRPI$ is not known, $\o$ may easily be chosen. If the RFIS is due to a known equilibrium point, then $\o$ can be chosen as this point. Else, if it is due to a known limit cycle, $\o$ can be chosen as any point in the interior of the cycle. We consider deformations where all vertices but one are mapped to themselves. Specifically, $\H_{j}: \V{\SiCo} \rightarrow \V{\SiCoII}$ is defined as
\begin{align} \label{Eq: Vertex_Map}
    \H_{j}(v_i) = 
    \begin{cases}
        v_i, & i \neq j \\
        (1 - \decay)\o + \decay v_i, & i = j
    \end{cases}
    ,
\end{align}
where $\decay > 0$ is a growth/decay parameter. Hence, the vertex map $\H_{i}$ only perturbs the vertex $v_i$ by moving it along a ray $\ray_i(\o)$ emanating from $\o$ and containing $v_i$: 
    \begin{align}
         \ray_i(\o) = \{\o + \lambda (v_i - \o) \st \lambda > 0\}.
    \end{align}
An illustration is provided for the case where a simplicial complex of $1$-simplices triangulates the boundary of a $2$-dimensional polytope (a pentagon) in Fig.~\ref{fig: Perturb}. The vertex $v_0$ is mapped to $\H_0(v_0)$ which is constrained on the ray $\ray_{0}(c)$.  
\begin{figure}[t]
	\centering
	\subfigure[$\decay < 1$.]
	{
		\includegraphics[width=0.227 \textwidth]{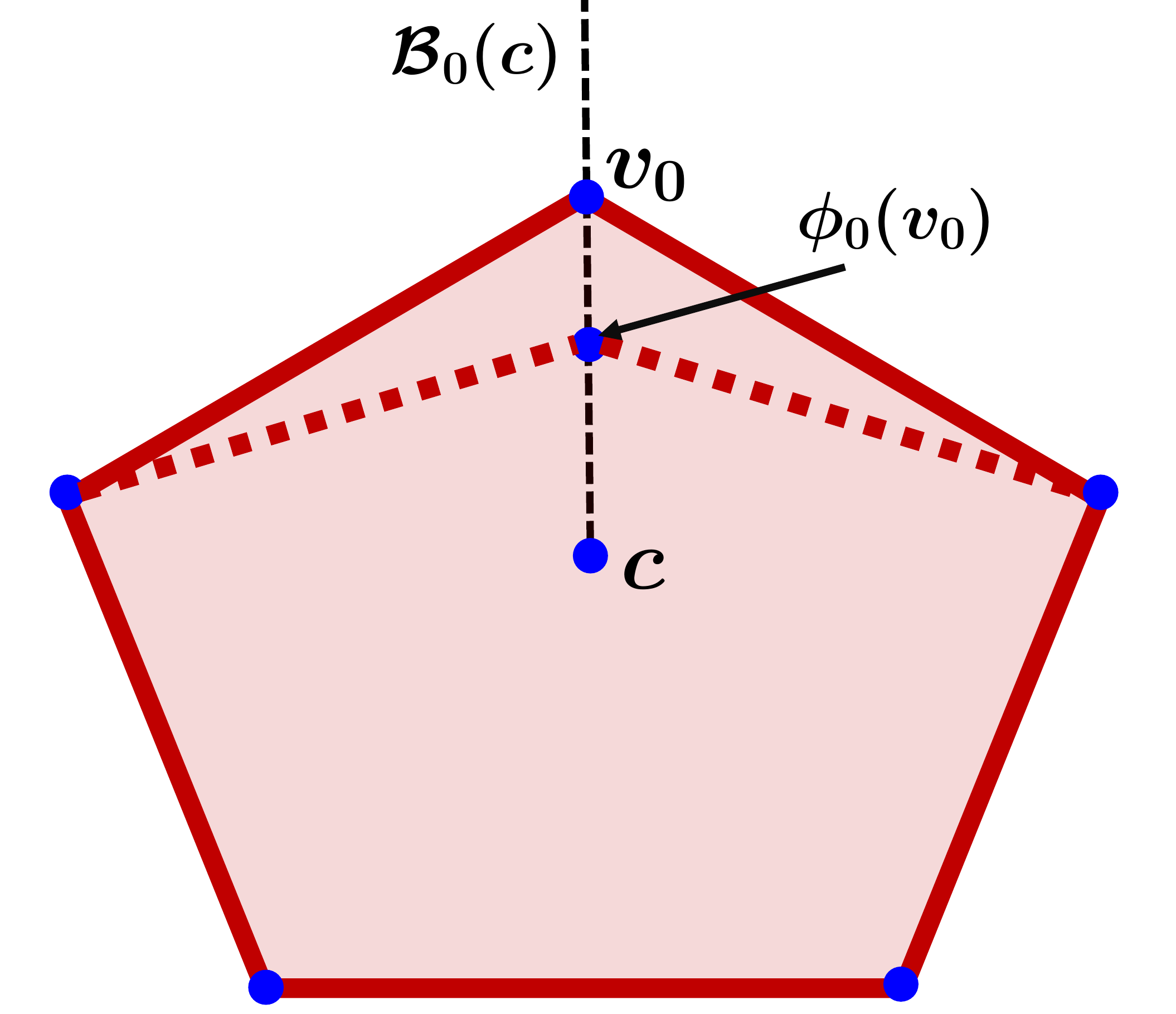}
		\label{fig: P_Decay}
	}
	\hfil
	\subfigure[$\decay > 1$.]
	{
        \includegraphics[width=0.227 \textwidth]{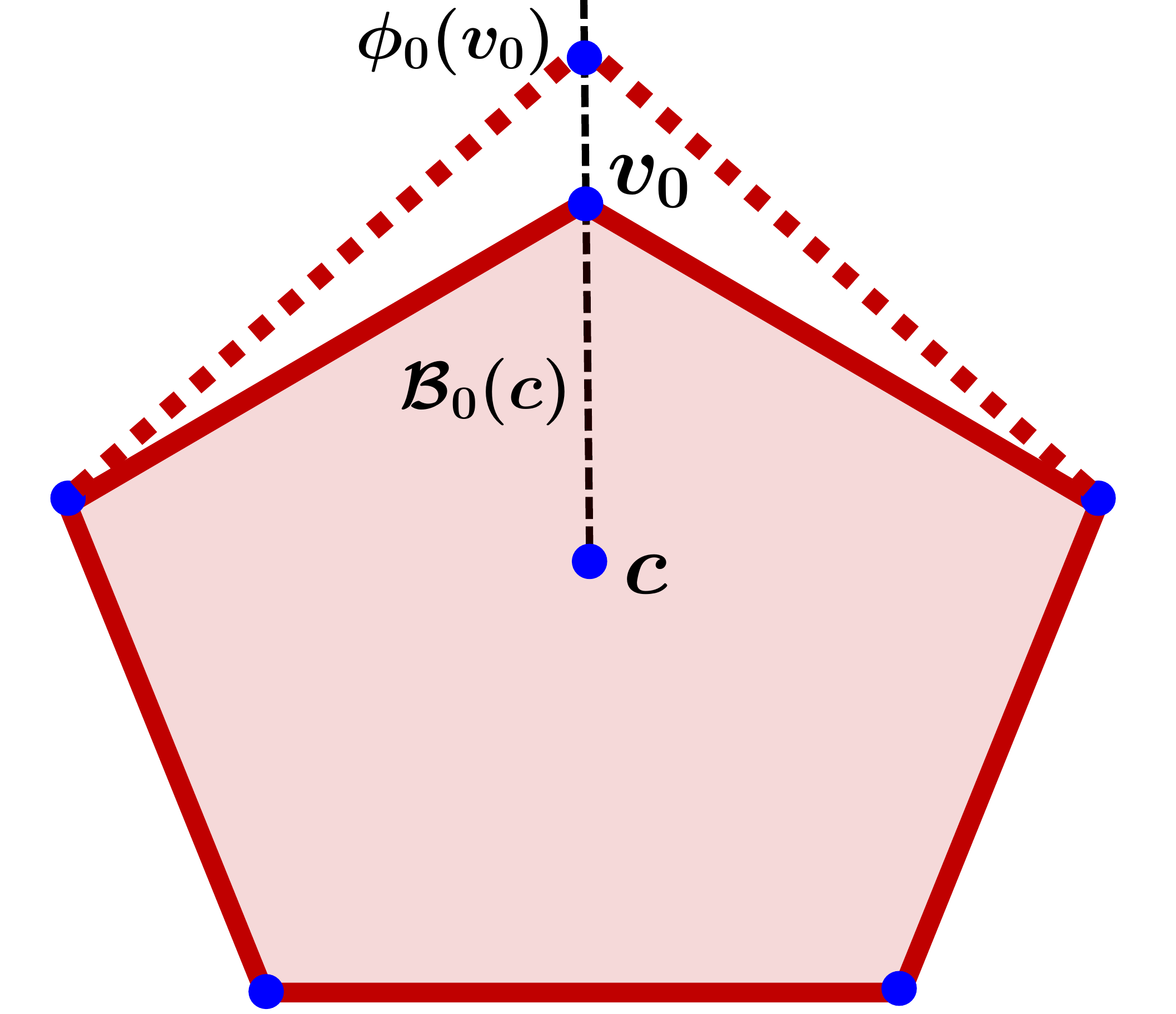}
		\label{fig: P_Grow}
	}
	\caption{Vertex map $\H_0$ when $\decay < 1$ and $\decay > 1$.}
	\label{fig: Perturb}
\end{figure}

Observe that only the points in the closed star of the perturbed vertex are changed due to the vertex map. It follows that $\CH_{j}(x) = x$ if $x \notin \star_{\SiCo}(v_j)$, where $\CH_{j}: |\SiCo| \rightarrow |\SiCoII|$ is the induced map due to $\H_j$, as defined in~\eqref{Eq: Cont_Homeo}. We now show that the deformed polytope is homeomorphic to the initial polytope.

\begin{theorem} \label{Thm: Deformation_is_Homeomorphism}
    Let $\CH_j$ be the induced map due to $\H_j$ as defined in~\eqref{Eq: Vertex_Map}. If $\decay > 0$, then $\CH_j$ is a homeomorphism $\forall j$.
\end{theorem}
\begin{proof}
    A vertex map induces a homeomorphism if and only if it has an inverse that is also a vertex map. Clearly, the map defined by
    \begin{align}
    \H_{j}^{-1}(v_i) = 
    \begin{cases}
        v_i, & i \neq j \\
        (1 - \decay^{-1})\o + \decay^{-1} v_i, & i = j
    \end{cases}
    ,        
    \end{align}
    is the inverse map if $\H_j(v_i) \neq \H_j (v_k) \ \forall j$ and $\forall i \neq k$. 
    Since $\decay > 0$, it follows that $\H_j(v_i) \in \ray_i(\o)$. Since $\o \in \interior(\P)$, and since the simplicial complex triangulates the boundary of the convex polytope $\P$, no two of its vertices lie on the same ray. Thus, $\ray_i(\o) \cap \ray_k(\o)$ is the empty set, and no two vertices are ever mapped to the same point, making $\H_j$ a bijective vertex map. Hence, its induced map, $\CH_j$ is a homeomorphism.  
\end{proof}

\begin{corollary}
    Let $\H : \V{\SiCo} \rightarrow \V{\SiCoII}$ be defined as
    \begin{align}
        \H(v_i) = \H_{j_k} \circ \cdots \circ \H_{j_1} (v_i), \label{Eq: Vertex_Map_Composed}
    \end{align}
    where $\H_{j_{l}}$ are vertex maps as defined in~\eqref{Eq: Vertex_Map}. Then, the geometric realization $|\SiCoII|$ is homeomorphic to the geometric realization $|\SiCo|$. 
\end{corollary}
\begin{proof}
    First, note that the vertex map $\H(v_i)$ in~\eqref{Eq: Vertex_Map_Composed} entails the application of $\H_{j_1}$ followed by $\H_{j_2}$. If $v_i \in \ray_i(\o)$, then $\H_{j_1}(v_i) \in \ray_i(\o)$ by definition. Since $\ray_i(\o) \cap \ray_k(\o)$ is the empty set, $\H_{j_2}$ is also a bijective vertex map and induces a homeomorphism. Since the composition of two homeomorphisms is also a homeomorphism, and since any finite sequence of the vertex maps in~\eqref{Eq: Vertex_Map} restricts the image of $v_i$ to $\ray_i(\o)$, it follows by induction that $\H_{j_k} \circ \cdots \circ \H_{j_1}$ also induces a homeomorphism. Thus, $|\SiCo|$ and $|\SiCoII|$ are homeomorphic.
\end{proof}

The vertex map in~\eqref{Eq: Vertex_Map_Composed} deforms the polytope into a manifold that is homeomorphic to it. Therefore, the deformations preserve the topological structure of the polytope, so that the simplicial complex obtained by the sequence of deformations triangulates the boundary of a new polytope. This new polytope need not be convex, since homeomorphisms do not preserve convexity in general. However, it is not self-intersecting since $\H(v_i) \in \ray_i(\o) \ \forall i$. A vertex map is either kept or discarded depending on whether or not it serves the objective of approximating an RFIS. For this, we require a measure of how \enquote{close} a given simplex $\Si$ is to being invariant. One such measure is $\NTP(\Si)$, the ratio of test points on $\Si$ that violate the \BdryC, to the total number of test points on $\Si$. We propose that a vertex perturbation $\H: \V{\SiCo} \rightarrow \V{\SiCoII}$ be kept if
\begin{align} \label{Eq: Stop_Condition}
\left( \sum_{\Si \in \SiCoII}\NTP(\Si) < \sum_{\Si \in \SiCo}\NTP(\Si) \right)\ \textbf{\textrm{or}} \ \sum_{\Si \in \SiCoII}\NTP(\Si) = 0,
\end{align}
and discarded otherwise. 

\subsection{Polytopic RFIS Computation Algorithm} \label{SS: RFIS Computation Algorithm}
This section presents an algorithm to compute RFISs of different sizes for a given non-linear $\ell$-Lipschitz dynamical system with bounded additive disturbances. The essence of the algorithm is to repeatedly perturb vertices in the simplicial triangulation through vertex maps so that each test point in all the simplices satisfies the~\BdryC. A particularly interesting application of this is when $\P$ itself is an RFIS for the system obtained through any other method. Each deformation would result in a smaller or larger RFIS depending on whether $\decay < 1$ or $\decay > 1$, thereby creating an RFIS family. The appropriate RFIS may be considered for the application in concern. Once a stage is reached where no vertex can be perturbed further (all perturbations would violate~\eqref{Eq: Stop_Condition}), the simplices in the simplicial complex $\SiCoII$ are subdivided. Since the geometric realizations before and after the subdivision are the same, i.e. $|\S(\SiCoII)| = |\SiCoII|$, it follows that the application of vertex maps of the form in~\eqref{Eq: Vertex_Map_Composed} preserves homeomorphicity. On the other hand, since $\star_{\S(\SiCo)}(v) \subset \star_{\SiCo}(v)$, the deformations are finer, allowing the new polytope to better approximate the shape of the RFIS. This allows for a family of algorithms, each with different types and sequences of simplicial deformations and subdivisions. One such algorithm is presented in Algorithm~\ref{Alg: RFIS_Computation}, where the polytope is deformed by perturbing a single vertex and checking whether or not to discard the deformation immediately after. The next paragraph details the steps. 

Algorithm~\ref{Alg: RFIS_Computation} takes as its input the simplicial complex $\SiCo$ that triangulates the boundary of the initial convex polytope, the decay parameter $\decay$ which controls the rate with which the volume of the set is increased or decreased while searching for the RFIS, and the maximum number of permissible subdivisions $t_{\max}$. Since each iteration involves simplex subdivisioning, $t_{\max}$ bounds the maximum number of simplices in the simplicial complex that triangulates the boundary of the final polytopic set. The algorithm returns as output an ordered list of sets $\SiCoSet$, for which the $j$-th entry $\SiCoSet[j]$ is the simplicial complex that triangulates the boundary of the deformed polytope after $j$ iterations of the algorithm. Line~\ref{L2: Init_I} of the algorithm initializes the values of the output set $\SiCoSet[0]$. 

The idea behind Algorithm~\ref{Alg: RFIS_Computation} is to repeatedly perturb the vertices in a sequence until no vertex can be perturbed further, i.e. every vertex perturbation violates~\eqref{Eq: Stop_Condition}. Since this entails possibly perturbing a vertex more than once, we introduce an indicator Boolean variable $\status[j] \in \{0,1\}$ that is 0 if the perturbation $\H_j(v_j)$ in the previous iteration was discarded due to violating~\eqref{Eq: Stop_Condition}, and 1 otherwise. Since all vertices are initially candidates for being perturbed, line~\ref{L2: Init_status} initializes $\status[j] = 1 \ \forall j$. Consequently, each vertex is perturbed at least once through the vertex map $\H_j$ as outlined in line~\ref{L2: Perturb}. Whether or not a perturbation is successful in deforming the polytope is determined by computing~\eqref{Eq: Stop_Condition}. Since $\H_j(v_j)$ induces a homeomorphism that only affects points in $\star_{\SiCo}(v_j)$, it is sufficient to check for~\eqref{Eq: Stop_Condition} over only the simplices in the closed star of the perturbed vertex. Therefore, lines~\ref{L2: ForEach2}-\ref{L2: Compute_NTP} invoke Algorithm~\ref{Alg: Simplex_Invariance_Test} to compute $\NTP(\Si)$ - the ratio of test points that violate the \BdryC~to the total number of test points, for all the simplices in the closed stars of a vertex in $\SiCo$ (before perturbation) and in $\SiCoII$ (after perturbation). Based on this, line~\ref{L2: uIf} checks if~\eqref{Eq: Stop_Condition} is true and lines~\ref{L2: Assign}-\ref{L2: If_End} accordingly update the simplicial complex that triangulates the boundary of the polytope as follows: If a vertex deformation $\H_j$ satisfies~\eqref{Eq: Stop_Condition}, then it is kept as shown in line~\ref{L2: Assign}, and the variable $\status[j]$ is set to $1$ in line~\eqref{L2: stat1}, indicating that further perturbation of vertices may be possible. If the vertex deformation must be discarded due to not satisfying~\eqref{Eq: Stop_Condition}, then the variable $\status[j]$ is set to $0$ as shown in line~\eqref{L2: stat2}. Note that the stop condition for the \textbf{while} loop is that $\status[j] = 0 \ \forall j$ (line~\ref{L2: While_Main}), and that \emph{all} vertices are perturbed even if $\status[j] = 1$ for only some $j$. This tends to avoid local minima, since perturbation of a vertex affects its entire closed star, and thus neighboring vertices that previously did not satisfy $\eqref{Eq: Stop_Condition}$ may do so now upon perturbation. If~\eqref{Eq: Stop_Condition} is still not satisfied, then the vertex deformation is discarded. This process is iterated until a stage is reached where all further perturbations violate~\eqref{Eq: Stop_Condition}, which causes the algorithm to exit the \textbf{while} loop. The resulting simplicial complex which triangulates the boundary of the deformed polytope is stored as $\SiCoSet[t]$ for iteration $t$ in line~\ref{L2: Store}. Next, all simplices in the simplicial complex undergo barycentric subdivision as shown in line~\ref{L2: Subdivision}, and the vertex set of the simplicial complex is accordingly updated. This application of a sequence of vertex perturbations followed by a simplex subdivision is considered one iteration of the algorithm. The algorithm terminates when $t_{\max}$ iterations are completed. Since any polytopic RFIS satisfies $\NTP(\Si) = 0 \ \forall \Si \in \SiCo$, where $\SiCo$ triangulates the boundary of the polytopic RFIS, the algorithm checks whether the last simplicial complex $\SiCoSet[t-1]$ triangulates the boundary of a polytopic RFIS in line~\ref{L2: Check_RFIS}.

\begin{algorithm}[htbp]
\DontPrintSemicolon
\caption{RFIS Computation Algorithm} \label{Alg: RFIS_Computation}
	\SetKwInOut{Inputs}{Inputs}
    \SetKwInOut{Outputs}{Outputs}
	\Inputs{$\SiCo$, $\decay$, $F$, $\ell$, $m_{\max}$, $t_{\max}$}
	\Outputs{$\SiCoSet$}
	Set $\SiCoSet[0] = \SiCo$\; \label{L2: Init_I}
    \For{$t = 1,\cdots,t_{\max}$}{ \label{L2: For_Main}
        Initialize $\status[j] = 1 \ \forall j \in \{1,\cdots,N \}$ \tcp*[1]{$N$ is //the cardinality of $\V{\SiCo}$} \label{L2: Init_status}
        \While{$\status[j] \neq 0 \ \forall j$ \label{L2: While_Main}}{
            \ForEach{$j \in \{ 1,\cdots,N \}$ \label{L2: ForEach}}{ 
                Apply $\H_j: \V{\SiCo} \rightarrow \V{\SiCoII}$ on $v_j$ \label{L2: Perturb} \tcp*{use~\eqref{Eq: Vertex_Map}}
                \ForEach{$\Si \in \star_{\SiCo}(v_j) \cup \star_{\SiCoII}(v_j)$ \label{L2: ForEach2}}{
                Compute $\NTP(\Si) =$ \texttt{BCD\_Test}$(\V{\Si}$,$F$,$\ell$,$m_{\max})$ \label{L2: Compute_NTP}\; 
                }
                \uIf{\eqref{Eq: Stop_Condition} is \emph{true} \label{L2: uIf}}{
                $\SiCo = \SiCoII$ \label{L2: Assign} \;
                Set $\status[j] = 1$ \label{L2: stat1} \;
                }
                \Else{
                Set $\status[j] = 0$ \label{L2: stat2} \;                
                } \label{L2: If_End}
            } 
        } \label{L2: While_End}
        Store $\SiCoSet[t] = \SiCo$ \label{L2: Store}\;
        Set $\SiCo = \Bary(\SiCo)$ \tcp*{Barycentric subdivision} \label{L2: Subdivision}
    }
    \If{$\sum_{\Si \in \SiCoSet[t_{\max}]} \NTP(\Si) \neq 0$ \label{L2: Check_RFIS}}{
        \textbf{display} \enquote{RFIS Not Found} \label{L2: Display}\; 
    }
    \textbf{return} \(\SiCoSet = \{\SiCoSet[0], \cdots, \SiCoSet[t_{\max}]\} \) \label{L2: Return}\;
\end{algorithm}


In order to confirm that the generated polytopes are RFISs of different size, we can easily compute and track the volume enclosed by the polytope as the algorithm progresses, using only the simplicial complex that triangulates its boundary. This is also useful in proving that the algorithm terminates (Theorem~\ref{Thm: Convergence}) and for benchmarking our results with existing work (Section~\ref{S: Simulations}). Although equation~\eqref{Eq: Volume_Of_Simplex} only presents the formula to compute the volume of a simplex, it can be invoked to compute the volume of the polytope as follows. Since $\SiCoSet[k]$ triangulates the boundary of a polytope for all $k$, we construct a triangulation for the polytope itself as follows: If $\Si \in \SiCoSet[k]$, then the set 
\begin{align}
\!\!\!    \SiCoSet'[k] \!=\! \Big\{ \Si ' \st \V{\Si '} = \{\o\} \cup \V{\Si}, \ \forall \Si \in \SiCoSet[k] \Big\},
\end{align}
with $\o \in \mRPI$ triangulates the polytope. Thus, the volume enclosed by the polytope can be computed from the triangulation of its boundary, $\SiCoSet[k]$, as
\begin{align} \label{Eq: Volume_Enclosed}
    \Vol(\SiCoSet[k]) = \sum_{\Si' \in \SiCoSet'[k]} \Vol(\Si'),
\end{align}
where $\Vol(\Si')$ is computed using~\eqref{Eq: Volume_Of_Simplex}.

\begin{theorem} \label{Thm: Convergence}
    Let $\SiCo$ be a simplicial complex that triangulates the boundary of a convex polytope containing $\mRPI$ for a nonlinear $\ell$-Lipschitz continuous dynamical system $F$. Then, starting with $\SiCoSet[0] = \SiCo$ and parameter $\decay \in (0,\infty)$, algorithm~\ref{Alg: RFIS_Computation} terminates. Further, if $\SiCoSet[k]$ is an RFIS for any iteration $k$, then the sequence $(\SiCoSet[l])_{l \geq k}$ is a sequence of RFISs that successively reduce in volume if $\decay < 1$, and increase in volume if $\decay > 1$. 
\end{theorem}

\begin{proof}
    We show that the \textbf{while} loop (lines~\ref{L2: While_Main}-\ref{L2: While_End}) terminates, i.e. the halting criterion $\status[j] = 0 \ \forall j$ is achieved in a finite number of iterations. Assume to the contrary that $\forall N \in \N,  \exists \ k > N$ such that $\status[j] = 1$ for some $j$, in the $k$-th iteration of the \textbf{while} loop. Equivalently,~\eqref{Eq: Stop_Condition} is true atleast once in every iteration. Because $\NTP(\Si) \geq 0 \ \forall \Si \in \SiCo$, and because every strictly decreasing sequence of elements in a finite set converges to its least element, it follows that $\exists N$ such that $\NTP(\Si) = 0 \ \forall \Si \in \SiCo$ and $\forall k > N$. In other words, $\SiCo$ triangulates the boundary of an RFIS for all $k > N$. Since the \textbf{while} loop continues indefinitely, there is atleast one $j$ for which $\status[j] = 1$ infinitely often.
    
    However, since $v_j \in \ray_j(\o) \ \forall j$ by definition, and since
    \begin{align} \label{Eq: Perturbation_Norm}
    \| \H_j(v_j) - \o \|
    \begin{cases}
     < \| v_j - \o \|, & \decay < 1 \\ 
     > \| v_j - \o \|, & \decay > 1 \\ 
    \end{cases}
    ,
    \end{align}
    where $c$ was chosen to be in the interior of $\mRPI$, it follows that $\H_j$ maps $v_j$ in the limit as $k \rightarrow \infty$ to either $\o$ or to $\infty$ along $\ray_j(\o)$, depending on whether $\decay < 1$ or $\decay >1$ repsectively. But, $v_j \in \V{\SiCo}$ and $\SiCo$ triangulates the boundary of an RFIS $\forall k > N$. This contradicts $\o \in \interior(\mRPI)$ if $\decay < 1$, and $\Vol(\MRPI) < \infty$ if $\decay > 1$. We conclude that the \textbf{while} loop terminates in finitely many iterations.
    
    Next, we show that the sequence of sets generated due to the perturbations either increase or decrease monotonically in volume depending on the value of $\decay$. Let $\H_j: \V{\SiCo} \rightarrow \V{\SiCoII}$, where $\SiCo$ and $\SiCoII$ triangulate the boundaries of polytopes $\P_{\SiCo}$ and $\P_{\SiCoII}$ respectively. We show that 
    \begin{align}
    \Vol(\P_{\SiCoII})
        \begin{cases}
            < \Vol(\P_{\SiCo}), & \decay < 1 \\
            > \Vol(\P_{\SiCo}), & \decay > 1 \\
        \end{cases}
        .
    \end{align}
The polytopes $\P_{\SiCo}$ and $\P_{\SiCoII}$ are themselves triangulated by simplicial complexes $\SiCo'$ and $\SiCoII'$ respectively, where these complexes are defined similar to~\eqref{Eq: Volume_Enclosed}. Then,
\begin{align}
    \Vol(\P_{\SiCo}) &= \sum_{\Si \in \SiCo'} \frac{1}{n!}\det(B_{\Si}) \\
    &= \!\! \sum_{\Si \in \star_{\SiCo'}(v_j)} \!\!\!\!\! \frac{1}{n!}\det(B_{\Si}) +\!\! \sum_{\Si \notin \star_{\SiCo'}(v_j)}\!\!\!\!\! \frac{1}{n!}\det(B_{\Si}) \label{Eq: Vol_PL}
\end{align}
The second term in~\eqref{Eq: Vol_PL} is not affected by application of $\H_j$. However, for each simplex $\Si \in \star_{\SiCo'}(v_j)$, the volume after the perturbation is $\alpha$ times the volume before perturbation, as explained in what follows. Let $\V{\Si} = \{v_0,\cdots,v_{n-1},c\}$, so that 
\begin{align}
    \Vol(\Si) &= \frac{1}{n!} \! \left| \det  \begin{bmatrix} v_1 \!-\! v_0,\! & \cdots,\! & v_{n-1} \!-\! v_{0},\! & c \!-\! v_0 \end{bmatrix}  \right|\\ 
    &= \frac{1}{n!} \left| \det \begin{bmatrix} v_0 & v_1 & \cdots & v_{n-1} & c \\ 1 & 1 & \cdots & 1 & 1 \end{bmatrix} \right| \label{Eq: Vol_Mat}
\end{align}
Without loss of generality, consider the vertex map $\H_0: \V{\SiCo} \rightarrow \V{\SiCoII}$, and denote by $\Si_0$ the simplex after perturbation that corresponds to $\Si$. By definition of $\H_0$, it follows that
\begin{align}
    \! \Vol(\Si_1) & \!=\! \frac{1}{n!} \left| \det  \begin{bmatrix} \decay v_0 \!+\! (1\!-\!\alpha) \o &  \cdots & v_{n-1} & c \\ 1  & \cdots & 1 & 1 \end{bmatrix}  \right|. \label{Eq: Vol_Mat2}
\end{align}
But the matrix in~\eqref{Eq: Vol_Mat2} can be obtained using elementary column transformations on the matrix in~\eqref{Eq: Vol_Mat}. Since multiplying a column by $\decay$ multiplies the determinant by the same factor, it follows that $\Vol(\Si_0) = \decay \Vol(\Si) \ \forall \Si \in \star_{\SiCo'}(v_j) \ \forall j$. Thus, each vertex map on the simplicial complex increases or decreases the volume of the polytope it triangulates, depending on whether $\decay > 1$ or $\decay < 1$ respectively.
\end{proof}

\begin{corollary} \label{Cor: Start_Stay}
    If the initial set $\SiCoSet[0]$ is an RFIS, then $\SiCoSet[k]$ is an RFIS for all $k$.
\end{corollary}
Theorem~\ref{Thm: Convergence} suggests that if $t_{\max}$ is sufficiently large, the sequence $(\SiCoSet[k])$ will converge either to $\mRPI$ or $\MRPI$ respectively, depending on whether $\decay < 1$ or $\decay > 1$. However, the algorithm is greedy and may get stuck at a local minima if the deformations do not significantly change the volume enclosed by $\SiCoSet[k]$. We simulate this for systems with known $\mRPI$ in two and three dimensions, and observe that when $\decay < 1$, a good choice of initial polytope results in fast convergence to a polytopic approximation of $\mRPI$, with $t_{\max}$ as less as 6.

It is worth mentioning that the entire algorithm can be implemented knowing only $\V{\Si} \ \forall \Si \in \SiCo$. The quantities $\Normal{\Si}$, $\Vol(\Si)$ and $\TP_m$ are directly computed using the vertex coordinates, and each vertex map replaces a row of $L_{\Si}$ for each $\Si$ that shares the vertex. Further, simplex subdivisioning simply replaces a given $L_{\Si}$ with multiple such matrices, where the rows (vertices) are computed as outlined in~\cite{Edelsbrunner_Book_2010}.

\section{Simulations} \label{S: Simulations}


This section presents the results of applying Algorithm~\ref{Alg: RFIS_Computation} to various nonlinear dynamical systems. In all the figures, blue arrows and red curves represent the vector field and system trajectories respectively. The computation time for each iteration of the algorithm, and the volume enclosed by the resulting simplicial complex is presented in Table~\ref{tab: Computation}. Section~\ref{SS: 2D} presents the simulation results for two-dimensional systems, and considers systems with and without additive disturbances. Further, Section~\ref{SS: 3D} presents the simulation results for three-dimensional systems. The algorithm terminates in an RFIS for each presented case, so that $\NTP(\Si) = 0$ for all the simplices in the covering of the final polytope boundary. All computations were performed on a Dell Workstation with an Intel Xeon processor @ 3.30 GHz and 16 GB of RAM.


\begin{table*}
   \centering
   \caption{Computation Time and Enclosed Volume after Each Iteration for Various Systems}
    \begin{tabular}{|r||r|r||r|r||r|r||r|r||r|r||r|r|}
    \hline
    \multicolumn{1}{|c||}{\multirow{2}{*}{Iteration} } & \multicolumn{2}{c||}{Van der Pol (VdP)} & \multicolumn{2}{c||}{Fitzhugh-Nagumo} & \multicolumn{2}{c||}{Curve Tracking} & \multicolumn{2}{c||}{Reversed VdP} & \multicolumn{2}{c||}{Phytoplankton Growth} &    \multicolumn{2}{c|}{Thomas' Attractor} \\
    \cline{2-13}
    & \multicolumn{1}{c|}{Time [s]} & \multicolumn{1}{c||}{Volume} & \multicolumn{1}{c|}{Time [s]} & \multicolumn{1}{c||}{Volume} & \multicolumn{1}{c|}{Time [s]} & \multicolumn{1}{c||}{Volume} & \multicolumn{1}{c|}{Time [s]} & \multicolumn{1}{c||}{Volume} & \multicolumn{1}{c|}{Time [s]} & \multicolumn{1}{c||}{Volume} & \multicolumn{1}{c|}{Time [s]} & \multicolumn{1}{c|}{Volume} \\
    \hline \hline
    0 &       & 59.051 &       & 33     &        & 0.0431 &       & 0.3717 &       & 0.5827 &       & 8000   \\
    \hline
    1 & 0.152 & 54.851 & 0.121 & 11.811 & 0.627  & 0.0431 & 6.55  & 0.1579 & 1.81  & 0.0028 & 1.404 & 371.22 \\
    \hline
    2 & 0.281 & 31.104 & 0.122 & 10.051 & 0.679  & 0.0278 & 12.6  & 0.0681 & 8.65  & 0.0014 & 6.661 & 303.4  \\
    \hline
    3 & 0.787 & 25.102 & 0.341 & 7.4795 & 1.319  & 0.0228 & 27.9  & 0.0373 & 271   & 0.0009 & 204.5 & 163.04 \\
    \hline
    4 & 2.752 & 17.738 & 1.051 & 7.1519 & 3.725  & 0.0205 & 65.9  & 0.0347 &       &        &       &        \\
    \hline
    5 & 10.37 & 16.961 & 3.739 & 6.3052 & 12.07  & 0.0199 & 174   & 0.0347 &       &        &       &        \\
    \hline
    6 & 41.62 & 16.771 & 14.18 & 6.1651 & 42.67  & 0.0198 &       &        &       &        &       &        \\
    \hline
    7   &        &       & 56.02 & 6.1651 & 160.99 & 0.0198 &       &        &     &        &       &        \\
    \hline
    \end{tabular}
  \label{tab: Computation}
\end{table*}

\subsection{Two-State systems} \label{SS: 2D}
\subsubsection{Van der Pol Oscillator} \label{SS: Van Der Pol Oscillator}
The Van der Pol (VdP) oscillator is modeled by the system of equations:
\begin{align}
    \dot{x}_1 &= x_2,\\
    \dot{x}_2 &= \mu (1 - x_1^2) x_2 - x_1, 
\end{align}
where $x_1,x_2 : [0, \infty) \rightarrow \R$ are the states and the parameter $\mu$ indicates the strength of damping. The VdP oscillator has a stable limit cycle, which is also its minimal RFIS.  Starting with a conservative initial polytopic invariant set with 6 vertices, $\SiCoSet[0]$, and $\o$ as the origin, the results of Algorithm~\ref{Alg: RFIS_Computation} with $\mu = 1$, $t_{\max} = 6$ are shown in Figs.~\ref{fig: VdP}~and~\ref{fig: VdP_Expand} for $\decay < 1$ and $\decay > 1$  respectively. Clearly, $\decay > 1$ corresponds to increasing the enclosed volume of the RFIS in every iteration. The more interesting application is when $\decay < 1$, since the algorithm attempts to approximate the minimal RFIS, $\mRPI$, and the result can be compared with the limit cycle of the VdP oscillator. The time taken for each iteration, along with the volume enclosed by $\SiCoSet[k] \ \forall \ 0 \leq k \leq 6$ is shown in Table~\ref{tab: Computation}. Note that $\SiCoSet[k]$ is an RFIS $\forall k$.

\begin{figure}[t]
    \centering
    \includegraphics[width = 0.49 \textwidth]{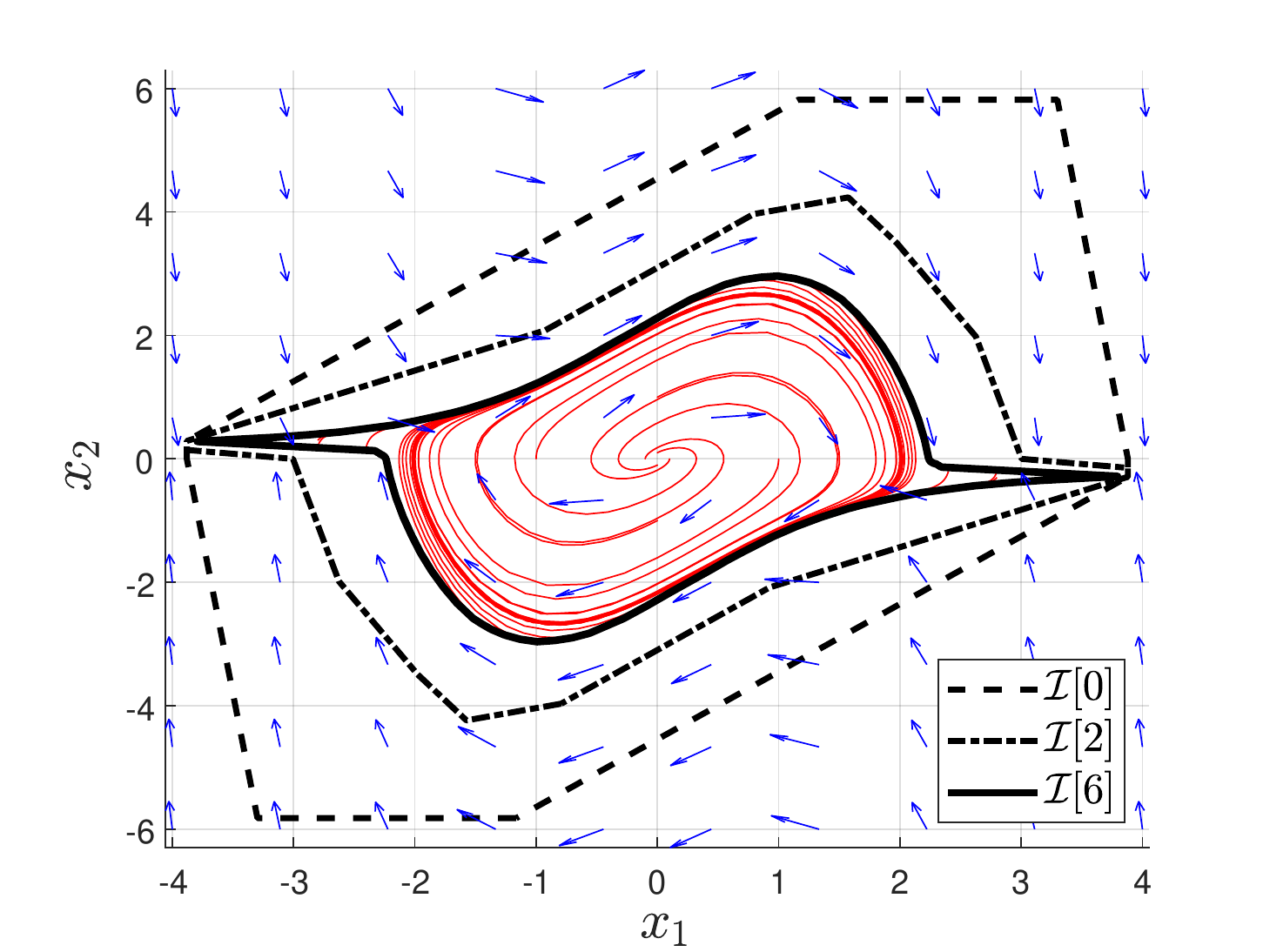}
    \caption{Van der Pol Oscillator, $\decay = 0.98$.}
    \label{fig: VdP}
\end{figure}

\begin{figure}[t]
    \centering
    \includegraphics[width = 0.49 \textwidth]{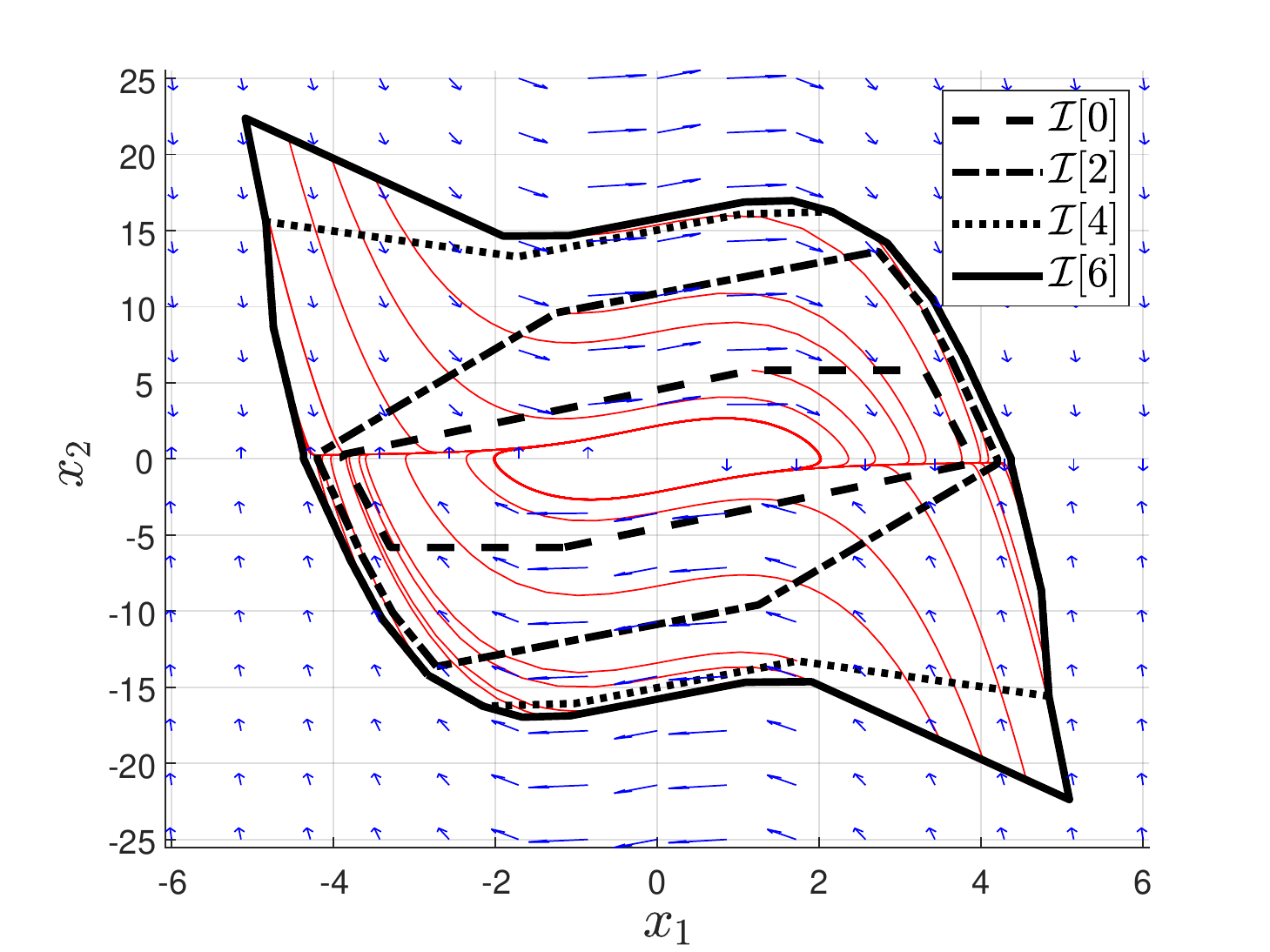}
    \caption{Van der Pol Oscillator, $\decay = 1.02$.}
    \label{fig: VdP_Expand}
\end{figure}

\subsubsection{Fitzhugh-Nagumo Neuron Model} \label{SS: Fitzhugh-Nagumo Neuron Model}
The Fitzhugh-Nagumo system models the activity of an excitable system such as a neuron. For comparison, we use the same choice of parameters as~\cite{Sassi_Automatica_2012}, for which the system is modeled as:
\begin{align}
    \dot{x}_1 &= x_1 - \frac{1}{3}x_1^3 - x_2 + \frac{7}{8},\\
    \dot{x}_2 &= 0.08 (x_1 + 0.7 - 0.8 x_2). 
\end{align}
We begin with an initial polytopic set (a convex quadrilateral) that is a conservative RFIS for the system. Setting $\decay = 0.95$, $t_{\max} = 7$ and $\o = [0 \ 1]^T$, it can be observed that Algorithm~\ref{Alg: RFIS_Computation} converges to a polytopic approximation of $\mRPI$, and there is no change in volume over the last iteration. Figure~\ref{fig: Fitzhugh} shows a subset of the sequence of RFISs obtained using the algorithm. These polytopic approximations are much tighter than previous works such as Fig.~1 in~\cite{Sassi_CDC_2014} and Fig.~4 in~\cite{Sassi_Automatica_2012}.

\begin{figure}[t]
    \centering
    \includegraphics[width = 0.49 \textwidth]{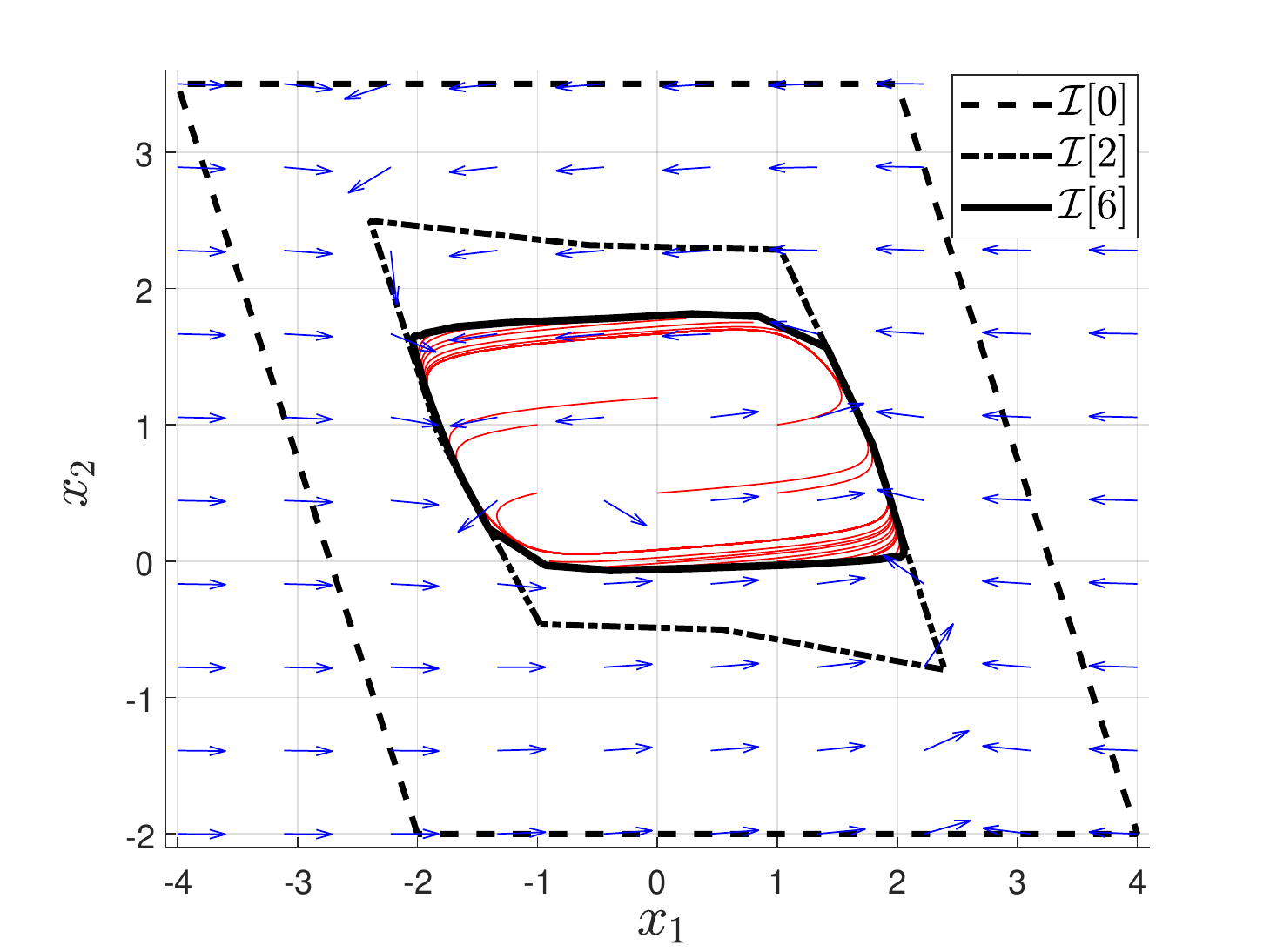}
    \caption{Fitzhugh-Nagumo Neuron Model}
    \label{fig: Fitzhugh}
\end{figure}

\subsubsection{Curve Tracking Problem} \label{SS: Curve Tracking Problem}
The curve tracking problem \cite{Malisoff_TAC_2011} is considered in the presence of disturbances, with $\noise = [\noise_1 \ \noise_2]^T : [0, \infty) \rightarrow \Noise$, with the noise space $\Noise = \{0\} \times [-1.5,1.5] \subset \R^2$. The dynamics are modeled as:
\begin{align}
    \dot{x}_1 &= -\sin(x_2) + \noise_1,\\
    \dot{x}_2 &= (x_1 - \rho)\cos(x_2) - \mu \sin(x_2) + \noise_2. 
\end{align}
The results in Fig.~\ref{fig: Curve Tracking} show the set valued map $F$ with the extreme rays obtained by using $\noise_2(t) = \pm 0.15$. The blue cones in the figure represent the set of all directions in which the trajectory may move due to the disturbance. Some trajectories starting at the boundary of the obtained RFIS are also displayed, for a disturbance function of the form $\noise_1(t) = 0$, $\noise_2(t) = 0.15 \sin(t)$, with $\rho = 1$ and $\mu = 6.42$. For comparison, the noise function and system parameters are chosen to be the same as~\cite{Mukhopadhyay_ACC_2014, Mukhopadhyay_Thesis_2014, Mukhopadhyay_AJC_2020}, where the minimal RFIS is computed using a different method. The agreement between the results confirms the validity of Algorithm~\ref{Alg: RFIS_Computation}, which converges in $6$ iterations as evident from Table~\ref{tab: Computation}. 

\begin{figure}[t]
    \centering
    \includegraphics[width = 0.49 \textwidth]{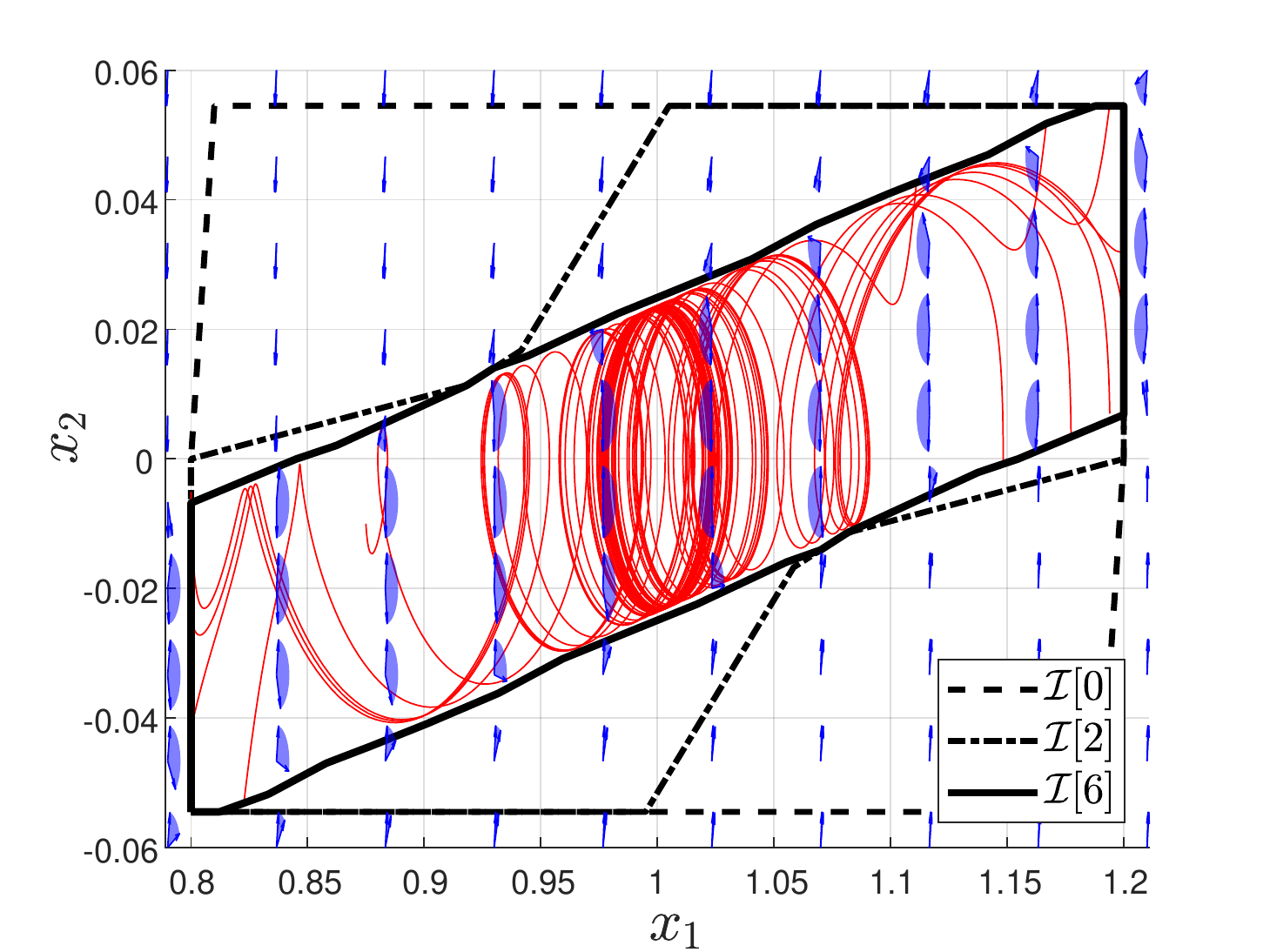}
    \caption{Curve Tracking Problem}
    \label{fig: Curve Tracking}
\end{figure}

\subsubsection{Reversed Van der Pol Oscillator} \label{SS: Reversed Van Der Pol Oscillator}
The reversed Van der Pol oscillator is modeled by:
\begin{align}
 \dot{x}_1 &= -x_2 + \noise_1,\\
 \dot{x}_2 &= x_1 - x_2 + x_2^3 + \noise_2.
\end{align}
We consider a polytopic noise space $\Noise = [-0.03,0.03] \times [-0.03,0.03]$. Figure~\ref{fig: RVdP} shows the set valued map $F$ at different points in the state space, with the blue cones representing the set of directions that the trajectory may move depending on the disturbance. The values of $\dot{x}$ obtained by using the vertices of $\Noise$, $(\pm 0.03, \pm 0.03)$ are shown as vectors, in addition to three sets of system trajectories at a fixed set of initial points on the boundary of the RFIS obtained through the system. The first two sets of system trajectories use constant disturbances $[\noise_1(t) \ \noise_2(t)]^T = [0.03 \ -0.03]^T$ and $[-0.03 \ 0.03]^T$ respectively, while the third set uses $\noise_1(t) = 0.01\sin(2t)+0.005\sin(\pi t)+0.015\sin(6.53t)$, and $\noise_2(t) = - 0.01\cos(0.2 t)+0.02\sin(5\pi t)$. Using $\o = [0 \ 0]^T$ and $\decay = 0.99$, Algorithm~\ref{Alg: RFIS_Computation} converges in $4$ iterations as evident from Table~\ref{tab: Computation}.

\begin{figure}[t]
    \centering
    \includegraphics[width = 0.49 \textwidth]{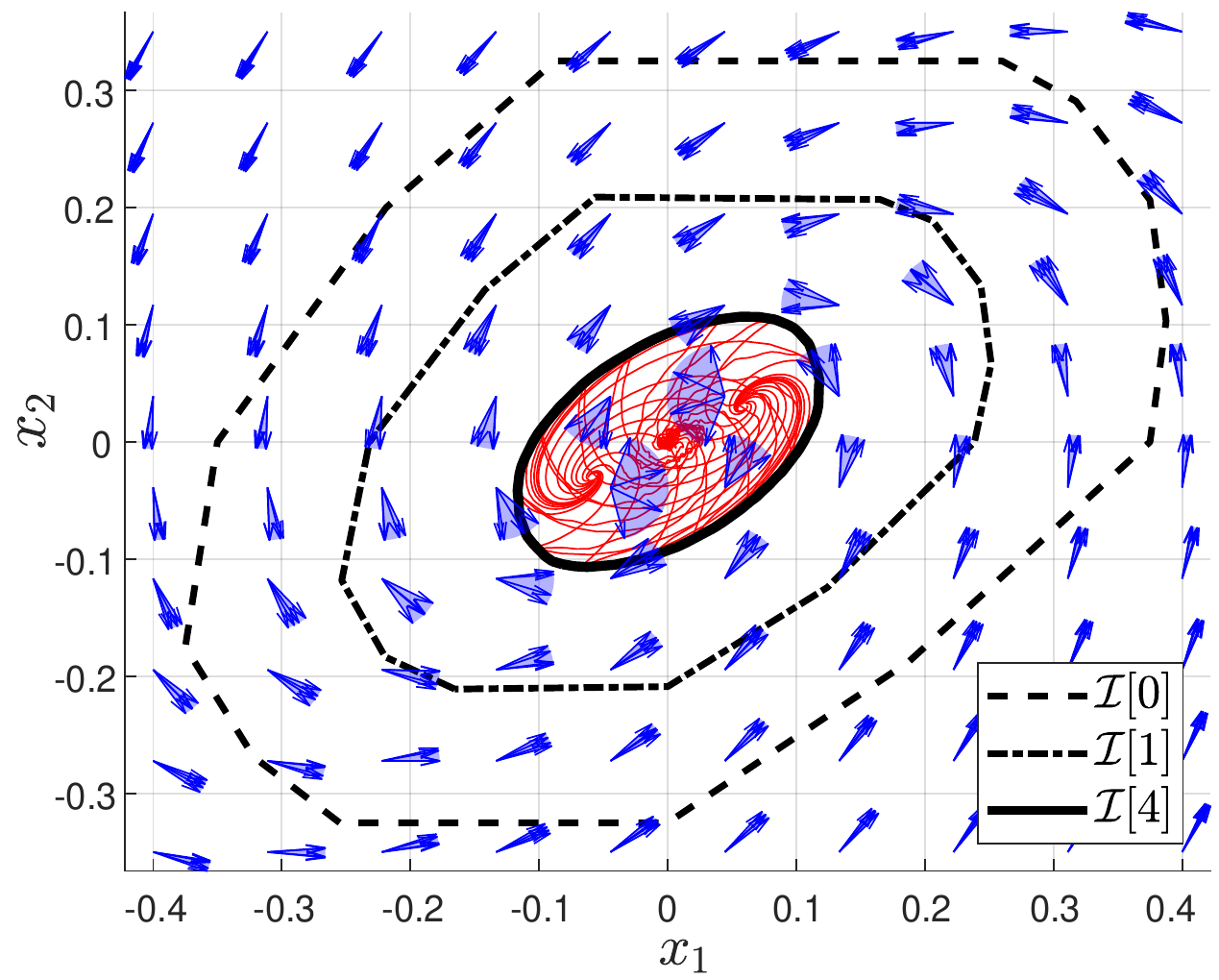}
    \caption{Reversed Van der Pol Oscillator}
    \label{fig: RVdP}
\end{figure}


\begin{figure*}[t]
	\centering
	\subfigure[$\SiCoSet \lbrack 0 \rbrack$.]
	{
		\includegraphics[width=0.315 \textwidth]{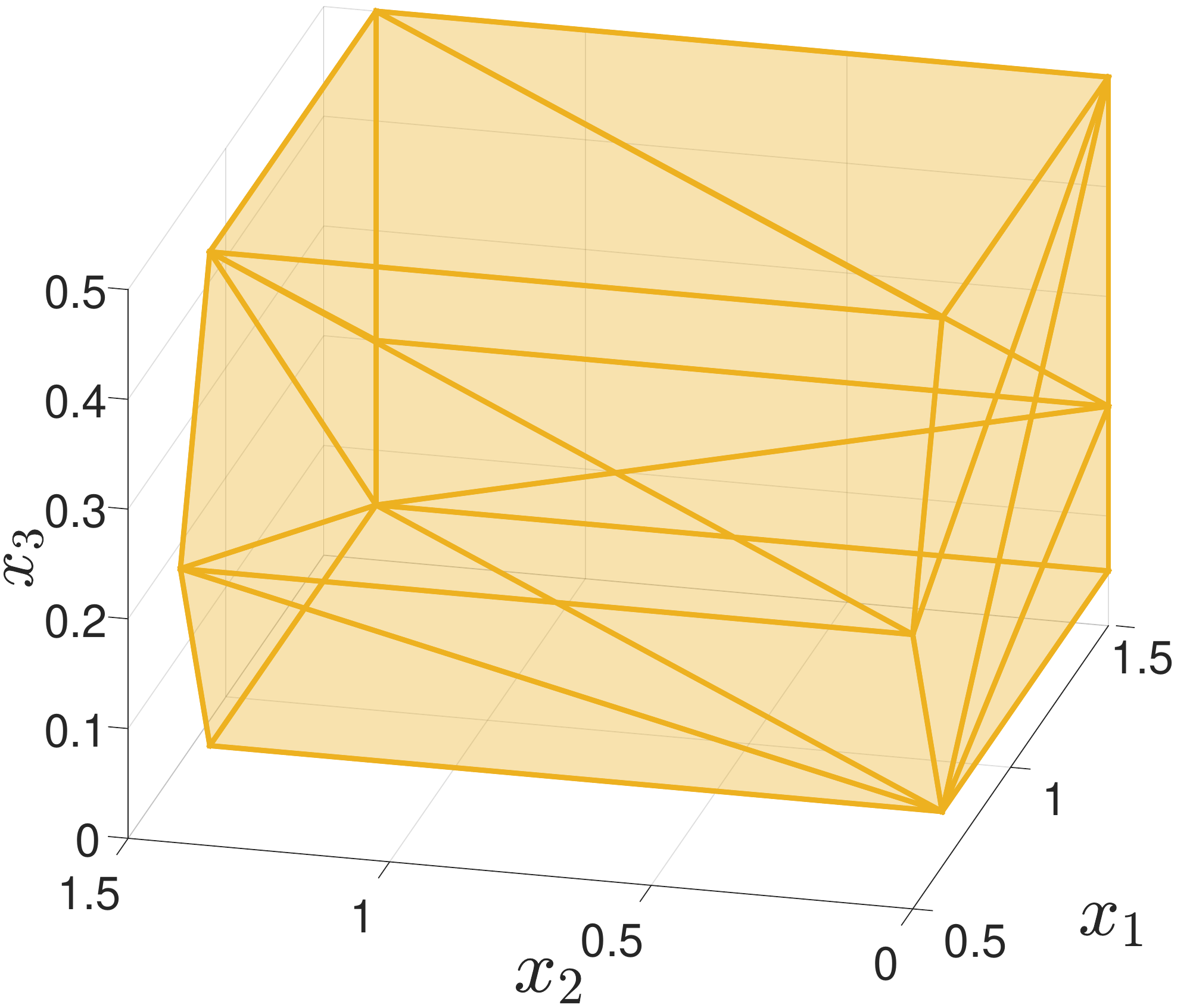}
		\label{fig: Phyto_A}
	}
	\hfil
	\subfigure[$\SiCoSet \lbrack 3 \rbrack$.]
	{
        \includegraphics[width=0.315 \textwidth]{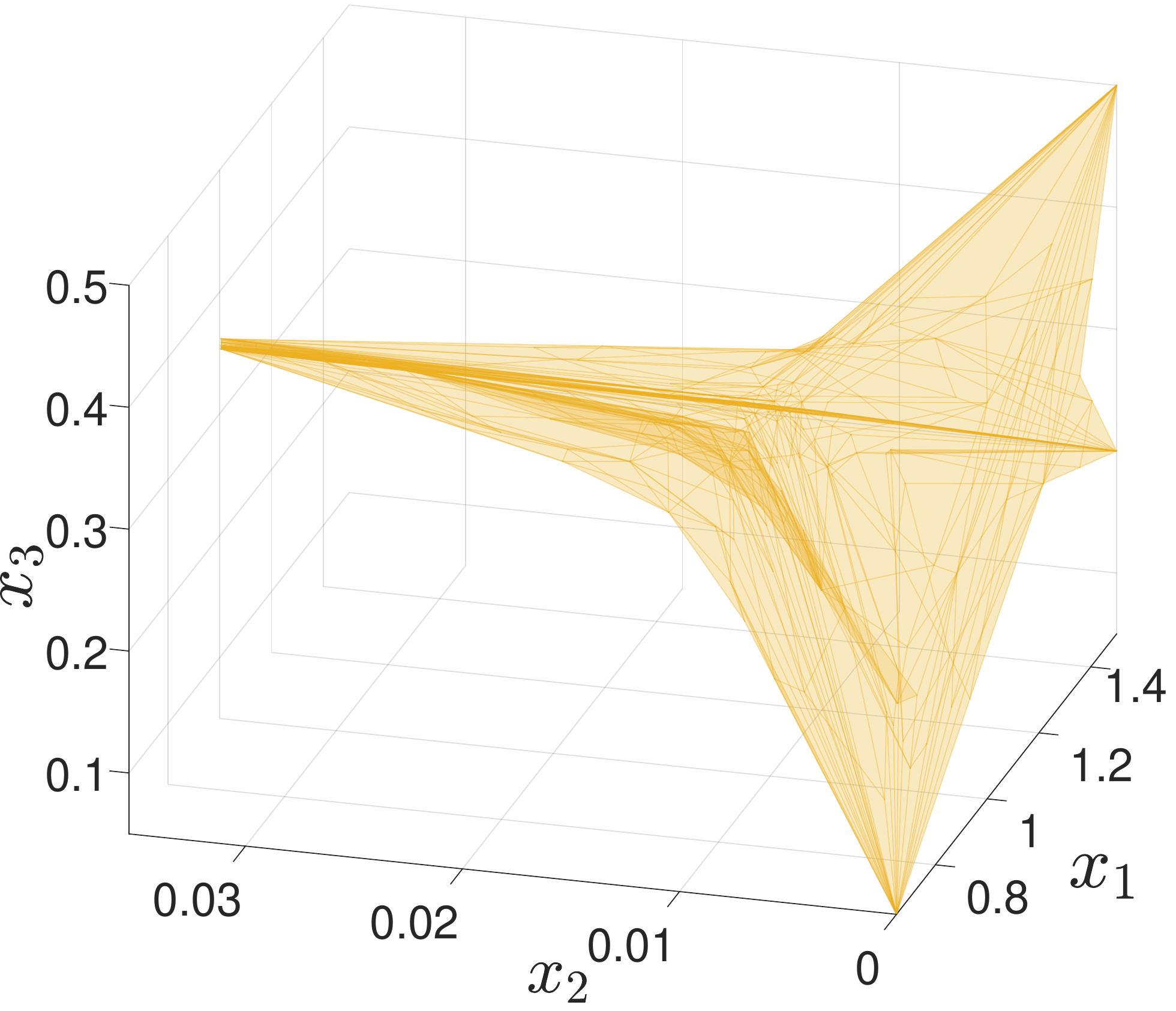}
		\label{fig: Phyto_B}
	}
	\hfil
	\subfigure[Trajectories.]
	{
        \includegraphics[width=0.315 \textwidth]{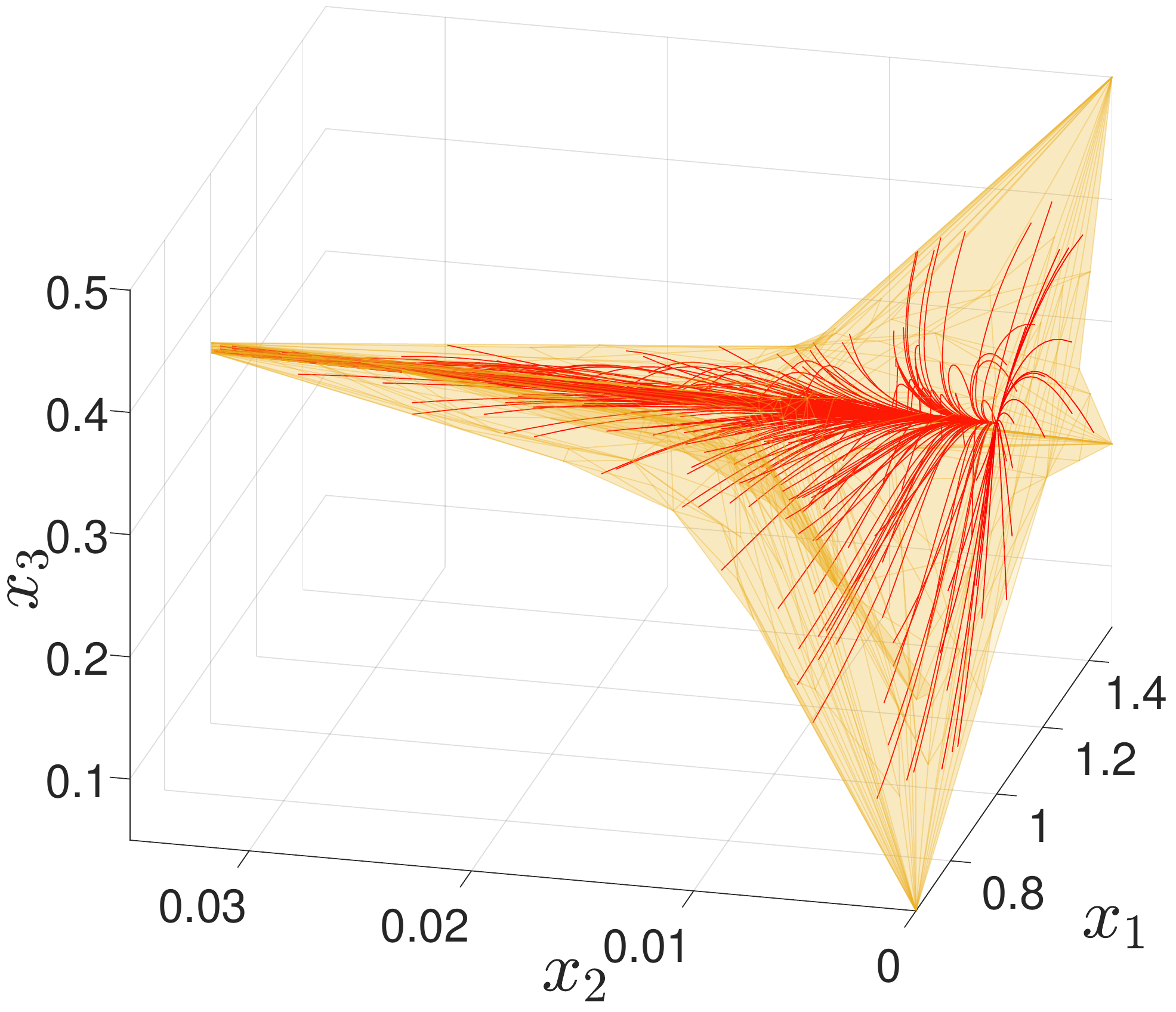}
		\label{fig: Phyto_Traj}
	}
	\caption{Phytoplankton Growth Model.}
	\label{fig: Phytoplankton}
\end{figure*}
\subsection{Three-state systems} \label{SS: 3D}
\subsubsection{Phytoplankton Growth Model} \label{SS: Phytoplankton Growth Model}
Phytoplankton growth is modeled by the dynamical system: 
\begin{align}
    \dot{x}_1 &= 1 - x_1 - \frac{1}{4}x_1 x_2, \\
    \dot{x}_2 &= (2x_3 - 1) x_2,  \\
    \dot{x}_3 &= \frac{x_1}{4} - 2 x_3^2. 
\end{align}
The authors in~\cite{Sassi_Automatica_2012} use optimization methods to obtain a conservative polytopic FIS for this model (see Fig.~5 in~\cite{Sassi_Automatica_2012}). Starting with an approximate recreation of this polytopic RFIS as $\SiCoSet[0]$, and setting $\o = [0.9969 \ 0.01 \ 0.3571]^T$, $\decay = 0.9$, we obtain a much tighter approximation of $\mRPI$ for the system. In just $3$ iterations of Algorithm~\ref{Alg: RFIS_Computation}, the polytope is deformed into a new polytopic RFIS, $\SiCoSet[3]$ that is a much smaller RFIS than previous work, and encloses just $0.15\%$ of the volume that $\SiCoSet[0]$ did. Figure~\ref{fig: Phytoplankton} shows the invariant sets $\SiCoSet[0]$ and $\SiCoSet[3]$, and Fig.~\ref{fig: Phyto_Traj} shows some system trajectories starting at the surface of the obtained RFIS. As expected, the trajectories move inward and never escape the RFIS.

\subsubsection{Thomas' Cyclically Symmetric Attractor} \label{SS: Thomas' Cyclically Symmetric Attractor}

\begin{figure*}[t]
	\centering
	\subfigure[$\SiCoSet \lbrack 0 \rbrack$.]
	{
		\includegraphics[width=0.315 \textwidth]{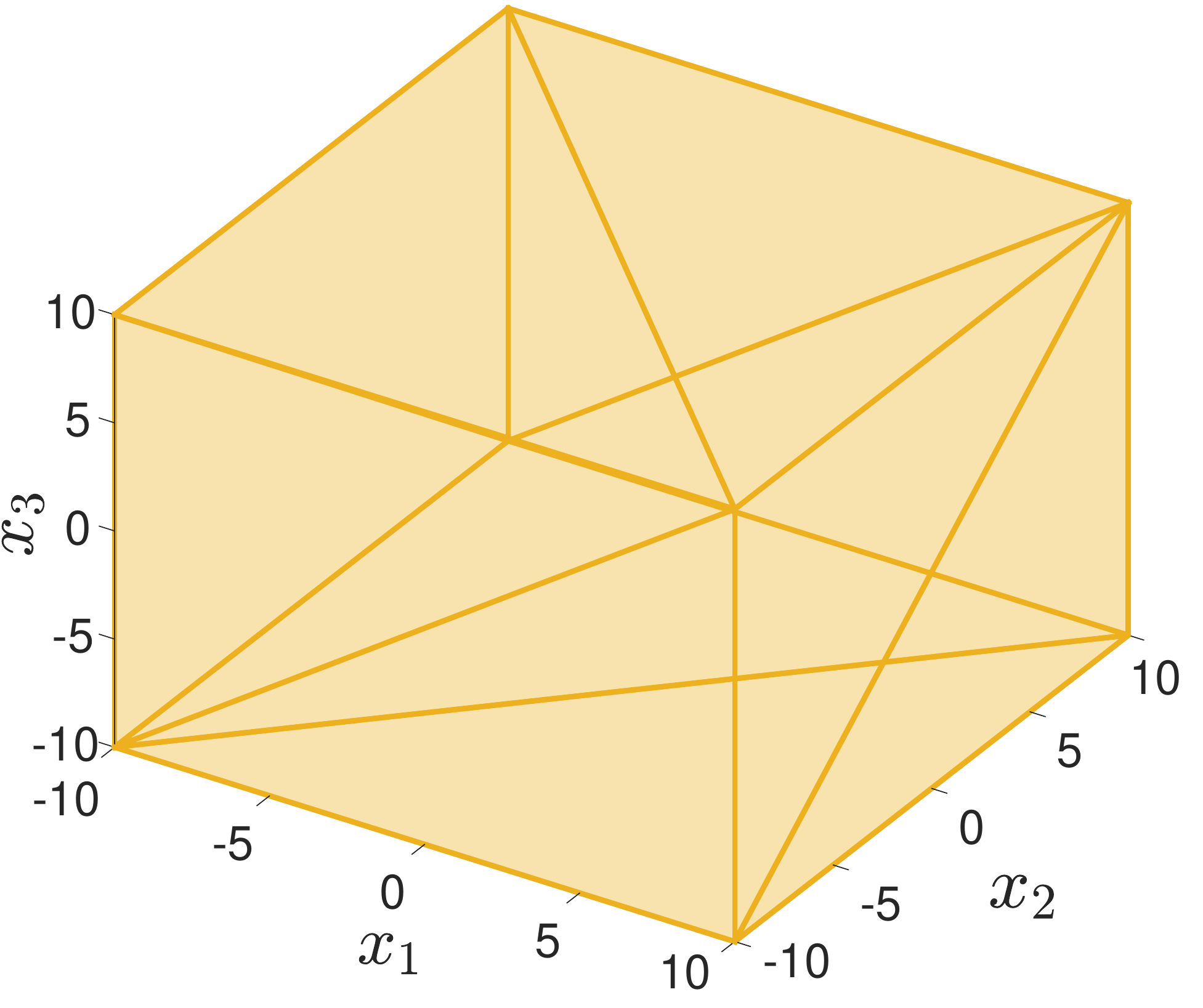}
		\label{fig: Thomas_A}
	}
	\hfil
	\subfigure[$\SiCoSet \lbrack 3 \rbrack$.]
	{
        \includegraphics[width=0.315 \textwidth]{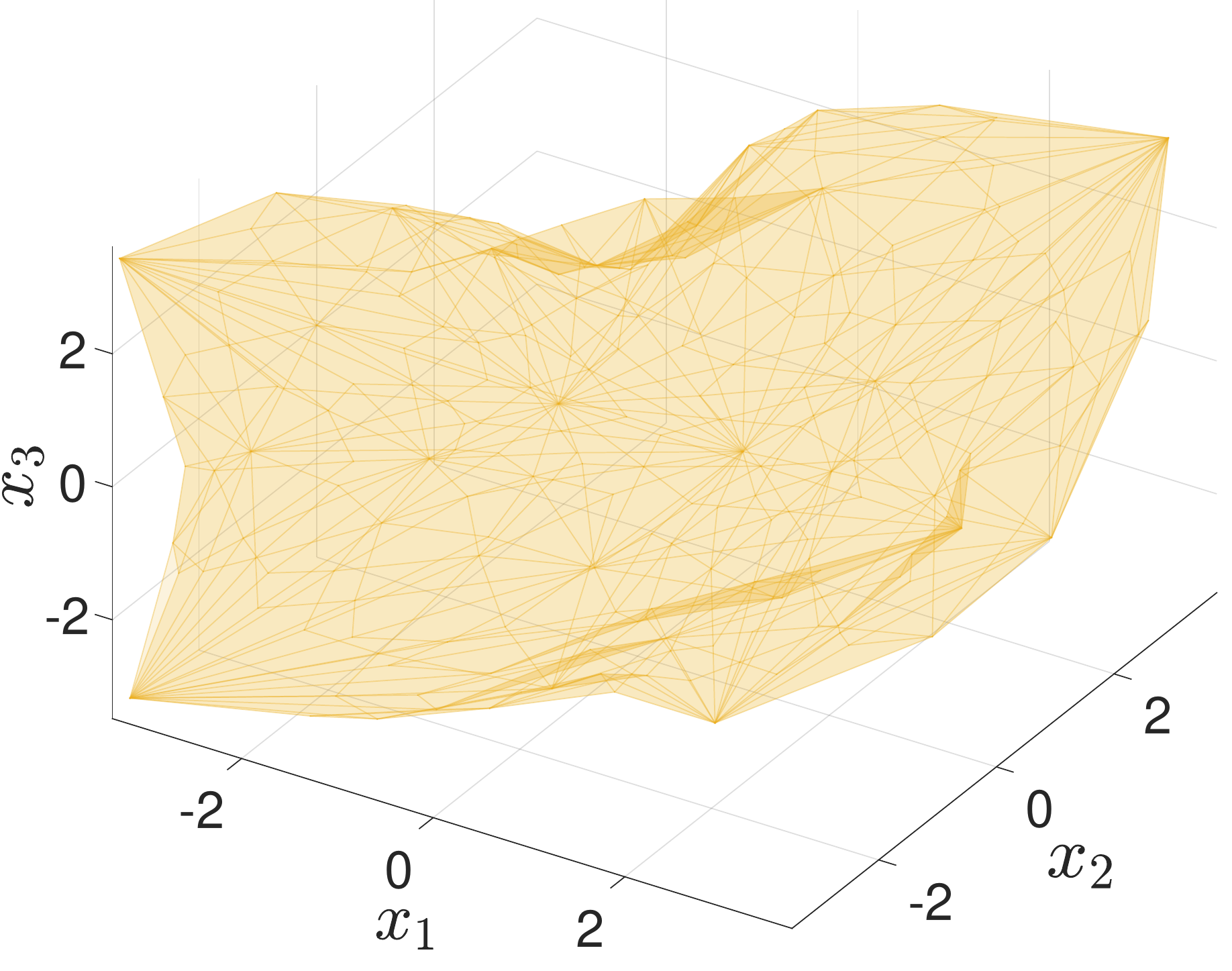}
		\label{fig: Thomas_B}
	}
	\hfil
	\subfigure[Trajectories.]
	{
        \includegraphics[width=0.315 \textwidth]{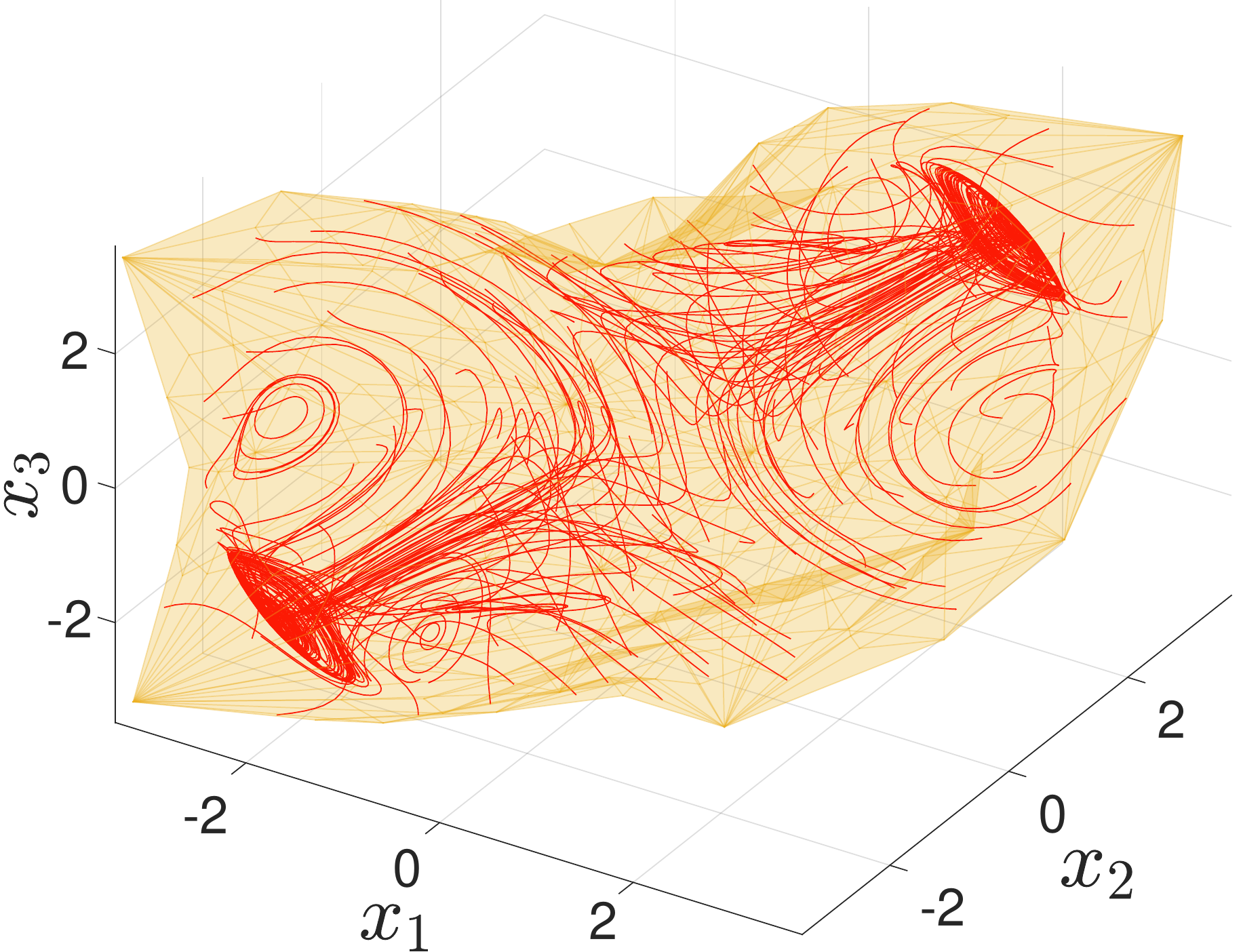}
		\label{fig: Thomas_Traj}
	}
	\caption{Thomas' Cyclically Symmetric Attractor, $b = 0.3$.}
	\label{fig: Thomas Attractor}
\end{figure*}

\begin{figure*}[t]
	\centering
	\subfigure[$\SiCoSet \lbrack 0 \rbrack$.]
	{
		\includegraphics[width=0.315 \textwidth]{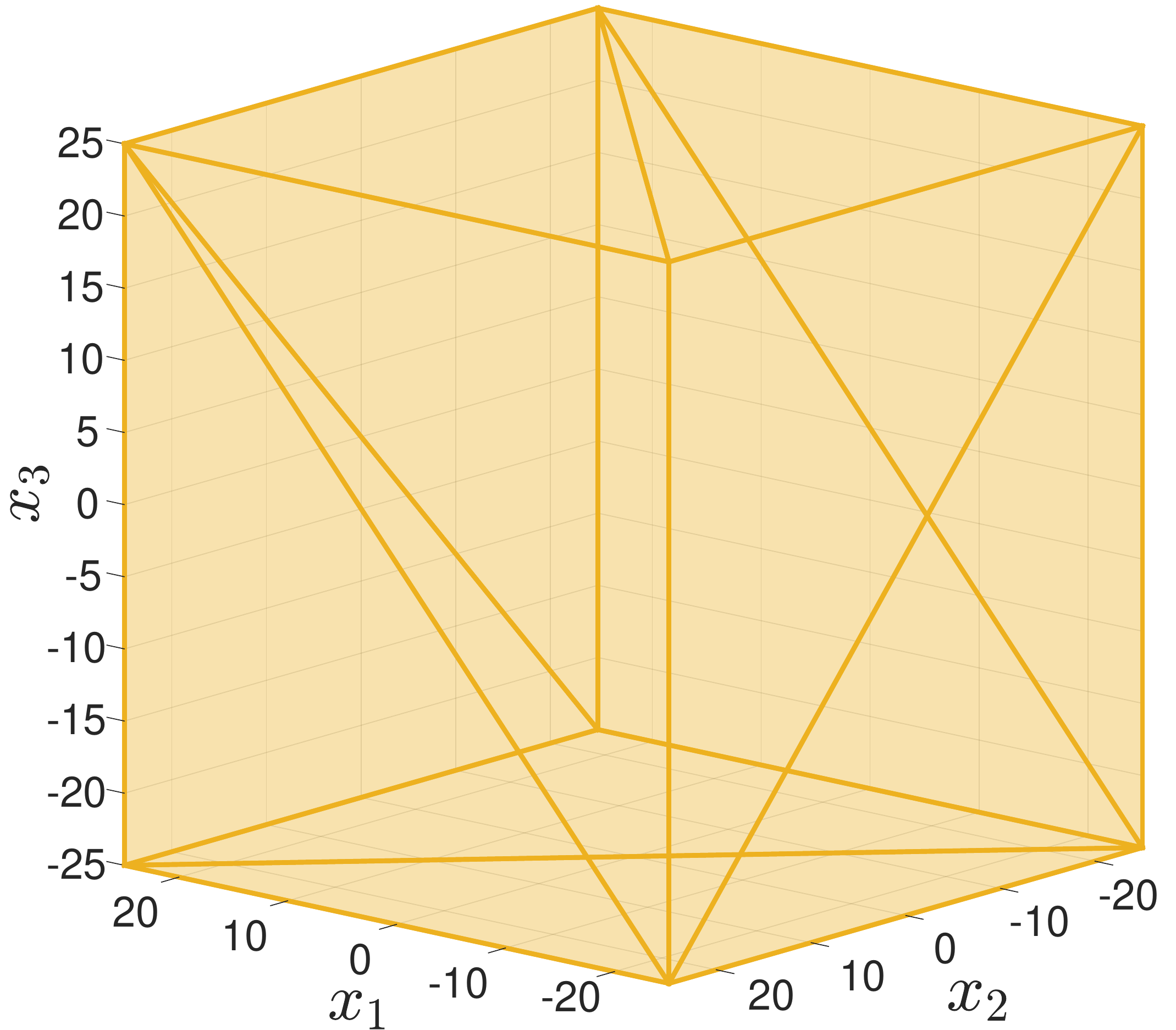}
		\label{fig: Thomas_Chaotic_A}
	}
	\hfil
	\subfigure[$\SiCoSet \lbrack 4 \rbrack$.]
	{
        \includegraphics[width=0.315 \textwidth]{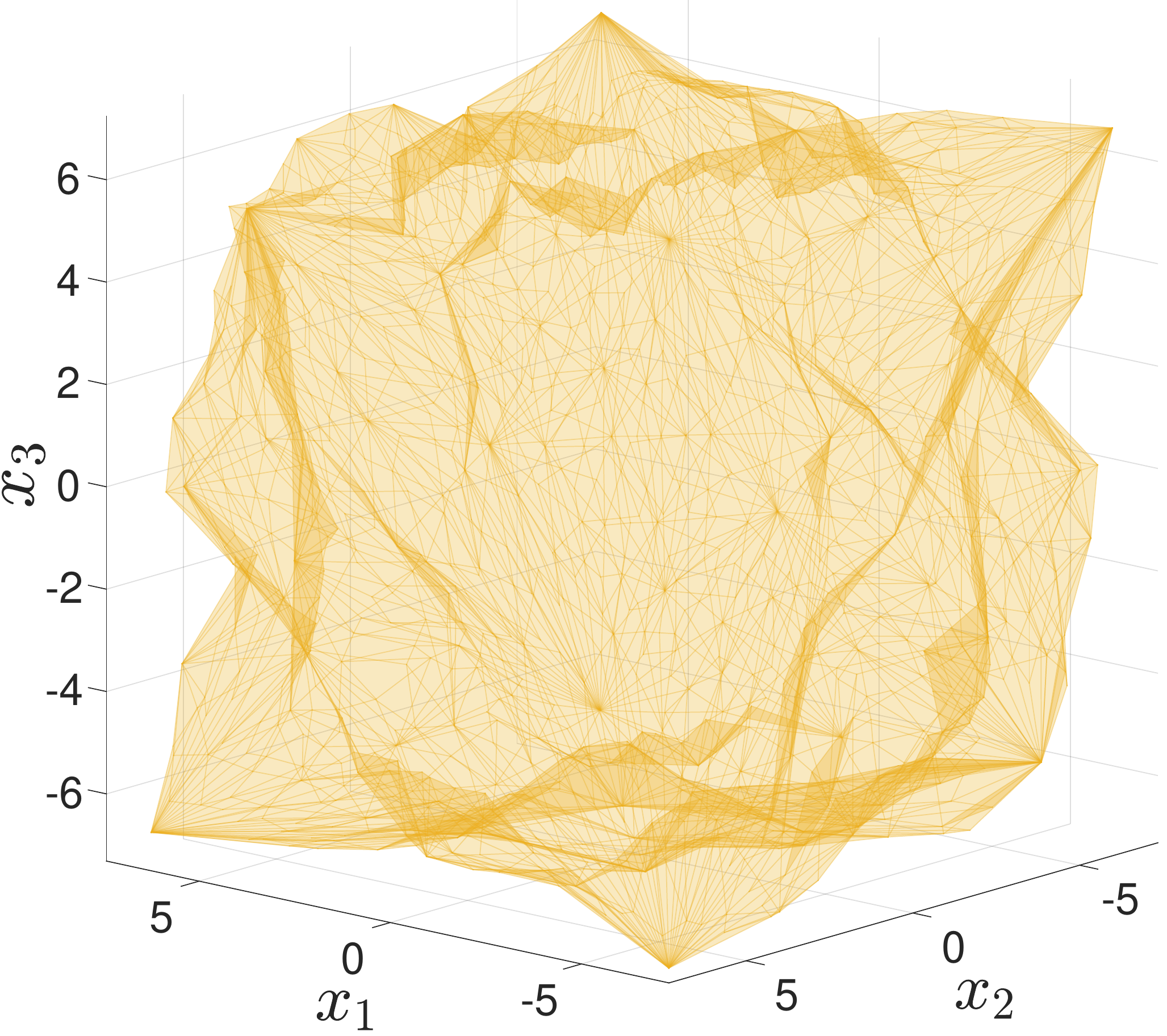}
		\label{fig: Thomas_Chaotic_B}
	}
	\hfil
	\subfigure[Trajectories.]
	{
        \includegraphics[width=0.315 \textwidth]{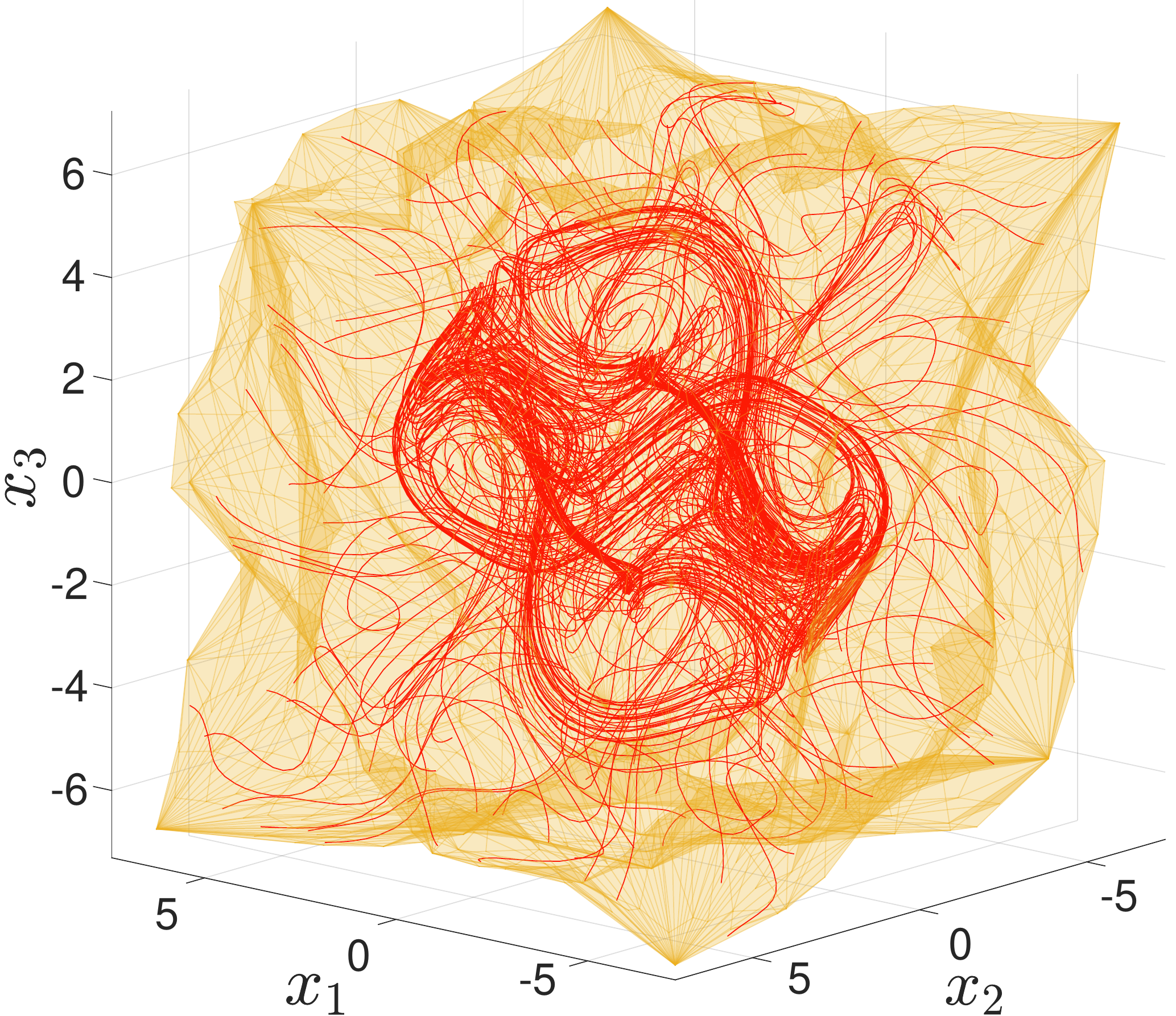}
		\label{fig: Thomas_Chaotic_Traj}
	}
	\caption{Thomas' Cyclically Symmetric Attractor, $b = 0.14$.}
	\label{fig: Thomas Attractor Chaotic}
\end{figure*}

Thomas' Attractor is described by the dynamical system:
\begin{align}
    \dot{x}_1 &= \sin(x_2) - b x_1, \\
    \dot{x}_2 &= \sin(x_3) - b x_2,  \\
    \dot{x}_3 &= \sin(x_1) - b x_3.
\end{align}

Here, $b$ is a constant and acts as a bifurcation parameter. We consider $b = 0.3$, for which the system has two stable attractive limit cycles. Setting $\SiCoSet[0]$ as a triangulation of the boundary of the cube with vertices $(\pm 10, \pm 10, \pm 10)$, the algorithm converges to an FIS in three iterations. The sets $\SiCoSet[0]$ and $\SiCoSet[3]$ are shown in Fig.~\ref{fig: Thomas Attractor}, with Fig.~\ref{fig: Thomas_Traj} showing some system trajectories.  Finally, when $b$ is reduced to below $\approx 0.202816$, the  limit cycle undergoes a period doubling cascade and becomes chaotic. We apply Algorithm~\ref{Alg: RFIS_Computation} on the attractor when $b = 0.14$, and summarize the results in Fig.~\ref{fig: Thomas Attractor Chaotic}, where $\SiCoSet[0]$, $\SiCoSet[4]$ and some system trajectories are shown.

\section{Future Work and Conclusion} \label{S: Conclusion}

\subsection{Future Work}
Our proposed approach opens new avenues for future research. In particular, a comprehensive study of vertex maps and the effect of their induced homeomorphisms on faster convergence of the algorithm would be interesting. For instance, vertex maps that map multiple vertices to different points, and adaptive adjustment of the parameters such as $\decay$ and $\o$, based on the behavior of the vector field at the boundary of the polytope is expected to speed up the algorithm. Further, if the homeomorphism is represented as a sequence of vertex maps as in our work, the algorithm may be further optimized by studying the ideal sequence in which to apply the maps, since homeomorphisms in general do not commute. The effect of simplex subdivisioning and the interplay between perturbing a vertex (through simplicial maps) and creating a vertex (through simplex subdivisioning) is another way to optimize the algorithm. Other measures of how close a simplex is to being invariant, besides $\NTP(\Si)$, may also help aid this process. Thus, our algorithm can in fact be extended to a family of computational algorithms that search for an RFIS using (possibly variants of) the~\BdryC. All such algorithms may be used on any Lipschitz continuous nonlinear dynamical system with bounded additive disturbances. Similar approaches may also be extended to compute control invariant sets for nonlinear systems.

\subsection{Conclusion}
This paper developed an algorithm to generate a sequence of Robust Forward Invariant Sets of different sizes. First, we showed that if a nonlinear system's dynamics are Lipschitz continuous, then testing for invariance of a given (polytopic) set can be easily done by checking for a~\BdryC~on finitely many points, amounting to matrix-vector multiplication. Based on this, we developed an algorithm to deform a polytope into an RFIS, through a sequence of homeomorphisms induced by vertex maps on a simplicial complex that triangulates the boundary of the polytope. The geometric nature of the approach allows the algorithm to be used on any Lipschitz continuous nonlinear dynamical system in the presence of additive bounded disturbances. Application of the algorithm to a variety of nonlinear systems showed fast convergence and resulted in a sequence of RFIS of different sizes, that converged to approximate the minimal RFIS of the system.

\appendix

\subsection{Proof of Lemma~\ref{Lem: Lattice_Is_A_Subset_Of_Simplex}}
    We show that the linear transformation $T(\M)$ from the standard simplex $\Lambda$ to the simplex of interest $\Si$ ensures that the reference test points on $\Lambda$ map to points on $\Si$.
    Consider $T(\M) = \M L_{\Si}$, and let $x \in \TP_m$, so that $x$ is a row of $T(\M)$ and thus a row combination of $\M$ and $L_{\Si}$. But the rows of $L_{\Si}$ are just vertex coordinates of $\Si$, and each row of $\M$ is a lattice test point on $\Lambda$ by construction. Further, $y \in \Lambda \implies \sum_{i = 1}^{n}y_i = 1$. The matrix multiplication can be carried out as follows.
    \begin{align}
        M_m L_{\Si} &= \begin{bmatrix}
        M_{1,1} & M_{1,2} & \cdots & M_{1,n} \\
        M_{2,1} & M_{2,2} & \cdots & M_{2,n} \\
        \vdots & \vdots & \ddots & \vdots \\
        M_{\Lm,1} & M_{\Lm,2} & \cdots & M_{\Lm,n}
        \end{bmatrix}
        \begin{bmatrix}
            v_0^T \\ v_1^T \\ \vdots \\ v_{n-1}^T
        \end{bmatrix} \\
        &= \begin{bmatrix}
            \sum_{j= 1}^{n} M_{1,j}v_{j-1}^T \\
            \sum_{j= 1}^{n} M_{2,j}v_{j-1}^T \\
            \vdots \\
            \sum_{j= 1}^{n} M_{\Lm,j}v_{j-1}^T
        \end{bmatrix}
    \end{align}
    Since $x$ is a row of this matrix, and since $0 \leq M_{i,j} \leq 1 \ \forall i,j$, and $\sum_{j = 1}^{n}M_{i,j} = 1 \ \forall i$, it follows that $x$ is a convex combination of the vertices of $\Si$, and hence, $x \in \Si$. Since this is true for all $x \in \TP_m$, we conclude that $\TP_m \subset \Si$.

\subsection{Proof of Thereom~\ref{Thm: Bdry_Condition_Discretized}}
We now prove Theorem~\ref{Thm: Bdry_Condition_Discretized}, which essentially says that testing for the~\BdryC~at a subset of test points is sufficient to check the invariance of the simplex. To establish this, we introduce the notion of adjacent lattice points.

\begin{customdef}{B.1}[Adjacent Lattice Point]
Let $w_1 \in \TP_m \subset \Si$ be a lattice test point. Another lattice point $w_2 \in \TP_m$ is said to be adjacent to $w_1$ if $w_2 = w_1 \pm \ss_m (v_i - v_j)$, for any $v_i$, $v_j \in \V{\Si}$, where $\ss_m$ is defined in~\eqref{Eq: ss_m}. If $w$ is a lattice point, then $\A(w)$ denotes the set consisting of $w$ and all its adjacent points.
\end{customdef}

Adjacent points are easily visualized from Fig.~\ref{fig: Adj_Test_Points}, where the points in $\A(w_1)$ are colored blue. The convex hull of $\A(w_1)$ is shaded and denoted by $\Conv(\A(w_1))$. 
Next, we show that these test points are appropriately spaced apart from each other. Specifically, there is a bound on the maximum distance a test point can be from all other test points. 

\begin{figure}[t]
	\centering
	\subfigure[Adjacent Test Points.]
	{
		\includegraphics[width=0.23 \textwidth]{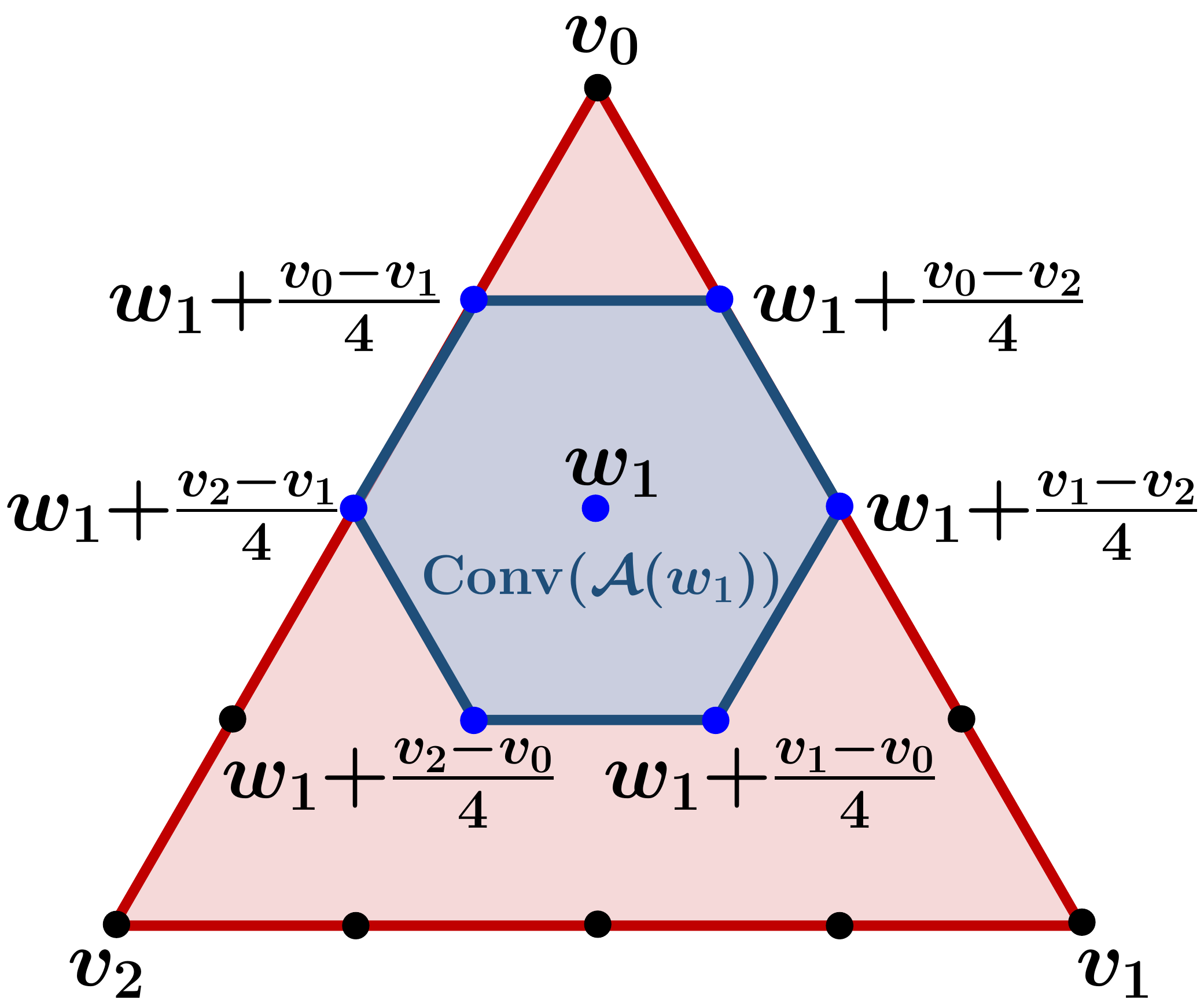}
		\label{fig: Adj_Test_Points}
	}
	\hfil
	\subfigure[Balls on Simplex.]
	{
        \includegraphics[width=0.22 \textwidth]{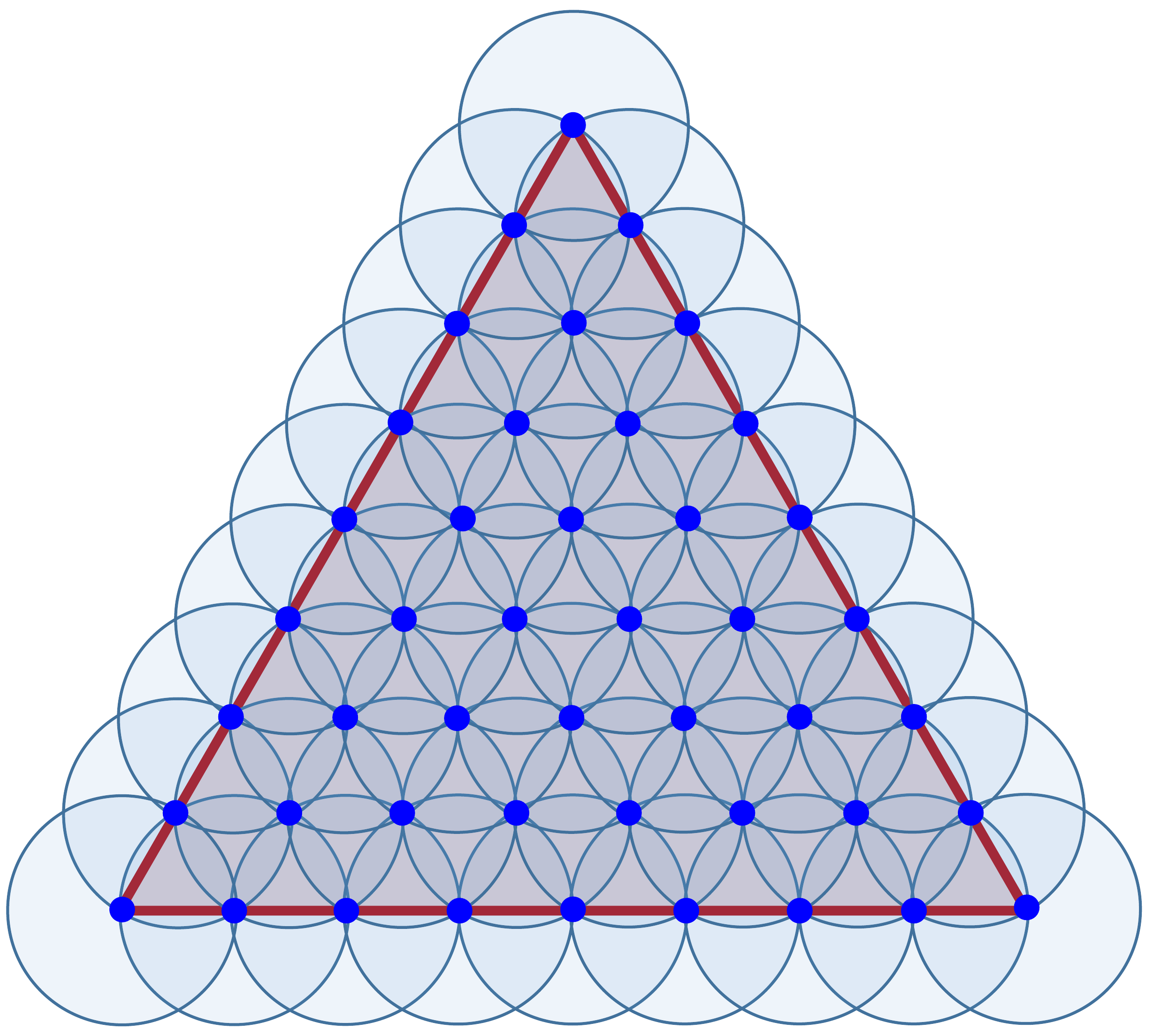}
		\label{fig: Balls_On_Simplex}
	}\\
	\caption{Simplex Proof Illustration.}
	\label{fig: Simplex_Proof_Concepts}
\end{figure}

\begin{customlem}{B.2} \label{Lem: Lattice_Is_Complete}
Let $r = \max_{i,j}{\| v_i - v_j \|}$, where $v_i, v_j \in \V{\Si}$. Further, let $w_i \in \TP_m$ be any lattice test point. Then, $\ w_j \in \A(w_i) \implies \| w_i - w_j \| \leq r \ss_m$.
\end{customlem} 

\begin{proof}
    Let $w_i\in \TP_m$ be a lattice test point on $\Si$. Let $w_j \in \A(w_i)$ so that $w_j = w_i \pm \ss_m (v_i - v_j)$ by definition. Further, it follows for all $w
    _j$ that
    \begin{align}
        \|w_i - w_j \| & = \ss_m  \|v_i - v_j \| \\
                      & \leq \ss_m \max_{i,j} \|v_i - v_j \| \\
                      & \leq  r\ss_m  ,
    \end{align}
    as desired. 
\end{proof}
Lemma~\ref{Lem: Lattice_Is_Complete} shows that a ball of radius $r \ss _m$ around a given lattice test point always contains all adjacent lattice test points. This property of the test points allows us to propose the invariance test in Theorem~\ref{Thm: Bdry_Condition_Discretized}, repeated here as Theorem~\ref{Thm: Bdry_Condition_Discretized_Apx} for convenience.

\begin{customthm}{B.3} \label{Thm: Bdry_Condition_Discretized_Apx}
    Let $\Si$ be a given simplex with $\V{\Si} = \{v_0, \cdots v_{n-1}\}$, $r = \max_{i,j}{\| v_i - v_j \|}$, and $\TP_m$ be the set of test points on $\Si$ for a given $\ss_m$. Suppose that
    \begin{align} \label{Eq: Bdry_Condition_Discretized_Apx}
        \langle F(x), \Normal{\Si} \rangle \leq - r \ss_m \ell, \quad \forall x \in \TP_m
    \end{align}
    where $\ell$ is the Lipschitz constant for the system dynamics given by $F$. Then, $\Si$ is an invariant simplex.
\end{customthm}

\begin{proof}
    Let $\B_k(z)$ represent the ball of radius $k$ centered at $z$. We prove that if~\eqref{Eq: Bdry_Condition_Discretized_Apx} is satisfied for all lattice test points, then every point on $\Si$ satisfies the \BdryC.
    
    From Lemma~\ref{Lem: Lattice_Is_Complete}, we know that $\B_{r\ss_m}(x)$ contains all \emph{adjacent} lattice points. Thus, $\A(x) \subset \B_{r\ss_m}(x)$, and since $\B_{r \ss_m}(x)$ is a convex set, we also have $\Conv (\A(x)) \subset \B_{r \ss_m}(x)$.
    
    Taking the union over all lattice test points, we have
    \begin{align}\label{Eq: Union_of_Hulls}
        & \bigcup_{x \in \TP_m} \Conv(\A(x)) \subset \bigcup_{x \in \TP_m} \B_{r\ss_m}(x) 
    \end{align}
   Since the vertices of the simplex are always lattice test points by construction, it follows that the left hand side of~\eqref{Eq: Union_of_Hulls} is the simplex, $\Si$ itself. Further, from Theorem~\ref{Thm: Lipschitz_Means_Finite}, we know that the~\BdryC~is satisfied at all points in each of the balls, and therefore on all subsets of their union. Since $\Si$ is one such subset, we conclude that $\Si$ is an invariant simplex.
\end{proof}
 A visualization of this proof is provided in Fig.~\ref{fig: Balls_On_Simplex}. The right hand side of~\eqref{Eq: Union_of_Hulls} is the union of the balls on the simplex which contains the simplex itself. 


\bibliographystyle{IEEEtran.bst}
\bibliography{references}


\end{document}